\documentclass[11pt]{article}

\usepackage{amssymb,amsmath,amsfonts,eurosym,geometry,ulem,graphicx,caption,color,setspace,sectsty,comment,footmisc,caption,natbib,pdflscape,array,bbm,bm,appendix,mathtools}
\usepackage[hidelinks]{hyperref}
\usepackage{multirow}
\usepackage{graphicx}
\usepackage{caption}
\usepackage{subcaption}

\usepackage{float} 
\usepackage{array}
\usepackage[flushleft]{threeparttable}
\usepackage{longtable}
\usepackage{booktabs}

\newcolumntype{L}{
>{\centering\arraybackslash}m{3cm}}
\include{setup}

\newtheorem{assumption}{Assumption}

\newtheorem{proposition}{Proposition}
\newenvironment{proof}[1][Proof]{\noindent\textbf{#1.} }{\ \rule{0.5em}{0.5em}}

\newcolumntype{L}[1]{>{\raggedright\let\newline\\arraybackslash\hspace{0pt}}m{#1}}
\newcolumntype{C}[1]{>{\centering\let\newline\\arraybackslash\hspace{0pt}}m{#1}}
\newcolumntype{R}[1]{>{\raggedleft\let\newline\\arraybackslash\hspace{0pt}}m{#1}}

\usepackage{tikz}
\usepackage{expl3}
\ExplSyntaxOn
\cs_set_eq:NN \fpeval \fp_eval:n
\ExplSyntaxOff
\usepackage{times}

\usepackage{enumitem}

\setdisplayskipstretch{1}
\AtBeginDocument{%
\setlength{\abovedisplayskip}{2pt plus 1pt minus 1pt}
\setlength{\abovedisplayshortskip}{0pt plus 1pt}%
\setlength{\belowdisplayshortskip}{2pt plus 1pt minus 1pt}%
}

\begin{document}

\newcommand*{\thisdraft}{August 2026} 
\newcommand*{\firstdraft}{}  
\begin{titlepage}

\title{The Geography of Research: The Trade-Off \\ \vspace{-0.5em} Between Knowledge Production and Access}
	
\author{Sitian Liu\thanks{Queen's University. Email: sitian.liu@queensu.ca.}
\and
Yichen Su\thanks{Southern Methodist University. Email: yichensu@outlook.com. We thank seminar participants at Hitotsubashi University and Renmin University of China for their valuable comments.}}

\date{\thisdraft \\ \firstdraft}

\maketitle
\vspace{-1.5em}
\begin{abstract}
\doublespacing
Research activity generates highly localized positive spillovers, yet in the U.S. it has become increasingly spatially misaligned with population and economic activity as people moved away from legacy cities where major research institutions remain anchored. Reallocating researchers toward population centers could broaden local access to knowledge but may reduce knowledge production if legacy institutions or large research clusters raise researcher output. The spatial reallocation of researchers therefore entails a trade-off between gains in knowledge access and losses in knowledge production. To quantify the knowledge production loss from reallocation, we use bibliographic data to separate individual effects from location effects and show that location effects account for substantial differences in research output across locations. Our instrumental-variables estimates provide robust evidence of own-institution agglomeration effects but weaker evidence of external-cluster effects. We then incorporate these estimates into counterfactual reallocations and show that marginally reallocating researchers toward large metropolitan areas with relatively little research activity can plausibly raise aggregate output because the associated knowledge production losses are modest and can be offset by small access gains.

\vspace{1.5em}
\noindent\textbf{Keywords: }Academic Research, Researchers, Knowledge Access, Knowledge Spillovers, Knowledge Production, Agglomeration Economies, Spatial Inequality, Universities, Researcher Reallocation\\ \vspace{3em}\noindent\textbf{JEL Codes:} O33, J24, R12 

\bigskip
\end{abstract}
\setcounter{page}{0}
\thispagestyle{empty}
\end{titlepage}
\pagebreak \newpage

\doublespacing

\section{Introduction}\label{introduction}

Universities and research institutions raise local productivity, innovation, entrepreneurship, and human capital, but many of these benefits decline sharply with distance \citep{jaff1989academic, moretti_return2004, aghion2009causal, kantor2014knowledge, kantor2019research, hausman2022university, biasi2023, mertens2024, howard2024universities}. The localized nature of these spillovers provides a rationale for dispersing research capacity beyond established hubs and toward population and economic activity, so that a broader segment of the economy can access frontier knowledge. This rationale has shaped higher-education policy for more than a century, from the U.S. land-grant system to more recent decentralization initiatives in Europe.\footnote{The U.S.\ land-grant system was established under the Morrill Act of 1862, which supported colleges teaching agriculture and the mechanic arts ``to promote the liberal and practical education of the industrial classes'' \citep{nara_morrill_1862}. Land-grant institutions were intended to broaden access beyond traditional elite colleges \citep{aplu_landgrant_tradition, nea_landgrant_overview}. For example, Illinois founded its land-grant campus at Urbana--Champaign to expand access for working people statewide. Later examples include the University of California's expansion into fast-growing regions through campuses such as Irvine and Merced \citep{pelfrey2004,merritt_lawrence2007,lao2024merced}. More recently, similar efforts have appeared in higher-education policy in Europe. For example, France’s Université 2000 plan aimed to decentralize higher education by creating new university campuses outside traditional centers.}

However, knowledge production may itself benefit from spatial concentration. Researchers can be more productive at some institutions than at others because of location-specific factors, and larger research clusters may further raise productivity through agglomeration effects. Reallocating researchers away from productive hubs may therefore lower movers’ productivity; by shrinking the sizes of major research clusters, it may also reduce the productivity of researchers who remain. These forces create a trade-off: Dispersing researchers toward population and economic activity may improve access to frontier knowledge, but it may also reduce the aggregate amount of knowledge produced.

This paper quantifies the trade-off between knowledge access and knowledge production. Specifically, we ask whether dispersing researchers toward population and economic activity reduces knowledge production enough to offset the resulting gains in knowledge access. Using detailed bibliographic data from 1970 to 2019, we first document the growing geographic misalignment between population and research activity across the U.S. and the resulting widening disparities in local access to knowledge. We then estimate location effects in knowledge production and the causal effects of own-institution and external-cluster sizes on research output using an instrumental-variables strategy. Finally, based on these estimates and a spatial model of researchers and non-researchers, we conduct a series of counterfactual reallocations. To our knowledge, this is the first paper to formally characterize and quantitatively evaluate the trade-off between access benefits of researcher dispersion and knowledge production benefits of researcher concentration.

We begin by documenting the geography of researchers and knowledge access. Combining bibliographic data from the Microsoft Academic Graph (MAG) with population and employment, we show that research activity is highly concentrated today. The Northeast Corridor, coastal California, and a number of established research centers in the Midwest account for a disproportionate share of researchers and publications, while many large and fast-growing areas in the South and West have relatively little nearby research activity. This mismatch was less pronounced historically: Research activity was more closely aligned with population in 1970, but the distributions of population and researchers diverged as population shifted toward the South and West while research remained anchored in long-established institutions. As a result, research output per capita became substantially more unequal across metropolitan areas. The same pattern appears when research output per capita is scaled by the college-educated population, employment, business establishments, or high-technology activity, and it holds for both university and non-university research and across nearly all broad fields. Given the highly localized spillovers from research activity, these patterns suggest potential gains from dispersing research resources toward population and economic activity. Such dispersion, however, may reduce knowledge production by moving researchers away from large and productive legacy research clusters.

To quantify the knowledge production loss from dispersing researchers, we estimate an individual-level research production function using a panel constructed from MAG. We model a researcher's output as the sum of a portable individual effect and an institution-specific location effect. First, we empirically separate the individual effects from the location effects. Separating the two effects is essential for evaluating reallocation: if research output is driven primarily by portable individual effects, moving researchers across locations would have little effect on knowledge production. But if location effects are important, reallocation can substantially change research output, leading to large aggregate effect on knowledge production. 

Exploiting the panel structure of the bibliographic data, we use the variation generated by moves across institutions, with researcher fixed effects, to separate location effects from sorting on individual effects. However, limited mobility leaves many institutions weakly connected, making unrestricted location fixed effects imprecise. We therefore parameterize the location effects as a field-specific function of lagged observable characteristics of the institution and its surrounding research environment, and use the fitted values as our location effects. These predicted effects capture the portion of location effects associated with observed characteristics.

A potential concern with using a mover design to distinguish location effects from individual effects is that researchers with greater expected productivity growth may sort into more productive research clusters. If so, our measured location effects may be biased by capturing movers' pre-existing productivity trajectories. To examine this concern, we conduct an event study around moves across institutions. The estimates show no differential trends before a move and a persistent increase in output after a move to a higher-output institution. We then assess the importance of location effects by decomposing output differences across individuals. Location effects explain a modest but meaningful share of the differences between high- and low-output researchers, but a substantially larger share of the differences between researchers at large and small institutions or in large and small metropolitan research clusters. These sizable location effects imply that reallocating researchers can directly affect the output of those who move.

Second, we estimate agglomeration elasticities for own institutions and external research clusters. Own-institution cluster size is the number of active researchers in the researcher's field at her institution in a given year, while external-cluster size is the number of active researchers in the same field at other institutions within the same metropolitan statistical area (MSA) in a given year. Estimating the agglomeration elasticities is essential for evaluating researcher reallocation because moving researchers changes cluster sizes at both origin and destination institutions and MSAs. These elasticities determine how strongly those size changes affect the output of researchers, including those who do not move. 

The main identification challenge is that cluster size may be correlated with unobserved local productivity advantages. We address this concern using Bartik shift-share instruments that combine the 1995 detailed {\it subfield} composition of each institution and external cluster within a broad field with subsequent {\it worldwide} growth in those subfields. The instruments predict greater cluster growth for research environments initially specialized in globally expanding subfields. The regressions absorb researcher fixed effects, field-specific life-cycle profiles, persistent differences across institutions within fields, and common field-year shocks. We also control for researchers' and their coauthors' direct exposure to worldwide subfield growth, helping ensure that the instruments do not simply capture expanding opportunities in their own research areas.

The IV estimates provide robust evidence of own-institution agglomeration effects across all three outcomes.  In our preferred specification, a 10\% increase in own-institution cluster size raises publications by about 2.0\%, citations by 8.3\%, and patent citations by 1.5\%, suggesting that institutional scale may matter especially for research impact. Evidence for external-cluster effects is weaker: A 10\% increase in external-cluster size raises publications by about 0.7\%, while the estimates for citations and patent citations are small and imprecise. The effects also vary across fields, with particularly strong effects in engineering and the life sciences. These sizable agglomeration effects imply that, beyond its direct effect on movers, researcher reallocation can also indirectly affect non-movers’ output by changing the size of their own-institution and surrounding research clusters.

We then combine the location effects and agglomeration elasticities to compute the counterfactual aggregate research production in counterfactual reallocation scenarios that move researchers from metropolitan areas with more researchers than their population share would imply to areas with fewer: movers experience changes in location effects, while researchers who remain at the origins and destinations are affected by changes in cluster sizes of their institutions and MSAs. Because we do not estimate the elasticity governing the local economic benefit of research access, we do not impose a particular value \textit{a priori}. Instead, for each reallocation, we invert the exercise and solve for the break-even access elasticity at which the access gain exactly offsets the knowledge production loss. 

The results are similar across several exercises. Bilateral reallocations of researchers from institutions in Boston to those in Dallas–Fort Worth and from San Francisco to Las Vegas imply break-even access elasticities of 0.018--0.019. Thus, if the local access spillover elasticity is at least this large, the aggregate gains from broader knowledge access would exceed the accompanying loss in knowledge production. Reallocating researchers among the 50 largest MSAs lowers the break-even access elasticity to about 0.01. Extending the exercise to all MSAs produces larger knowledge-production losses because many destinations contain very small institutions with low location effects, but even then, the break-even elasticity remains around 0.04. All of these thresholds fall far below the 0.08 estimate of local university research spillovers in \citet{kantor2014knowledge}, as well as related estimates in \citet{jaff1989academic}, although the margin is narrower when researchers are dispersed across all MSAs containing research institutions. These results suggest if the true access elasticity is 0.08 or above, marginally reallocating researchers to become closer the population and economic activity would raise aggregate output.

The findings are robust across alternative measures of knowledge access and cluster size. Measuring access to research by publications or citations per capita lowers the break-even elasticity because movers from research-intensive origins bring more research output and impact than their headcount alone captures. Measuring cluster size by publication counts rather than researcher counts also reduces the estimated knowledge production loss. Reallocating researchers toward the college-educated population rather than the total population raises the break-even access elasticity only modestly, and restricting the analysis to STEM researchers produces similar results. Taken together, the results suggest that reallocating researchers to align more closely with population and economic activity---particularly toward large, underserved metropolitan areas---can plausibly raise aggregate output.

The paper makes three contributions. First, it documents a new fact about the geography of research: geographic access to research has become increasingly unequal because research activity remained anchored in legacy cities while population and economic activity moved. This finding is particularly relevant given existing work that shows that universities generate important local spillovers in productivity, innovation, entrepreneurship, and economic activity \citep{jaff1989academic,aghion2009causal,kantor2014knowledge,kantor2019research,hausman2022university,lerner2024wandering}. A related literature shows that proximity to universities also affects educational attainment and access to higher education \citep{card1995distance,rouse1995cc,black2020access,fu2022access,mountjoy2022cc,biasi2023,acton2024access}. \cite{fabre2023geography} and \cite{ishimaru2024geographic} show how the uneven distribution of colleges contributes to persistent spatial inequality in France and the United States. We contribute to this literature by documenting the geography of research and showing how it has diverged from the geography of population and economic activity over time and highlight the potential benefits of reallocating researchers to be more aligned with population and economic activity.

Second, we contribute to the literature on how local environments and spatial concentration affect knowledge production. Within academia, \citet{bosquet2017sorting} find that same-field peers raise publication quantity and quality among economists in France, \cite{chandra2025person} show that institutional factors account for a large share of variation in life-science researchers' output, and \citet{helmers2017my} find that proximity to a major UK scientific facility increases related research activity. Evidence on localized peer effects is not uniform: \cite{azoulay2010superstar} and \cite{waldinger2012peer} find limited evidence that geographic proximity to prominent scientists raises peers' productivity. Beyond academia, a large literature documents agglomeration economies among firms and workers \citep{glaeser1999,glaeser2001,rosenthal_strange2003,charlot_duranton2004,duranton2005testing,Arzaghi2008,rosenthal_strange2008,ellisonglaeserkerr2010,milliondollarplant,dcosta_overman2014,roca_puga2017,gaubert2018,eckertwalsh2022,baumsnow_pavan2024}, while research on innovation shows that knowledge spillovers are highly localized and that inventors, startups, and other knowledge workers benefit from proximity to innovative clusters \citep{jaffe1993,belenzon2013spatial,belenzon2013spreading,carlino2015agglomeration,moretti2021,atkin2022,emanuel2023power,guzman2024go}.\footnote{A related literature examines how team structure and communication costs shape research, and how information technologies alter the process of knowledge production \citep{hesse1993returns,cohen1996computer,van1996could,walsh1996computer,kaminer1998bibliometric,walsh2000connecting,hamermesh2002tools,rosenblat2004getting,adams2005scientific,wuchty2007increasing,agrawal2008restructuring,butler2008equalizing,jones2008multi,kim2009elite,ding2010impact,winkler2010diffusion,adams2011role,goldstein2024communication}.} On the theory side, \citet{dingel_davis_2019} show how interactions among knowledge workers can increase more than proportionally with cluster size. We contribute to this literature by separating the location effects from individual effects in knowledge production and estimating agglomeration economies at both the own-institution and external-cluster levels for academic researchers across a broad range of fields.

Third, we connect the literature on the local benefits of research access with the literature on the productivity gains from concentrating knowledge workers. The first provides a rationale for dispersing researchers, while the second implies a potential production benefit from concentration. We formally characterize this trade-off in a spatial framework and take a first step toward quantifying it by embedding our estimates of location and agglomeration effects. In doing so, we provide an initial quantitative assessment of whether alternative spatial allocations of researchers can plausibly generate net gains in aggregate output. This access-concentration trade-off parallels \citet{dingel_medical2023}, who study productivity and geographic access in medical services, and \citet{rossi_hansberg2023}, who examine spatial allocations in the presence of localized high-skilled spillovers.

The rest of the paper proceeds as follows. Section \ref{framework} presents a spatial framework to characterize the trade-off between knowledge access and knowledge production. Section \ref{data} describes the data. Section \ref{geography_statistics} documents the geography of research and access to knowledge. Section \ref{estimation} presents our empirical strategies and the corresponding estimates of location effects and agglomeration effects. Section \ref{counterfactual} conducts the counterfactual analysis. Section \ref{conclusion} concludes. 

\section{Trade-off between Knowledge Access and Knowledge Production}
\label{framework}

To characterize the spatial trade-off between knowledge access and knowledge production, we develop a spatial framework with researchers and non-researchers. The spatial allocation of researchers affects aggregate output through two channels. First, researchers improve nearby non-researchers' access to frontier knowledge. Second, researchers produce the frontier knowledge itself, and their productivity can be higher in larger or more productive research clusters.

Consider $J$ locations, indexed by $j$. Let $N_j$ denote the quality-adjusted number of researchers in location $j$, and let $M_j$ denote the number of non-researchers. We define $N_j$ formally below after introducing individual researcher productivity. 
The baseline model takes the spatial allocation of non-researcher population $\{M_j\}_{j=1}^J$ as given and studies the effect of exogenously reallocating researchers on aggregate output. Appendix \ref{model_extension} endogenizes the location choices of researchers and non-researchers.

\paragraph{Knowledge Access and the Local Economy} Economic output in location $j$ is modeled as the product of non-researcher input $M_j$ and labor productivity $S_{0j}^{1-s^N}\, G^{s^N}_j$:
\begin{equation*}
    Y_j = S_{0j}^{1-s^N}\, G^{s^N}_j\, M_j.
\end{equation*}
Labor productivity combines a location-specific component unrelated to frontier knowledge, $S_{0j}$, and access to frontier knowledge, $G_j$, in Cobb-Douglas form, with weight $s^N$ on knowledge access. Access to frontier knowledge in location $j$ is
\begin{equation}
\label{eq:access}
    G_j = S(N_j, M_j)\, Q = \left( S_0 + \left(\frac{N_j}{M_j}\right)^\theta \right) Q,
\end{equation}
where $Q$ is the location-invariant stock of frontier knowledge relevant to production, endogenized later as the steady-state outcome of knowledge
accumulation. The function $S(N_j, M_j)$, abbreviated as $S_j$, governs how accessible this frontier knowledge is to the local non-researchers.

We adopt the tractable form in Equation~\eqref{eq:access}. The baseline term $S_0$ captures non-local channels, such as publications and online media, which deliver access regardless of local researchers' presence. 
The second term captures access mediated by nearby researchers. Holding $N_j$ fixed, local access declines with $M_j$, reflecting congestion in researchers’ attention.
We assume $0<\theta<1$, so local access is increasing and concave in $N_j$: an additional researcher improves access, but the marginal gain declines as the local research community expands. For a given stock $Q$, this concavity creates a force favoring the dispersion of researchers toward economically large locations with relatively few researchers.\footnote{The current specification assumes that local access to knowledge is rival: access mediated by nearby researchers is congestible because researchers have limited time and attention.} 

To summarize the relative importance of local and non-local access factors, define
\begin{equation*}
    \lambda_j \;\equiv\; \frac{(N_j/M_j)^{\theta}}{S_0 + (N_j/M_j)^{\theta}} \in (0,1),
\end{equation*}
which is the share of local knowledge access factor attributable to the relative presence of nearby researchers. This term governs how strongly local output responds to the number of researchers: The elasticity of local access with respect to the number of local researchers is $\partial \ln S_j / \partial \ln N_j = \theta\lambda_j$, holding the number of local non-researchers fixed. Thus, holding the frontier knowledge stock $Q$ and local non-researchers $M_j$ fixed, the elasticity of local output $Y_j$ with respect to $N_j$ is $s^N\theta\lambda_j$.\footnote{When $S_0$ is small, access comes primarily from local researchers---$\lambda_j$ approaches 1 and the output elasticity approaches $s^N\theta$. When $S_0$ is large, non-local channels dominate---$\lambda_j$ and the output elasticity approach 0. Between these cases, $\lambda_j$ is small where local researchers are scarce and large where they are abundant, so output becomes more responsive to researcher presence as the existing local research base grows. Without the baseline term $S_0$, a constant-elasticity specification would imply very large gains from adding researchers to locations with a very small initial research base. The term $S_0$ tempers this small-base effect: when researchers are scarce, $(N_j/M_j)^\theta$ is small relative to $S_0$, so $\lambda_j$ is also small. An additional researcher therefore has a limited effect on total access in the thinnest research locations.}

\paragraph{Knowledge Production} The stock of frontier knowledge, $Q$, is also endogenous to the spatial distribution of researchers across locations. Because knowledge accumulates but can become obsolete over time, we model $Q$ using a perpetual-inventory approach \citep{romer1990, jones1995, bloometal2020}. Researchers add to the stock of knowledge each period, while a constant fraction of existing knowledge depreciates as it moves away from the frontier. We first specify the flow of new research produced in each period, and then describe how these flows accumulate into the stock $Q$.

The research output of researcher $i$ in location $j$ at time $t$ is 
\begin{equation}
\label{eq:research_output}
   F_{ijt}={D_{i}} A_{j}(N_j) Q_{t-1}^\eta.
\end{equation} 
The term $D_i$ captures the researcher's portable individual productivity. 
The term $A_{j}(N_j)$ captures location productivity, through factors such as institutions, infrastructure, networks, and local research spillovers.  
Finally, $Q_{t-1}$ is the existing stock of frontier knowledge, capturing the idea that new discoveries may stand on the shoulders of giants.\footnote{Our analysis focuses on steady-state spatial allocations. We therefore suppress time subscripts on $D_i$ and $N_j$, treating individual productivity and the spatial allocation of researchers as fixed over time. We use time subscripts for $F$ and $Q$ only to describe the law of motion. After deriving the steady-state values, we drop these time subscripts.} We assume $0<\eta<1$, so the existing knowledge stock raises current research productivity with diminishing returns and a finite steady state exists.\footnote{\label{footnote:eta}We restrict attention to $0<\eta<1$, so the knowledge stock converges to a finite steady state, as in \cite{jones1995}. The case $\eta=1$ corresponds to the linear ``standing on shoulders'' specification in \citet{romer1990}, while $\eta<0$ captures a “fishing out” effect in which ideas become harder to find as the stock of knowledge grows. If $\eta>1$, knowledge production rises more than proportionally with the existing stock, so there is no finite steady state.
}

To allow for potential agglomeration economies, we specify location productivity as
\begin{equation}
\label{eq:agglomeration}
    A_j(N_j)= A_{0j} N^\alpha_j,
\end{equation}
where $A_{0j}$ captures exogenous local determinants of research productivity that are unrelated to cluster size, while $\alpha>0$ governs the strength of research agglomeration economies. It can be interpreted as the elasticity of location productivity with respect to the local research cluster size. A larger $\alpha$ implies stronger productivity gains from  researcher concentration.

Let $\mathcal{R}_j$ denote the set of researchers working in location $j$, and define the quality-adjusted number of  researchers as $N_j=\sum_{i\in \mathcal{R}_j} D_i$. Total research output in location $j$ is defined as the sum of individual research output across all researchers in that location:
\begin{align*}
   \sum_{i\in \mathcal{R}_j}F_{ijt} &=\sum_{i\in \mathcal{R}_j} D_{i} A_{j} (N_j) Q_{t-1}^\eta \\
   &= A_{j} (N_j) N_j Q_{t-1}^\eta = A_{0j} N_j^{1+\alpha} Q_{t-1}^\eta.
\end{align*}
Aggregate research output is the sum of research output across all locations:
\begin{equation*}
    F_t=\sum_{j=1}^J A_{0j}N_j^{1+\alpha} Q_{t-1}^\eta.
\end{equation*}
The stock of frontier knowledge then evolves according to the standard perpetual-inventory law of motion: 
\begin{equation*}
    Q_t=(1-\delta)Q_{t-1} + F_t,
\end{equation*}
where $0<\delta<1$ is the depreciation rate of the existing knowledge stock. $F_t$ denotes the flow of new knowledge produced by researchers in period $t$. In steady state, $Q_t=Q_{t-1}=Q$. The steady-state stock of frontier knowledge is therefore
\begin{equation*}
    Q=\left(\frac{\sum_{j=1}^J A_{0j} N_j^{1+\alpha}}{\delta}  \right)^{\frac{1}{1-\eta}}.
\end{equation*}
Because our analysis compares steady states, we drop time subscripts from this point forward. We use $Q$ and $F$ to denote their steady-state values. 

\paragraph{Aggregate Output}
Combining the steady-state knowledge stock with Equation \ref{eq:access} gives aggregate output:
\begin{equation}
\label{eq:aggregate_production}
    Y =\sum_{j=1}^{J}
M_j S_{0j}^{1-s^N}
\underbrace{ \left( S_0 + \left(\frac{N_j}{M_j}\right)^{\theta} \right)^{s^N} }_{\text{Access}}
\underbrace{ \left(
\frac{\sum_{j=1}^{J} A_{0j} N_{j}^{1+\alpha}}{\delta}
\right)^{\frac{s^N}{1-\eta}} }_{\text{Frontier Knowledge}}.
\end{equation}
Let $y=\ln Y$. Taking logs gives
\begin{equation*}
    y = \ln B + s^N \ln Q,
\end{equation*}
where
\begin{equation*}
    B = \sum_{j=1}^{J} M_j S_{0j}^{1-s^N} \left(S_0 + \left(\frac{N_j}{M_j}\right)^{\theta} \right)^{s^N}.
\end{equation*}
The term $B$ captures the spatial distributions of researchers and non-researchers and how they determine access factor given the frontier knowledge stock and the spatial distribution of non-research productivity. The term $Q$ is the steady-state stock of frontier knowledge defined above.

\subsection{Spatial Reallocation of Researchers}
\label{trade_off}
Having derived aggregate output, we now characterize the paper’s central trade-off: knowledge access and knowledge production. Consider the {\it marginal} effect on log aggregate output of moving one quality-adjusted researcher from location $j$ to location $j'$: 
\begin{equation}
\label{eq:trade_off}
 \Delta y_{j \rightarrow j'}
=
\underbrace{
s^N\!\left[
\frac{Y_{j'}}{Y}\,\frac{\theta\lambda_{j'}}{N_{j'}}
-
\frac{Y_j}{Y}\,\frac{\theta\lambda_j}{N_j}
\right]
}_{\text{Effect through knowledge access}}
+
\underbrace{
\frac{s^N(1+\alpha)}
{(1-\eta) \delta Q^{1-\eta}}
\left[
A_{0j'}N_{j'}^{\alpha}
-
A_{0j}N_j^{\alpha}
\right]
}_{\text{Effect through knowledge production}}.
\end{equation}
Equation \ref{eq:trade_off} separates the two channels through which researcher reallocation affects aggregate output.
The first term captures the change in access to a given stock of frontier knowledge. The second term captures the change in the steady-state knowledge stock. The net effect is generally ambiguous: access tends to favor dispersion toward economically large, researcher-scarce locations, whereas knowledge production tends to favor locations with high research productivity.

\paragraph{Effect Through Knowledge Access}
Relocating a researcher from $j$ to $j'$ reduces access in the origin and raises it in the destination, so aggregate output increases only when the access gain at the destination outweighs the loss at the origin. The marginal access effect in each location depends on three factors. The output share $Y_j/Y$ captures the economic importance of the location---improved access matters more in economically larger locations. The access elasticity $\theta\lambda_j$ measures how strongly local access responds to researcher presence. Finally, $1/N_j$ reflects diminishing marginal access as the research base expands.

To make the intuition concrete, suppose that the destination $j'$ is Dallas, which has few researchers relative to its economic size. Let the origin $j$ be Boston, which has a smaller economy but many more researchers. If $\lambda_j$ and $\lambda_{j'}$ do not differ too sharply, moving a researcher from Boston to Dallas can raise aggregate access because the gain in Dallas exceeds the loss in Boston. This channel creates a force toward a spatial allocation of researchers that is more closely aligned with the distribution of non-research economic activity.

\paragraph{Effect Through Knowledge Production}
The second term captures the effect of reallocation on the steady-state stock of frontier knowledge, which depends on the location-specific research productivity, $A_j(N_j)=A_{0j}N_j^{\alpha}$, of the origin and destination. Moving a researcher from $j$ to $j'$ reduces knowledge production and thus the steady-state knowledge stock when $A_j(N_j)>A_{j'}(N_{j'})$. This loss arises even when $\alpha=0$, because locations differ in baseline productivity $A_{0j}$. When $\alpha>0$, the contraction of the origin cluster and expansion of the destination cluster create an additional agglomeration channel.\footnote{Under the constant-elasticity specification, moving researchers from a location with higher marginal research productivity to one with lower marginal research productivity reduces knowledge production. If agglomeration gains decline sufficiently with cluster size, however, an additional researcher could contribute more in a small cluster than in a large one, and dispersion could instead increase knowledge production. Section \ref{estimation} examines this possibility by allowing the elasticity to vary with cluster size.}

\paragraph{Decomposing the Knowledge-Production Effect}
The effect on knowledge production consists of two channels. First, when a researcher $i$ moves from location $j$ to $j'$, she experiences a direct change in output because she is exposed to a different research environment. Holding each location's productivity component fixed, researcher $i$’s contribution to the change in log steady-state knowledge stock is
\begin{equation*}
\Delta_i^{direct}\ln Q =\frac{ D_i\left[A_{j'}(N_{j'}) - A_{j}(N_{j})\right]}{(1-\eta)\delta Q^{1-\eta}}
= \frac{D_i\left[A_{0j'} N^\alpha_{j'}-A_{0j} N^\alpha_{j}\right]}{(1-\eta)\delta Q^{1-\eta}}.
\end{equation*}
This is the direct effect of reallocation on the mover. It depends on the difference in location productivity between the destination and origin, so identifying it requires separating location productivity, $A_j(N_j)=A_{0j} N^\alpha_{j}$, from individual researcher productivity, $D_i$.

Second, the move reduces cluster size at the origin and increases it at the destination, so a move also changes the productivity of researchers who remain in both locations. This indirect contribution to the change in $\ln Q$ is
\begin{equation*}
\Delta_i^{indirect}\ln Q = \frac{ D_i\left[N_{j'} A'_{j'}(N_{j'}) - N_{j} A'_{j}(N_{j})\right]}{(1-\eta)\delta Q^{1-\eta}}
= \frac{\alpha D_i\left[A_{0j'}N^{\alpha}_{j'} - A_{0j}N^{\alpha}_{j}\right]}{(1-\eta)\delta Q^{1-\eta}}.
\end{equation*}
Intuitively, removing a researcher of quality $D_i$ from origin $j$ lowers its location productivity by $D_i A_j'(N_j)$, reducing the output of incumbent researchers by $D_i N_j A_j'(N_j)$. The corresponding increase in cluster size at the destination raises the output of incumbent researchers at destination $j'$ by $D_i N_{j'} A_{j'}'(N_{j'})$. Under the constant-elasticity specification, $N_j A_j'(N_j)=\alpha A_j(N_j)$, so the indirect effect equals the direct effect scaled by the agglomeration elasticity $\alpha$. Section \ref{estimation} estimates both location productivity at observed cluster size, $A_j(N_j)$, and the agglomeration elasticity $\alpha$.

\subsection{Model Extension: Endogenous Location Choices}

The baseline model treats the spatial distributions of researchers and non-researchers as given and evaluates exogenous reallocations of researchers. Appendix \ref{model_extension} endogenizes both groups' location choices. This extension introduces an additional general-equilibrium channel: When researcher reallocation changes local access to frontier knowledge, it changes local productivity and can induce non-researchers to migrate.

Reallocating researchers toward researcher-scarce locations may attract non-researchers toward places with lower baseline productivity, partially offsetting the gains from improved knowledge access. Section \ref{counterfactual} shows that this migration channel is quantitatively small relative to the two central forces in the model: the access gains from dispersion and the knowledge production losses from weakening productive research clusters. Thus, we present the extension in the appendix.

\section{Data}
\label{data}
We combine bibliographic and patent records with supplementary data on higher education and local economic conditions to estimate the research production function and to conduct counterfactual analyses.

\subsection{Bibliographic Data: Microsoft Academic Graph (MAG)}

Our primary data source is Microsoft Academic Graph (MAG), a large bibliographic database developed by Microsoft Research.\footnote{Microsoft Academic was an academic search engine developed by Microsoft Research. It ceased operations at the end of 2021. Its successor, OpenAlex, incorporates and extends MAG. The MAG dataset, compiled by \cite{mas_compile}, is publicly available at https://zenodo.org/records/2628216. Appendix \ref{app:MAG} discusses MAG’s coverage and reliability, showing that, relative to other major bibliographic databases, MAG provides broader coverage of scientific publications and comparable citation counts \citep{thelwall2018does,wang2020microsoft,martin2021google,visser2021large}. The section also summarizes MAG’s use in recent economics research.} 
MAG contains records for more than 260 million scientific publications and their citation relationships since 1900. For each publication, the data provide consistent identifiers of authors, institutional affiliations at the time of publication, journals, and fields of study. Institutional affiliations are further matched to geographic coordinates and U.S. state and county identifiers by \cite{emakg}.

MAG assigns publications to detailed subfields. We use its classification of 20 broad fields:   
\begin{itemize} [nosep]
    \item {\bf Social Sciences}: Economics, Business, Sociology, Political Science, Psychology, Population, and Geography.
    \item {\bf Humanities}: History, Philosophy, and Art.
    \item {\bf Natural Sciences}: Mathematics, Physics, Chemistry, Geology, and Environmental Science.
    \item {\bf Engineering and Computer Science}: Engineering, Computer Science, and Materials Science.
    \item {\bf Life Sciences}: Biology and Medicine.
    \end{itemize} 
We assign each researcher a primary field based on the modal field of her publications. When several fields have the same publication count, we randomly select one of the tied fields, following \cite{moretti2021} and \cite{lerner2024wandering}.

Using MAG, we construct an annual panel of researchers with at least one U.S. institutional affiliation from 1970 to 2015. We define the beginning of a researcher’s career as the year of her first publication and follow her for 35 years. A researcher’s geographic location is determined by her institutional affiliation. If multiple affiliations are reported in a year, we select the institution that appears most frequently over the researcher’s career.\footnote{A researcher may move between U.S. and non-U.S. institutions during her career; she enters the sample if she has at least one publication associated with a U.S. institution. When several affiliations appear with equal frequency over a researcher's career, we first select a U.S. institution when applicable and otherwise select the institution with the smallest research cluster to obtain a more conservative estimate.} Researchers do not publish in every year. For years without publications, we record zero publications and carry forward her most recently observed affiliation. We define a move as occurring in the first year in which she publishes with a new institution.\footnote{Table \ref{table:sample_size} reports sample sizes.} 

\paragraph{Research Outcomes}
MAG provides two core measures of researcher output. Our baseline quantity measure is the number of publications produced by a researcher in a given year.
To measure research impact, we construct a ten-year citation measure. For each paper, we count the citations it receives within ten years of publication. We then sum these citations across all papers published by the researcher in that year.

\paragraph{Research Clusters}
We measure cluster sizes by field, both within a researcher's own institution (the internal cluster) and outside it (the external cluster). Internal cluster size is the number of active researchers in the researcher’s field at her own institution. External cluster size is the number of active researchers in the same field at other institutions within the same MSA, excluding the home institution.

We also use alternative measures of cluster size and an alternative geographic definition of the external cluster. First, we measure internal and external cluster size using the number of publications rather than the number of active researchers. Second, we replace the MSA definition with a distance-based external cluster consisting of same-field researchers within 25 kilometers of the home institution, again excluding the home institution. Table \ref{table:cluster} presents the top 10 metropolitan areas by the number of researchers and publications.

\subsection{Patent Data}
To measure the influence of scientific research on patented innovation, we link MAG to U.S. patent records from PatentsView, an open-access database maintained by the U.S. Patent and Trademark Office. It provides linked records on patents, inventors, assignees, inventor locations, technology classes, and patent citations. 

We supplement PatentsView with the patent-to-science linkage developed by \cite{reliance_on_science}, which identifies citations from patents to {\it scientific} publications and matches the cited papers to MAG records. We aggregate these linked citations to the researcher-year level to construct our third research outcome: the number of patent citations received by papers published by the researcher in that year.

\subsection{Additional Data Sources}

\paragraph{Integrated Postsecondary Education Data System}
We use the Integrated Postsecondary Education Data System (IPEDS), administered by the National Center for Education Statistics, to measure tenured and tenure-track faculty across institutions and MSAs. These counts allow us to track changes in the geography of university faculty, a group distinct from the researchers identified in our bibliographic data.

\paragraph{American Community Survey}
We use the American Community Survey (ACS) accessed through IPUMS NHGIS to measure population and demographic composition at detailed geographic levels \citep{nhgis}. The data provide population counts by educational attainment, income, and occupation. We use census-tract centroids to calculate geographic access to research activity and aggregate these measures to larger geographic units using appropriate population weights. We also use 2015--2019 ACS wage data to construct the measure of non-research productivity used in the counterfactual analysis.

\paragraph{Employment and Business Establishments}
We use the Quarterly Census of Employment and Wages and County Business Patterns to measure local employment and business establishments by industry and geography \citep{cbp2021}. These data allow us to examine geographic access to research from the perspective of workers and firms, including those in high-technology industries.

\section{The Geography of Researchers and Access to Knowledge}
\label{geography_statistics}
\paragraph{Current Spatial Distribution} Research activity is highly uneven across the United States, generating large spatial differences in access to knowledge. Figure \ref{fig:map_access} maps nearby research papers per 1,000 residents between 2014 and 2019 at the county level, with values ranging from essentially zero to 95 papers per 1,000 residents. Research activity is concentrated along the Northeast Corridor, coastal California, and around several established research centers in the Midwest, reflecting the legacy of nineteenth- and early-twentieth-century research institutions. 
In contrast, many areas in the South, interior West, and Great Plains have relatively little nearby research activity despite substantial population growth in recent decades. Thus, the current geography of research does not match that of population and economic activity.

\paragraph{Historical Divergence} 
The current mismatch was not always as pronounced. Research activity was more closely aligned with population historically, but the two distributions diverged as population shifted toward the South and West while research remained anchored in long-established institutions in legacy cities. Figure \ref{fig:case_studies} illustrates this pattern for seven large MSAs. In 1970, research papers per resident were broadly similar across these metropolitan areas. Over the next five decades, Boston and San Francisco pulled far ahead, while research output per resident in Dallas/Fort Worth, Miami, and Phoenix remained much closer to its earlier levels despite substantial population growth. The pattern is similar when research output per capita is scaled by the college-educated population, indicating that the divergence is not simply driven by increasing concentration of college-educated workers in legacy cities \citep{diamond2016}. 

The spatial persistence of research activity helps explain this divergence. Figure \ref{fig:top_10_cities} shows that, in 1970, the ten largest research clusters accounted for similar shares of the population and researchers. Their shares of researchers and research papers remained broadly stable through 2019, while their population share declined substantially. Thus, population moved away from established research centers, while research activity remained anchored to long-lived institutions.

Figure \ref{fig:trends} shows that the selected-MSA examples reflect a broader national pattern. Both the 90/50 ratio and the Gini coefficient of research papers per resident increased substantially between 1970 and 2019. The same measures also rose when research output was scaled by the number of college graduates, indicating that geographic access to research has become increasingly unequal even relative to the distribution of high-skilled workers. This increase appears in research produced by both universities and non-university institutions (Figure \ref{fig:gini_uni_nonuni}) and in every broad field except the humanities (Figure \ref{fig:gini_field}).\footnote{Figure \ref{fig:institution_type} reports the composition of types of research institutions by field. Universities remain the dominant host of researchers, but some fields show notable shares elsewhere: A large share of researchers in the Life Sciences work at medical centers, and a meaningful share of those in Engineering and the Natural Sciences are based at private enterprises.}

The appendix complements these inequality measures with evidence from the full distributions. The distribution of research output per resident and per college graduate became substantially more dispersed between 1970 and 2019 (Figure \ref{fig:distribution}).\footnote{Figure \ref{fig:population_researchers} shows that the rising access inequality is primarily the result of population migration away from the existing research-intensive locations.} The distributions of research output relative to local employment and business establishments, including high-technology sectors only, also widened (Figure \ref{fig:dist_emp_est}). Similar widening occurred separately for research produced by universities and non-university institutions (Figure \ref{fig:dist_univ_nonuniv}). Finally, the distribution of faculty per capita remained comparatively stable over time, suggesting that the growing geographic dispersion of research access reflects more than the broader distribution of faculty employment (Figure \ref{fig:research_vs_faculty}).

\section{Estimation of the Research Production Function}
\label{estimation}

The previous section documents substantial geographic disparities in research access, suggesting potential access gains from reallocating researchers toward underserved locations.\footnote{Appendix \ref{app:back_of_envelope} provides a back-of-the-envelope calculation of the potential aggregate output gains from improving geographic access to research, holding the stock of frontier knowledge fixed.} The key remaining question is how much such reallocation would change the production of frontier knowledge. In this section, we estimate the research production function used to answer this question.

We proceed in two steps. First, we measure the field-specific location effects across institutions, separately from the portable individual effects that researchers carry with them. These location effects capture baseline differences in research output associated with moving across institutions. Second, we estimate the \textit{causal} effects of own-institution and external-cluster size on researchers’ output. These elasticities determine how reallocation affects researchers’ output through changes in cluster sizes, including that of researchers who remain in the origin and destination clusters.

\subsection{Individual vs. Location Effects}
\label{estimate_location}
The research production function in Equation \ref{eq:research_output} specifies that a researcher's output depends on both an individual-specific component and a location-specific component. Taking logs gives
\begin{equation}
\label{eq:log_research_output}
    \ln F_{ijt} = \underbrace{\ln D_i}_{\text{Individual-specific}} + \underbrace{\ln A_j(N_j)}_{\text{Location-specific}} + \eta \ln Q_{t-1}.
\end{equation}
Equation \ref{eq:log_research_output} clarifies why location effects are central to the counterfactual analysis. When a researcher is reallocated, she carries her individual productivity with her but her output may change because she enters a different research environment. Thus, to quantify the direct effect of reallocation on the productivity of researchers who move (hereafter ``movers"), we need to empirically separate location effects from individual researcher effects.

Because our data contain information on researchers' fields, institutions, and years of publication, we enrich the baseline model to allow for field- and time-specific variation. Let $F_{ijft}$ denote the research output of researcher $i$, affiliated with institution $j$, in field $f$, and year $t$. 
Let $a$ denote academic age, measured as years since the researcher's first publication. As an empirical counterpart of Equation \ref{eq:log_research_output}, we write research output as the additive sum of an individual effect, a field-specific location effect, a field-year effect, and an idiosyncratic component:
\begin{equation*}
\ln F_{ijft} = \underbrace{d_i + d_{fa}}_{\text{Individual Effect}}+ \underbrace{d_{jft}}_{\text{Location Effect}} + d_{ft} + u_{ijft}. 
\end{equation*}
The individual researcher effect $d_i$ captures time-invariant researcher productivity, and the field-by-academic-age effect $d_{fa}$ captures field-specific life-cycle profiles. Together, they correspond to the portable individual productivity component $\ln D_i$ in the model. The term $d_{jft}$ is the empirical counterpart of the location-specific productivity component $\ln A_j(N_j)$, allowing location productivity to vary over time and across fields within an institution. The field-year effect $d_{ft}$ absorbs the frontier-knowledge component $\eta \ln Q_{t-1}$, along with other nationwide field-specific time shocks. 

\paragraph{Inverse Hyperbolic Sine (IHS) Transformation of $F_{ijft}$} 
Because research output can be zero, in the empirical implementation we replace $\ln F_{ijft}$ with the inverse hyperbolic sine (IHS) transformation of $F_{ijft}$, which is defined and behave linearly at zero and behaves like a logarithm for large positive values. To preserve the connection with the model, we continue to express the specification in logarithmic form. In the counterfactual analysis, we use the IHS transformation in the research production function to remain consistent with our empirical estimates. In Section \ref{zero_outcomes}, we assess and demonstrate the robustness of the IHS specification. 

\paragraph{Identification of Location Effects} 
The annual researcher panel allows us to empirically separate individual effects from location effects. Identification comes from researchers who move across institutions. Because a mover carries her individual effect with her, the change in her output following a move identifies the difference between the location effects of the origin and destination institutions. Movers therefore identify location effects throughout the network of institutions they connect. These location effects, in turn, allow us to recover the individual effects of non-movers at those institutions as each researcher’s output net of the relevant location and field-year effects.

\paragraph{Empirical Challenge: Limited Mobility}
A natural approach would estimate $d_{jft}$ using institution-field-year fixed effects.\footnote{Because a complete set of institution-field-year effects spans the field-year effects, such an approach would first require an additional restriction, such as holding location effects fixed within multi-year periods.}
However, in our context, only 44,009 of 187,865 researchers---approximately 23\%---move at least once. Thus, many institution-field cells are weakly connected or disconnected from the mobility network. A fixed-effects approach would yield imprecise estimates. Restricting the analysis to a sufficiently connected mobility network would exclude many institutions and result in a substantially smaller selected sample.

\paragraph{Predicted Location Effects}
To address the limited-mobility problem, we model the location effect using observable characteristics of the institution and its surrounding research environment:
\begin{equation*}
d_{jft} = X_{jft}'\beta_f + \xi_{jft},
\end{equation*}
where $X_{jft}$ is a vector of observed local characteristics and $\xi_{jft}$ is the portion of location effects not captured by these observed characteristics. Substituting this expression into the empirical production function gives
\begin{equation}
\label{eq:location_predict}
\ln F_{ijft} = \underbrace{d_i + d_{fa}}_{\text{Individual Effect}}+ \underbrace{X_{jft}'\beta_f}_{\text{Predicted Location Effect}} + d_{ft} + \epsilon_{ijft}, 
\end{equation}
where $\epsilon_{ijft}=u_{ijft}+\xi_{jft}$. 

The vector $X_{jft}$ includes lagged log measures of active researchers, publications, ten-year citations, and star researchers in field $f$, measured both at institution $j$ and in the external MSA. These variables are constructed using the preceding five years, excluding the current year, so that the researcher’s current output does not directly enter the local characteristics. We allow the coefficients $\beta_f$ to vary across fields.

We define the predicted location effect as
\begin{equation*}
\hat{d}^{~pred}_{jft} = X_{jft}'\hat{\beta}_f,
\end{equation*}
which captures the portion of the full location effect predicted by observed characteristics of the institution and its surrounding research environment. The portion associated with unobserved local characteristics enters the residual through $\xi_{jft}$. These predicted location effects provide the location-productivity inputs used in the counterfactual analysis. When a researcher is reallocated, her individual effect is held fixed, while the predicted location effect changes from that of her origin to that of her destination.

\paragraph{Identification of $\beta_f$ and the Predicted Location Effects} 

Importantly, the coefficients $\beta_f$ should \textit{not} be interpreted as the causal effects of location characteristics, because those characteristics may be correlated with unobserved features of the research environment. The predicted location effects $X_{j,t}'\beta_f$ nevertheless \textit{do} have a causal interpretation. Coefficients in $\beta_f$, though predictive and non-causal, captures how the location effects, which are causal, vary with observed location characteristics. 

The identification of $\beta_f$ follows the same logic as the mover design rather than relying on exogenous variation in any particular location characteristic. Changes in a researcher’s output following a move identify differences in location effects across institutions, net of her portable individual effect and field-year shocks. The coefficients $\beta_f$ then describe how these mover-identified location effects covary with the observed characteristics $X_{j,t}$. If researchers experienced no systematic change in output when moving between institutions with different characteristics, the predicted difference $\left(X_{j,t}-X_{j',t}\right)'\beta_f$ would be close to zero. Conversely, if researchers systematically experienced larger output gains when moving to larger institutions, the coefficient on institution size would be positive, generating larger predicted location effects for larger institutions. This relationship does not imply that institution size causally increases location effects; it only describes how mover-identified location effects covary with institution size.

\subsubsection{Event-Study Evidence Around Moves}
\label{event_study}

A key identification concern in the mover design is that researchers anticipating future productivity growth may sort into more productive research clusters. In that case, the predicted location effects could partly reflect researchers’ pre-existing productivity trajectories rather than the effects of their new research environments. To assess its likelihood, we examine research output around moves and test whether changes in output occur only after migration or instead reflect trends already underway before the move. We do not use the event-study estimates to construct the location effects; the event study provides complementary evidence on whether changes in research environments are followed by changes in the same researcher’s output.

We follow the event-study approach in \cite{finkelstein2016sources} and \cite{lerner2024wandering}.\footnote{A standard event study that plots average output around all moves is not informative in this setting because movers have different origins and destinations. Moves from less to more productive locations and moves in the opposite direction could offset each other, even if location effects are large.} 
We restrict the sample to researchers who change institutions once, so that pre-move outcomes are not affected by earlier moves. To account for the direction and magnitude of each move, we scale the event-study coefficients by the destination-origin difference in average research output. Let $o(i)$ and $d(i)$ denote researcher $i$'s origin and destination institution, respectively, and let $f(i)$ denote her field. Define
\begin{equation*}
    \Delta_i
    =
    \overline{\ln F}_{d(i)f(i)}
    -
    \overline{\ln F}_{o(i)f(i)},
\end{equation*}
where $\overline{\ln F}_{jf}$ is the average log research output in institution $j$ and field $f$ over all researcher-years (including non-movers). We then estimate
\begin{equation}
\label{eq:event}
    \ln F_{ijft} = \sum_{s\neq -1} \delta_s \Delta_i \mathbf{1}\{t-t_i^m=s\}
    + \sum_{s\neq -1} \lambda_s \mathbf{1}\{t-t_i^m=s\}
    + d_i + d_{fa} + d_{ft} + u_{ijft},
\end{equation}
where $t_i^m$ is researcher $i$'s move year, and $s$ indexes years relative to the move. The omitted period is the year before the move ($s=-1$). 

The coefficients of interest are $\delta_s$, which measure how much a mover's research output changes around the move, scaled by the average output gap between the destination and origin. If moves are driven by anticipated output growth, we would expect $\delta_s$ to rise before the move. In the absence of sorting on expected productivity growth, the pre-move coefficients should be close to zero.

\paragraph{Event-Study Results}

Figure \ref{fig:event_location} shows no statistically significant pre-trend in any outcome during the 15 years before a move. Output rises after researchers move to higher-output institutions and remains persistently higher. 
The initial increase may partly reflect short-run institutional support, but the persistence of the effect is consistent with location-specific productivity advantages. 
The estimates for patent citations are less precise, likely because patent citations are rare and relevant to fewer fields. Overall, the results are consistent with adjustment to the destination research environment rather than a trend already underway.\footnote{Appendix Figure \ref{fig:event_location_direction} examines moves to higher- and lower-output institutions separately and shows similar adjustment patterns, suggesting that research output adjusts in the direction predicted by the destination-origin output gap. This symmetry also supports the additive separability assumption in Equation \ref{eq:location_predict}.}

\subsubsection{Additive Decomposition: Individual vs. Location Effects}
\label{decomposition}

We next use the estimates from Equation \ref{eq:location_predict} to quantify the importance of individual and predicted location effects for differences in research output across researchers.
We conduct the decomposition over 2010--2015. Let $g\in\{H,L\}$ denote the high and low groups in a given comparison. For comparisons by researcher productivity, researchers are ranked within their fields by their average value of the corresponding research output over the period. For comparisons by researchers' cluster size, institutions and MSAs are ranked within each field by their average number of active researchers.

For each researcher represented in group $g$, we average the outcome and each component estimated in Equation \ref{eq:location_predict} over the years in which she is observed in that group and then average across researchers. Let $\overline{\ln F}_{g}$, $\bar{d}^{\,indiv}_{g}$, $\bar{d}^{\,location}_{g}$, $\bar{d}^{\,field-year}_{g}$, and $\bar{\epsilon}_{g}$ denote the resulting group averages, where $\bar{d}^{\,indiv}_{g}$ is the group average of $\hat{d}_i + \hat{d}_{fa}$, $\bar{d}^{\,location}_{g}$ is the group average of $X'_{jft}\hat{\beta}_f$,  $\bar{d}^{\,field-year}_{g}$ is the group average of $\hat{d}_{ft}$, and $\bar{\epsilon}_{g}$ is the group average of the estimated residuals.

Therefore, the difference in average output between the high and low groups can be written as
\begin{align*}
\overline{\ln F}_{H}-\overline{\ln F}_{L} &= \left( \bar{d}^{\,indiv}_{H} - \bar{d}^{\,indiv}_{L}\right) + \left( \bar{d}^{\,location}_{H} - \bar{d}^{\,location}_{L}\right) \\
&+ \left( \bar{d}^{\,field-year}_{H} - \bar{d}^{\,field-year}_{L}\right) +\left( \bar{\epsilon}_{H} - \bar{\epsilon}_{L} \right).
\end{align*}
The shares of the observed output difference attributable to the difference in individual effects, predicted location effects, and unobserved factors are
\begin{align}
\label{eq:decomp_share}
    \text{Share}^{indiv}_{H,L} &=
    \frac{\bar d^{\,indiv}_{H} - \bar d^{\,indiv}_{L}}
    {\overline{\ln F}_{H} - \overline{\ln F}_{L} }, \\ \notag
    \text{Share}^{location}_{H,L} &=
    \frac{\bar d^{\,location}_{H} - \bar d^{\,location}_{L}}
    {\overline{\ln F}_{H}-\overline{\ln F}_{L}}, \\ \notag
     \text{Share}^{unexplained}_{H,L} &=\frac{\bar{\epsilon}_{H} - \bar{\epsilon}_{L}}
    {\overline{\ln F}_{H}-\overline{\ln F}_{L}}.
\end{align}
The unexplained share may partly capture location effects associated with unobserved local characteristics.
We expect the field-year component to be small because groups are defined within fields over a short period.

\paragraph{Decomposition Results}

Table \ref{table:decomp_output} presents the decomposition by individual researcher productivity. Its three panels correspond to IHS publications, citations, and patent citations; within each panel, the columns compare researchers above and below the median, in the top and bottom quartiles, and in the top and bottom deciles. The upper block reports the overall output difference and the corresponding component differences, while the lower block reports each component as a share of the overall difference. Researcher effects account for 63–67\% of the publication gap, 71–76\% of the citation gap, and 75\% of the patent-citation gap. Location effects account for 8\%, 7\%, and 5–10\% of these gaps, respectively.  Although these shares are smaller, they are still quantitatively important.\footnote{The remaining 16–29\% of the output gap is unexplained and may include location productivity associated with unobserved local characteristics. The field-year component is not reported because its contribution is negligible.} Since we are comparing \textit{individuals} based on their productivity, it is unsurprising that the researchers' individual effects account for the largest explanatory share. 

More directly relevant for assessing the effects of reallocating researchers across institutions and MSAs, Table \ref{table:decomp_size} groups researchers by institution or MSA size within fields. Location effects explain a substantial share of the research output differences: 38–80\% for publications, 34–64\% for citations, and 82–94\% for patent citations.\footnote{These large explanatory shares may partly arise mechanically because the predicted location effects are constructed from observed characteristics of the local research environment, including the grouping variables themselves, namely, institution and MSA size. Grouping researchers by these same characteristics may generate greater separation in predicted location effects by construction. Figures \ref{fig:decomp_field_output}--\ref{fig:decomp_field_msasize} report the corresponding field-specific results. The broad patterns are similar across fields, although the magnitude of the shares varies substantially.} The sizable variation in location effects across institutions and MSAs suggests that reallocating researchers across locations may have substantial effects on research output, which we formally quantify in Section \ref{counterfactual}.

\subsection{Estimating Agglomeration Effects}
\label{estimation_agglomeration}

Section \ref{estimate_location} separates individual from location effects to quantify the direct effect of reallocation on researchers who move. We now estimate how research output responds to changes in own-institution and external-cluster size. These elasticities determine the indirect effect of reallocation on researchers who remain in the origin and destination clusters.

Recall from Equation \ref{eq:agglomeration} that location productivity depends on research cluster size: $\ln A_j(N_j)=\ln A_{0j}+\alpha \ln N_j$, where $A_{0j}$ captures location-specific productivity unrelated to cluster size and $\alpha$ is the agglomeration elasticity. To capture agglomeration flexibly, we distinguish between internal and external cluster size. Internal cluster size, $N^I_{jft}$, is the number of active researchers in field $f$ at institution $j$ in year $t$. External cluster size, $N^E_{jft}$, is the number of active researchers in the MSA containing $j$ within the same field, excluding the home institution. 
Specifically, we estimate
\begin{equation}
\label{eq:agg_reg}
\ln F_{ijft} = \underbrace{d_i + d_{fa}}_{\text{Individual Effect}}+ \underbrace{\alpha^{I} \ln N^{I}_{jft} + \alpha^{E} \ln N^{E}_{jft} + d_{jf}}_{\text{Location Effect}} + d_{ft} + u_{ijft}, 
\end{equation}
where $\alpha^{I}$ and $\alpha^{E}$ are the internal and external agglomeration elasticities, respectively.\footnote{\label{footnote:log}Because cluster size can be zero, in the empirical implementation, we enter both cluster size measures as $\ln(N+c)$, with $c=1$ in the baseline analysis, although we retain the notation $\ln N$ to maintain the connection with the model. Adding one is the minimal adjustment needed to define the regressor at zero and keeps the specification closest to the model's log form. The implied elasticity with respect to $N$ is $\alpha N/(N+c)$, which approaches $\alpha$ when $N$ is large relative to $c$. When $N$ is small, the implied elasticity is more sensitive to the additive constant. Section \ref{counterfactual} examines this sensitivity in counterfactuals involving very small clusters.} The institution-field fixed effect $d_{jf}$ absorbs persistent productivity differences across institution-field cells that are unrelated to contemporaneous cluster size, corresponding to the empirical counterpart of $\ln A_{0j}$. 

Section \ref{estimate_location} explains that limited mobility makes it difficult to estimate the levels of location fixed effects precisely. This is less of a concern here because we do not use the estimates of $d_{jf}$---they enter only as controls. Our parameters of interest are $\alpha^{I}$ and $\alpha^{E}$, which are identified from changes in cluster size within institution-field cells over time, controlling for individual effects and common field-year shocks.

As a further robustness check, we add MSA-by-year fixed effects to absorb time-varying shocks common to all institutions in the same metropolitan area. 
These fixed effects also reduce the identifying variation in external-cluster size, so we treat this specification as a robustness check rather than the baseline.

\paragraph{Identification Challenge} 
The main identification challenge is that cluster size is endogenous. Cluster size may be correlated with unobserved local productivity advantages, such as the generosity of funding or the quality of research infrastructure. Institution-field fixed effects absorb time-invariant advantages of each institution-field pair, but they do not absorb time-varying institution-field productivity shocks. Thus, OLS estimates of $\alpha^I$ and $\alpha^E$ may be biased.

\subsubsection{Bartik Instrument for Cluster Size}

To address this identification problem, we construct Bartik shift-share instruments using detailed {\it subfields} nested within each broad field $f$. The instruments combine each institution or external cluster's predetermined composition across detailed subfields with subsequent {\it worldwide} growth in those subfields. The idea is that institutions and local research clusters initially specialized in globally expanding subfields should experience larger predicted growth in the number of researchers. This predicted growth is likely driven by the global evolution of scientific fields rather than by contemporaneous productivity shocks at any single institution or local cluster.\footnote{For example, a statistics department or regional statistics community with a large preexisting share of researchers in machine-learning-related areas would be predicted to grow more as machine learning expands worldwide. For a researcher in that department or cluster, controlling for characteristics of her own subfields, this predicted growth shifts local cluster size through channels that are plausibly unrelated to the direct productivity shocks of any individual researcher.}

\paragraph{Construction of the Bartik Instruments} 

To allow worldwide subfield growth sufficient time to generate variation while smoothing annual noise in research output, we organize the IV construction around three periods: 1995, 2005, and 2015. Each period pools data over the focal year and the two preceding years: For example, the 1995 period uses data from 1993--1995. We use 1995 to measure base-period cluster sizes and subfield shares. We then predict cumulative growth from the base period to each later period, 2005 and 2015, using worldwide growth in detailed subfields.

Let $h$ index a detailed subfield within broad field $f$. Let $N^I_{jh,0}$ and $N^E_{jh,0}$ denote the 1995 base-period number of active researchers in subfield $h$ at institution $j$ and in its external cluster, respectively.
The corresponding baseline shares of subfield $h$ within broad field $f$ at institution $j$ are
\begin{equation*}
    \omega^I_{jh(f)0} =
    \frac{N^I_{jh(f)0}}
    {\sum_{\ell \in f} N^I_{j\ell0}} \text{~~~and~~~}
    \omega^E_{jh(f)0} =
    \frac{N^E_{jh(f)0}}
    {\sum_{\ell \in f} N^E_{j\ell0}}.
\end{equation*}

Using worldwide publication records, let $g^w_{ht}$ denote the growth rate in the number of active researchers in subfield $h$ from the base period to period $t$, where $t\in\{2005,2015\}$. The predicted internal and external cluster growth rates are share-weighted sums of worldwide subfield growth:
\begin{equation*}
    \widehat g^I_{jft}
    =
    \sum_{h\in f}
    \omega^I_{jh(f)0} g^w_{ht}
    \text{~~~and~~~}
    \widehat g^E_{jft}
    =
    \sum_{h\in f}
    \omega^E_{jh(f)0} g^w_{ht}.
\end{equation*}

Because we estimate the production function in levels rather than first differences, we convert these predicted growth rates into predicted cluster-size {\it levels}. Let $N^I_{jf0}$ and $N^E_{jf0}$ denote base-period internal and external cluster sizes. The internal and external Bartik instruments for periods 2005 and 2015 are  
\begin{align}
\label{eq:bartik}
\ln \hat{N}^{I,\text{Bartik}}_{jft} &= \ln\!\left[\, N^{I}_{jf0}\left(1+\hat{g}^{I}_{jft}\right)\right], \text{~and} \\ \notag
\ln \hat{N}^{E,\text{Bartik}}_{jft} &= \ln\!\left[\, N^{E}_{jf0}\left(1+\hat{g}^{E}_{jft}\right)\right],
\end{align}
for $t\in\{2005,2015\}$. The predicted cluster size in 1995 is simply its baseline value and does not provide meaningful shift-share variation. Because all IV specifications include institution-field fixed effects, these base-period cluster sizes are absorbed by these fixed effects. Identification therefore comes from differences across institution-field cells in Bartik-predicted growth over 1995--2005 and 1995--2015.

\paragraph{Identifying Assumptions} The Bartik instruments must satisfy two conditions:
\begin{itemize}[noitemsep]
\item \textbf{Inclusion Restriction:} Baseline exposure to rapidly growing detailed subfields must predict subsequent growth in actual cluster size. We examine this restriction directly in the first-stage regressions.
\item \textbf{Exclusion Restrictions:} Conditional on the fixed effects and controls in Equation \ref{eq:agg_reg}, the Bartik-predicted changes in cluster size must affect a researcher's output only through internal and external cluster size. To ensure the validity of this restriction, one of the following two conditions should hold.
\begin{itemize}[noitemsep]
\item \textit{Exogeneity of worldwide subfield growth.} Following \citet{borusyak2022quasi}, worldwide subfield growth $g^{w}_{ht}$ must be uncorrelated with institution-field-specific productivity shocks. This requires worldwide subfield growth to reflect broad scientific trends rather than shocks originating from any particular institution or local research cluster. This condition is plausible because no single institution or cluster accounts for a large share of the worldwide stock of researchers in a subfield.

\item \textit{Exogeneity of baseline shares.} Following \cite{goldsmith2020bartik}, base-period subfield shares must be uncorrelated with subsequent institution-field productivity shocks. The condition is plausible because the shares are measured before the growth they predict, and persistent local advantages associated with the shares are absorbed by institution-field fixed effects. The remaining concern is that base-period shares may proxy for differential institution-field trends, which we examine using the pre-trend tests below.

\end{itemize}
\end{itemize}

\paragraph{Potential Threats to Identification}

We address three potential threats to the exclusion restriction.

\textit{Direct exposure to worldwide subfield growth.} Worldwide growth in a researcher's own subfields may directly affect her output by increasing publication opportunities or citation demand, rather than solely through changes in local cluster size. To address this concern, we control for the researcher’s and her coauthors’ direct exposure to worldwide subfield growth. Let $\mathcal H_i$ denote the set of subfields associated with researcher $i$, and let $\omega^A_{ijh,0}$ denote the base-period share of subfield $h$ among subfields in $\mathcal H_i$ at institution $j$. Researcher $i$'s predicted {\it direct} exposure to worldwide subfield growth is thus 
\begin{equation*}
\widehat g^A_{ijt} =
\sum_{h\in \mathcal H_i}
\omega^A_{ijh,0} g^w_{ht}.
\end{equation*}
We construct the author's exposure control as
\begin{equation}
\label{eq:author_control}
\ln \widehat{N}^{A}_{ijft}
=
\ln\left[
N^I_{jf0}
\left(1+\widehat g^A_{ijt}\right)
\right].
\end{equation}
We construct an analogous coauthor exposure control, $\ln \widehat{N}^{C}_{ijft}$, using the detailed subfields associated with researcher $i$'s coauthors. Including both controls helps absorb the direct effects of worldwide growth in the researcher’s own research areas and collaboration network, and thus isolates the variation in cluster size generated by the Bartik instruments.

\textit{Concentrated weights.} A second concern is that the instruments may be driven by a small number of rapidly growing subfields whose baseline shares are correlated with unobserved institution-field-level productivity shocks. To assess this concern, we reconstruct the Bartik instruments, omitting one detailed subfield at a time, and re-estimate the preferred specification. If the estimates remain stable across these exercises, the results are unlikely to be driven by any single subfield.

\textit{Pre-existing productivity trends.} A third concern is that institutions initially specialized in rapidly growing subfields may already have been on different productivity trajectories before 1995. If future Bartik-predicted growth is correlated with output growth before the base period, the instruments could capture these pre-existing trends rather than exogenous changes in cluster size.
To assess this concern, 
we test whether Bartik-predicted cluster growth after 1995 is associated with institution-field output growth before the base period. Let ${\overline{\ln F}}_{jfr}$ denote average IHS research output for institution $j$ in field $f$ over the three-year period ending in year $r$, where $r\in\{1985,1990\}$. We estimate
\begin{equation}
\label{eq:pretrend}
\overline{\ln F}_{jf,1995}-\overline{\ln F}_{jfr}
= \pi_0 + \pi^I\widehat g^I_{jf,2005} + \pi^E\widehat g^E_{jf,2005} + \pi_f + \varepsilon_{jf},~~r\in\{1985,1990\},
\end{equation}
where $\pi_f$ denotes field fixed effects.
We repeat the same regression using predicted cluster growth over the longer horizon, $\widehat g^I_{jf,2015}$ and $\widehat g^E_{jf,2015}$. If predicted growth in an institution-field cell is not preceded by an upward trend in research output, the coefficients $\pi^I$ and $\pi^E$ should not be positive.

\subsubsection{Main Results}

\paragraph{OLS Estimates}
Appendix Table \ref{table:result_OLS_full} presents the OLS estimates of Equation \ref{eq:agg_reg} using the full sample from 1970--2015. Own-institution size is positively associated with all three research outcomes across specifications. The estimates for external-cluster size are smaller and vary across specifications. Because these estimates are descriptive and do not enter the counterfactual analysis, we focus on the IV estimates below.

\paragraph{First-Stage Estimates}
Before discussing our main IV estimates, we assess the relevance of the instruments in Table \ref{table:result_first_stage}.
The IV sample is organized around three periods: 1995, 2005, and 2015; each period pools data over the focal year and the two preceding years. Columns 1--3 use log own-institution size as the dependent variable, and Columns 4--6 use log external cluster size. Columns 1 and 4 include the baseline fixed effects from Equation \ref{eq:agg_reg}. Columns 2 and 5 add the author and coauthor controls for direct exposure to worldwide subfield growth. Columns 3 and 6 further add MSA-by-year fixed effects.

The results show that own-institution size is predicted almost entirely by the own-institution instrument, whereas external cluster size responds to both instruments. This cross-loading is plausible because an institution's initial subfield composition may reflect broader local scientific specialization and therefore predict growth both within the institution and in the surrounding research environment. As a result, the external-cluster elasticity is identified from part of predicted external-cluster growth that is distinct from predicted own-institution growth. This limited independent variation may contribute to the lower precision of the external-cluster estimates.

Adding the author and coauthor controls has little effect on the first stage. MSA-by-year fixed effects absorb additional geographic variation, particularly variation generated by the external-cluster instrument, but the instruments remain jointly relevant. The Kleibergen--Paap rk Wald F-statistics reported in Table \ref{table:result_IV_main} range from 12.97 to 15.88 across specifications.

\paragraph{Main OLS and IV Estimates}
Table \ref{table:result_IV_main} reports OLS and 2SLS estimates using the IV sample and follows the same sequence of fixed effects and controls as in Table \ref{table:result_first_stage}. We use Column 5 as our preferred specification because it includes the author and coauthor controls for direct exposure to worldwide subfield growth, while MSA-by-year fixed effects are included as a more demanding robustness specification in Column 6.

The OLS estimates provide a within-sample benchmark. Own-institution size is positively associated with all three research outcomes, whereas the coefficients on external-cluster size are close to zero.

The IV estimates show a robust positive effect of own-institution size. The coefficient is positive and statistically significant for all three outcomes and across all specifications. In our preferred specification, the estimates imply that a 10\% increase in own-institution size raises output by approximately 1.99\% for publications, 8.28\% for citations, and 1.48\% for patent citations.\footnote{One concern with the patent-citation results is that patent citations are relevant to fewer fields and are zero for many observations. As a robustness check, we restrict the sample to fields in which more than 30\% of papers receive at least one patent citation. The estimates remain similar. Due to space constraints, these results are not reported but are available upon request.} The effect is largest for citations, suggesting that agglomeration may matter more for research impact than for publication quantity. 

The evidence for external-cluster effects is weaker. The external-cluster coefficient is positive and statistically significant for publications, but is small and insignificant for citations and patent citations. Adding MSA-by-year fixed effects further reduces the precision of the estimates. The weaker evidence for external-cluster effects is consistent with the limited independent variation discussed in the first-stage results.\footnote{The IV estimates are larger than the corresponding OLS estimates, especially for citation-based outcomes. This pattern is consistent with attenuation bias in OLS due to measurement error in cluster size. It may also reflect heterogeneous treatment effects: The Bartik instrument identifies effects for institution-field cells whose size grows because they were initially exposed to globally expanding subfields. These cells may have especially high returns to scale, since larger institutions in fast-growing research areas may be particularly effective at increasing visibility, improving research quality, and strengthening connections to other researchers and inventors. This may generate especially large effects on citations and patent citations.}

Overall, the results indicate a robust {\it own-institution} agglomeration effect. For the baseline counterfactual analyses, we use the publication agglomeration elasticities $\hat{\alpha}^I=0.199$ and $\hat{\alpha}^E=0.074$.

\paragraph{Heterogeneity Across Fields}

The strength of agglomeration forces may differ across fields because the processes of knowledge production and dissemination vary substantially across fields. In addition, fields differ in the extent to which research relies on team production, specialized infrastructure, and commercialization, all of which may affect the gains from geographic clustering.

Table \ref{table:result_IV_field} re-estimates our preferred IV specification separately for the five broad field groups. Because splitting the sample weakens the first stage, we interpret these estimates as suggestive. The estimates broadly mirror the pooled results.
The own-institution coefficient is positive for publications in every field group, and the estimates are of the same order of magnitude, although only statistically significant for engineering and computer science. For citations, the coefficient is positive and statistically significant in all groups except the humanities, with the largest estimate in the life sciences. This is consistent with the possibility that institution scale may be particularly important for research impact in fields that rely heavily on teams, infrastructure, and professional networks.
For patent citations, the larger positive estimates are concentrated in the natural and life sciences, where chemistry, physics, biology, and medicine are closely linked to patents.

The external-cluster coefficients are generally smaller and less precisely estimated. They are positive and statistically significant for publications in the social sciences and life sciences, but show no consistent pattern for citations or patent citations.
Overall, the field-group-specific estimates are consistent with the pooled results: Agglomeration effects appear to operate primarily through the size of researchers’ own institutions, while the evidence for broader external-cluster effects is weaker. Because the subgroup estimates are less precisely identified, the counterfactual analysis uses the pooled estimates from Table \ref{table:result_IV_main}.

\paragraph{Leave-One-Out Test}
To assess whether a small number of subfields drive the IV estimates, we reconstruct both Bartik instruments 364 times after omitting one detailed subfield at a time and re-estimate the preferred specification for each outcome. Appendix Figure \ref{fig:leave_one_out} displays, for each outcome and both internal and external elasticities, the ten leave-one-out estimates with the largest absolute deviations from the corresponding baseline estimate. Even among these most influential omissions, the estimates remain nearly unchanged for both elasticities and all three outcomes. Moreover, every leave-one-out estimate of the own-institution elasticity remains statistically significant across all three outcomes, and every estimate of the external-cluster elasticity remains statistically significant for publications.

\paragraph{Pre-Trend Tests}
Table \ref{table:result_pretrend} reports the estimates of Equation \ref{eq:pretrend}, relating pre-1995 changes in institution-field research output to Bartik-predicted growth in own-institution and external-cluster size. We examine output growth over 1990--1995 and 1985--1995, using predicted cluster growth over 1995--2005 and 1995--2015. Predicted own-institution growth is generally uncorrelated or negatively correlated with pre-1995 output growth. These patterns would tend to {\it attenuate} rather than generate our positive own-institution estimates. Predicted external-cluster growth is positively associated with some pre-period outcomes, so we interpret the external-cluster estimates more cautiously, consistent with the generally weaker evidence for this channel.\footnote{The positive pre-trends associated with predicted external-cluster growth are less relevant to the own-institution elasticity, because in the first stage, the estimated effect of the external Bartik instrument on own-institution size is close to zero.}

\paragraph{Anderson-Rubin Confidence Sets} Although the Kleibergen--Paap rk Wald F-statistics do not indicate a serious weak-instrument problem, we report Anderson–Rubin confidence sets because they provide inference that remains valid under weak identification. Figure \ref{fig:ar} shows that the 95\% confidence sets exclude all parameter pairs with non-positive values of $\alpha^I$
for all three outcomes.

\subsubsection{Additional Robustness Checks}
\paragraph{Heterogeneity by Cluster Size}
Our model and baseline empirical specification assume {\it constant} internal and external agglomeration elasticities. However, if agglomeration gains decline sufficiently with cluster size, an additional researcher could contribute more in a small cluster than in a large one, so dispersion might raise rather than reduce knowledge production.
To assess this possibility, we assign each institution-field cell to a tertile based on its 1995 cluster size and interact current log cluster size with the corresponding tertile indicators: 
\begin{equation}
\label{eq:heter}
\ln F_{ijft} = d_i + d_{fa}+ \sum_{q=1}^3 \alpha^{I}_q \left( \ln N^{I}_{jft} \times Q^I_{jf,q} \right) + \sum_{q=1}^3 \alpha^{E}_q \left(\ln N^{E}_{jft} \times Q^E_{jf,q} \right) + d_{jf} + d_{ft} + u_{ijft},
\end{equation}
where $Q^I_{jf,q}$ indicates that institution-field cell $(j,f)$ belongs to tertile $q$ of the institution size distribution, and $Q^E_{jf,q}$ is defined analogously using external-cluster size; $\alpha_q^I$ and $\alpha_q^E$ are the internal and external agglomeration elasticities that vary across cluster-size groups. 

Table \ref{table:result_IV_heter} reports IV estimates of $\alpha_q^I$ and $\alpha_q^E$.
For publications, the own-institution elasticity is similar in the smallest and largest tertiles. For citations and patent citations, the own-institution elasticity declines with cluster size but remains positive and sizable even in the largest tertile. The external-cluster estimates show no systematic decline with size.
Overall, the results do not indicate a sharp decline in agglomeration gains, so we retain the pooled constant-elasticity estimates in the counterfactual analysis.

\paragraph{Alternative Measures of Cluster Size}
Table \ref{table:result_IV_cluster} considers alternative measures of cluster size---the number of active researchers and publications---and alternative definitions of external cluster---the MSA or a 25-km radius. All cluster measures are field-specific, and the home institution is excluded from the external cluster in all specifications. Own-institution effects remain positive and are especially stable for citations. However, for publications, the estimate becomes statistically insignificant when cluster size is measured by paper counts rather than author counts. The external-cluster effect on publications is positive and statistically significant under the MSA definition, but we find no statistically significant external-cluster effect for any outcome when the cluster is defined using a 25-km radius. 

\paragraph{Alternative Approaches for Zero-Valued Outcomes} 
\label{zero_outcomes}
Our baseline analysis uses the IHS transformation of a research outcome as the dependent variable because research outcomes contain zeros. 
\citet{chen2024logs} show that estimates using log-like transformations with zero-valued outcomes are not unit-invariant and the unit of the outcome implicitly determines the weight placed on the extensive margin. Thus, we use their approach by explicitly calibrating the value placed on the extensive margin.\footnote{\citet{chen2024logs} also recommend estimating the ATE in levels using Poisson regression and expressing it as a percentage effect. However, we do not pursue this approach because they note that Poisson IV regression with a continuous IV may not have a local average treatment effect interpretation.} We transform a research outcome using $m(F)=\log(F)$ for $F>0$ and $m(0)=-x$, considering $x=0.1$ and $x=1$, so that moving from zero to one is valued as equivalent to 10- and 100-log-point changes along the intensive margin, respectively.
Table \ref{table:result_IV_extensive} shows that the estimated effects using $x=0.1$ are slightly smaller than the baseline estimates, while the estimates using $x=1$ are slightly larger. Overall, the estimates are close to the baseline results and are not very sensitive to the choice of $x$.

\section{Costs and Benefits of Spatial Reallocation of Researchers}
\label{counterfactual}
Section \ref{estimation} provides the inputs from the research-production side needed to evaluate researcher reallocation: location effects and agglomeration effects. We now move to the counterfactual analysis based on Equation \ref{eq:aggregate_production}. For any counterfactual reallocation, the change in log aggregate output is
\begin{equation*}
\Delta \ln Y = \Delta \ln B + s^N \Delta \ln Q,
\end{equation*}
where $\Delta \ln B$ is the change in the knowledge access component of aggregate output and $\Delta \ln Q$ is the change in the steady-state knowledge stock.

In principle, we could evaluate each reallocation by directly computing both terms. However, in practice, we estimate the knowledge-production-side inputs in the previous section but not the access elasticity that governs $\Delta \ln B$.  Rather than imposing a particular estimate from the prior literature, we invert the exercise. For each reallocation, we solve for the smallest access elasticity that sets $\Delta \ln Y=0$, which we refer to as the {\it break-even access elasticity}. If estimates from the prior literature exceed this threshold, the gain from improved access would be large enough to offset the loss in knowledge production.

\subsection{Parameter Calibrations}
\label{section:calibration}
To calibrate the knowledge production parameters, we use the predicted location effects, $\hat{d}^{~pred}_{jft}$, and internal and external agglomeration elasticities, $\hat{\alpha}^I$ and $\hat{\alpha}^E$, estimated in Section \ref{estimation}. On the knowledge-access side, we take $s^N$, the importance of frontier knowledge in economic output; $\delta$, the depreciation rate of knowledge; and $\eta$, the dependence of new research on the existing knowledge stock, from the literature. We calibrate $\lambda_j$, the share of local knowledge access attributable to nearby researchers to several plausible constant values. We do not impose a value for $\theta$, the elasticity of knowledge access with respect to researchers per capita, but instead leave it free and solve for its break-even value in each counterfactual.

\paragraph{Production-Side Inputs} The internal and external agglomeration elasticities are
\begin{equation*}
\hat{\alpha}^I=0.199 \text{~~and~~} \hat{\alpha}^E=0.074
\end{equation*}
in the baseline counterfactual analyses. They are from our preferred IV specification in Column 5 of Table \ref{table:result_IV_main} using publications as the research output. As a robustness check, we also use elasticities in Table \ref{table:result_IV_cluster}, where cluster size is measured by paper counts, rather than author counts.

\paragraph{Knowledge-Stock Parameters} 
We set $s^N = 0.2$, near the upper end of estimates of the elasticity of aggregate productivity with respect to research activity in the macro and cross-country literature, roughly ranging from 0.08--0.23 \citep{coe_helpman1995,van_reenen2019,mertens2024}. Because $s^N$ scales the contributions of both knowledge access and knowledge production, this choice does not mechanically favor either side of the trade-off.

We set the depreciation rate of the frontier knowledge stock $\delta = 0.15$, the conventional value in the R\&D-stock literature \citep{hall2010measuring}.\footnote{Because $\delta$ enters the steady-state knowledge stock as an additive constant in $\ln Q$, it cancels from the log change in output that our reallocation exercise computes, so our results are insensitive to its value.}
We set the dependence of new research on the existing knowledge stock $\eta = 0.5$, a value in the range implied by \citet{jones1995}.\footnote{There is a debate about the magnitude of $\eta$ in the growth literature. More details are provided in footnote \ref{footnote:eta}. }

\paragraph{Access-Side Calibration}
The share of local knowledge access attributable to nearby researchers is
\begin{equation*}
\lambda_j = \frac{(N_j/M_j)^\theta}{S_0 + (N_j/M_j)^\theta}.
\end{equation*}
Because the non-local access component $S_0$ is unobserved, we do not attempt to recover the location-specific values of access share $\lambda_j$. Instead, we approximate $\lambda_j$ using common values $\bar{\lambda}=0.2, 0.5,$ or $0.7$, assuming that the local and non-local channels each account for a fixed share of knowledge access everywhere. For each value, we solve for the break-even value $\theta^*$ and calculate the corresponding access elasticity $s^N \theta^*\bar{\lambda}$. This is the object that we compare with estimates of local research spillovers in the literature. 

\subsection{Counterfactual Reallocation Procedure}
\label{section:reallocation_rule}
We conduct each counterfactual separately within each field $f$, reallocating researchers from MSAs with more researchers than their population share implies to MSAs with fewer. We then aggregate the resulting outcomes across fields. To simplify notation, we suppress the field subscript below.
\begin{enumerate}
    \item {\bf Measure researcher surpluses and deficits.} Let $\bar r = \sum_j N_j / \sum_j M_j$ denote the aggregate researcher-to-population ratio. The researcher allocation proportional to population is $N_j^{\text{eq}} = \bar r\, M_j$. Define MSA $j$'s researcher \emph{gap} as 
    \begin{equation*}
    g_j = N_j^{\text{eq}} - N_j.
    \end{equation*}
    An MSA has a researcher surplus if $g_j < 0$ and a deficit if $g_j > 0$.
    \item {\bf Choose the reallocation intensity.} Let $\kappa \in [0,1]$ denote the fraction of the gap to be closed, i.e., the reallocation intensity. The counterfactual number of researchers in MSA $j$ is
  \[
    N_j(\kappa) = N_j + \kappa\, g_j = (1-\kappa)\,N_j + \kappa\,\bar r\, M_j .
  \]
  At $\kappa = 0$, the observed allocation is unchanged. At $\kappa = 1$, researchers are allocated in proportion to population. 
 \item {\bf Select and assign movers.} From each MSA with a researcher surplus, we randomly select researchers without replacement until the required outflow, $N_j-N_j(\kappa)$, is reached. We assign each mover to a destination MSA with a researcher deficit with probability proportional to its remaining researcher deficit, $g_j\,\kappa$. Within a destination MSA, we assign movers to institutions in proportion to their sizes.
  \item {\bf Compute counterfactual outcomes.} We calculate the counterfactual access $\ln B$ and the counterfactual steady-state knowledge stock $\ln Q$, and the implied changes $\Delta \ln B$ and  $\Delta \ln Q$. Because individual assignments are random, we repeat each exercise 50 times and report the average outcome.
\end{enumerate}

\subsection{Counterfactual Knowledge Access}
A counterfactual reallocation affects population's access to knowledge through each MSA's counterfactual researcher count. In each counterfactual, we hold fixed the non-research population $M_j$, baseline non-research productivity $S_{0j}$, and the calibrated non-local access component $S_0$. Only the researcher count changes from its observed value $N_j$ to its counterfactual value $N_j(\kappa)$. Under reallocation intensity $\kappa$, the counterfactual access term is
\[
\widehat{B}(\kappa) = \sum_{j=1}^{J} M_j\, S_{0j}^{1-s^N}\,
\left[( S_0 + \left( \frac{N_{j}(\kappa)}{M_j}\right)^{\theta} \right]^{s^N}.
\]

\paragraph{Non-Researcher Migration} 
The baseline analysis holds the non-research population fixed across MSAs. As a robustness exercise, we allow non-researchers to respond to counterfactual wage changes using the location-choice model in Appendix \ref{equalizing_access}. Researcher reallocation raises knowledge access and wages in destination MSAs, which can attract non-researchers toward locations with lower baseline productivity. This response only modestly offsets the direct access gain and does not materially alter our conclusion. The spatial misallocation due to migration is similar to the efficiency cost of place-based policies in general.

\subsection{Counterfactual Knowledge Production}
Throughout this subsection, we retain the log notation for log research outcome $\ln F$ used in the model, but the empirical calculations use the IHS-transformed $F$ instead, consistent with our estimation. 
\paragraph{Non-Movers} Based on Equation \ref{eq:aggregate_production}, researchers who are not relocated are indirectly affected because a reallocation changes the size of their cluster. Thus, for a researcher who remains at institution $j$, her output changes through an agglomeration adjustment:
\begin{equation*}
\Delta \ln \hat{F}^{non-mover}_{ijft} = \hat{\alpha}^I \left [ \ln(N^{I,cf}_{jft}) - \ln(N^{I,0}_{jft}) \right] + \hat{\alpha}^E \left [ \ln(N^{E,cf}_{jft}) - \ln(N^{E,0}_{jft}) \right] 
\end{equation*}
where $N^{I,cf}_{jft}$ and $N^{E,cf}_{jft}$ are the counterfactual internal and external cluster sizes, and $N^{I,0}_{jft}$ and $N^{E,0}_{jft}$ are the corresponding initial values.\footnote{\label{footnote:log2}As noted in footnote \ref{footnote:log}, we use $\ln (N+1)$ to estimate agglomeration elasticities, so in counterfactual analyses, we apply the same transformation to both observed and counterfactual cluster sizes. Because this transformation is steep near zero, predictions can be sensitive when a counterfactual changes very small clusters by only a few researchers. Thus, for the exercise in Section \ref{counterfactual_all} that involves very small clusters, we apply a larger value of $c$ to make the transformation less curved near zero, reducing the influence of changes involving only a few researchers. Specifically, we re-estimate the elasticities using $c=5$ and apply the same transformation in the counterfactual. Results are similar if we use $c=10$.}

\paragraph{Movers} For a researcher who is relocated from institution $j$ to $j'$, the change in her output reflects the change in the full location effect between the destination and origin. This is empirically approximated using the difference in the predicted location effect, augmented by the change in agglomeration effects in the destination induced by the reallocation:
\begin{equation*}
\Delta \ln \hat{F}^{mover}_{ijft} = \hat{d}^{pred} _{j'ft} - \hat{d}^{pred} _{jft} + \hat{\alpha}^I \left [ \ln(N^{I,cf}_{j'ft}) - \ln(N^{I,0}_{j'ft}) \right] + \hat{\alpha}^E \left [ \ln(N^{E,cf}_{j'ft}) - \ln(N^{E,0}_{j'ft}) \right]
\end{equation*}
In sum, non-movers are affected only through changes in cluster size, while movers experience a change in the full location effect. Individual fixed effects, academic-age effects, and field-year effects are held fixed.

We add the corresponding change to each researcher's initial IHS research output and apply the standard hyperbolic sine function to convert the {\it counterfactual} IHS research output back to levels, $\hat{F}^{cf}_{ijft}$. After aggregating counterfactual {\it flow} research output $\hat{F}^{cf}_{ijft}$ across researchers, the change in the steady-state frontier knowledge {\it stock} is
\begin{equation*}
    \Delta \ln \widehat{Q}=\ln \left(\frac{\sum_i \widehat{F}_{ijft}}{\delta}\right)^{\frac{1}{1-\eta}} - \ln \left(\frac{\sum_i F_{ijft}}{\delta}\right)^{\frac{1}{1-\eta}}.
\end{equation*}

After completing the reallocation separately within each field, we aggregate counterfactual research output across all researchers and fields to obtain the national change in the knowledge stock.
For each value of $\theta$, we compute $\Delta \ln \hat{Y} = \Delta \ln \hat{B}(\theta) + s^N \Delta \ln \hat{Q}$. We identify $\theta^{\ast}$ as the smallest value of $\theta$ for which $\Delta\ln \hat{Y} \geq 0$ in each simulation, and report the implied access elasticity $s^N \theta^{\ast} \bar{\lambda}$.

\subsection{Counterfactual Results}
Tables \ref{table:case_studies}--\ref{tab:alternatives} present the counterfactual results. The left block reports changes in {\it flow} research output for three groups: researchers who remain in the origin (``Origin Stayers"), researchers already located in the destination (``Dest. Incumb."), and researchers who move (``Movers"). The group-specific percentages use each group’s own baseline output as the denominator and therefore do not add directly. The ``Net Change" column reports the corresponding aggregate change. The final two columns convert this production loss into the break-even $\theta^{\ast}$ and the implied break-even knowledge-access elasticity $s^N \theta^{\ast} \bar{\lambda}$.

\subsubsection{Bilateral Case Studies}
We begin with two transparent bilateral counterfactual reallocations: Boston to Dallas–Fort Worth and San Francisco to Las Vegas. These exercises move researchers from established research hubs to large or rapidly growing MSAs with relatively few researchers per capita. In each case, we set $\kappa=0.1$.\footnote{Under $\kappa=0.1$, Boston and San Francisco would lose around 5\% and 3\% of their researchers, while Dallas–Fort Worth and Las Vegas would gain approximately 31\% and 70\%, respectively.} 

Table \ref{table:case_studies} shows that shrinking the origin lowers output among stayers at institutions in the origin MSAs, expanding the destination raises output among incumbents in the destination MSAs, and movers lose output because destination institutions have lower location effects. The declines for movers are moderated because they carry their individual effects. Moreover, the convergence in cluster sizes narrows the gap in location effects between the destination and the origin. On net, aggregate knowledge production flow (i.e., publications) would fall by around 0.6\% in each case.

The implied break-even access elasticity ranges from 0.0176 to 0.0193, well below existing estimates of localized university spillovers. \citet{kantor2014knowledge}, for example, estimate an elasticity of local noneducation income with respect to university spending of 0.08, while \citet{jaff1989academic} report similarly large or larger elasticities of local patenting with respect to university R\&D. These estimates do not correspond exactly to our parameter $s^N \theta \bar{\lambda}$, which measures the elasticity of local non-research productivity with respect to access to researchers. Nevertheless, they provide useful benchmarks for its plausible magnitude. Unless $s^N \theta \bar{\lambda}$ is substantially smaller than these estimates, the access gains from the bilateral reallocations would outweigh the associated loss in knowledge production.

\subsubsection{Reallocation Among Large MSAs}

We next extend the bilateral exercises to the 50 most populous MSAs. Within this group, we move researchers from the ten MSAs with the most researchers per capita to the ten with the fewest, again setting $\kappa=0.1$. Restricting the exercise to large MSAs avoids reallocations into very small research clusters, where counterfactual predictions are more sensitive to functional-form assumptions.\footnote{See footnotes \ref{footnote:log} and \ref{footnote:log2} for a detailed discussion of this sensitivity.}

Panel A of Table \ref{table:systematic} shows that reallocation among the 50 largest MSAs lowers aggregate flow research output by only 0.32\%. The break-even access elasticity is about 0.010 and rises modestly to 0.012 when non-researcher migration is allowed.\footnote{Appendix \ref{app:nonresearcher} provides more details on how we simulate non-researchers' migration.} Thus, the conclusion from the bilateral exercises extends to a broader set of large MSAs: only a modest access gain is required to offset the loss in knowledge production.

\subsubsection{Reallocation Across All MSAs}
\label{counterfactual_all}

We finally extend the exercise to all MSAs that host research institutions. This broader reallocation includes many very small research clusters, for which the log transformation is particularly sensitive to changes involving only a few researchers. To limit the influence of this near-zero curvature, we use transformations that are less sensitive at small cluster sizes.\footnote{As discussed in footnote \ref{footnote:log2}, the baseline specification uses $\ln(N+1)$. For this exercise, we re-estimate the agglomeration coefficients using $\ln(N+5)$ and apply the same transformation to observed and counterfactual cluster sizes. Although the coefficient estimates vary with the additive constant (as shown in Table \ref{table:result_IV_zero}), the counterfactual results are stable when the same transformation is used in estimation and simulation. Counterfactual results  using $\ln(N+10)$ are omitted for brevity and available upon request. }

Panel B of Table \ref{table:systematic} shows that extending the exercise to all MSAs raises the production loss to 1.14\%, largely because many movers are assigned to very small institutions with lower location effects. The break-even access elasticity rises to 0.039--0.043, still below the 0.08 benchmark but by a narrower margin. This contrast suggests that reallocation toward large underserved MSAs could be less costly than broad dispersion into small, isolated research clusters.

\subsubsection{Alternative Specifications}

Table \ref{tab:alternatives} considers four alternatives to the baseline large-MSA reallocation, all under $\bar{\lambda}=0.5.$

\paragraph{Alternative Measures of Knowledge Access} Panel A considers the same reallocation as Table \ref{table:systematic} Panel A but allows local access to depend on publications or citations per capita rather than researcher counts. The break-even elasticity falls from 0.0096 to 0.0068 and 0.0034, respectively, because movers from research-intensive origins bring more research output and impact to destinations than their headcount alone captures.

\paragraph{Agglomeration Through Research Activity} Panel B measures cluster size by publication output (a quality-weighted scale) rather than researcher counts and applies the corresponding agglomeration estimates from Table \ref{table:result_IV_cluster}. Because movers are relatively productive, their relocation changes publication-based cluster size more than headcount-based size, narrowing the origin–destination gap in location productivity. The net decline in flow research output is only 0.10\%, reducing the break-even elasticities to 0.0008--0.0031.

\paragraph{College-Educated Population} If college-educated workers benefit more from frontier knowledge, the baseline exercise may overstate access gains because these workers already co-locate more closely with researchers. Panel C therefore reallocates researchers according to the spatial distribution of college graduates and measures access relative to that population. The break-even elasticities rise only modestly and remain between 0.0042 and 0.0117 across the alternative access measures.

\paragraph{STEM Researchers} If economically relevant local spillovers are concentrated in technical fields, an exercise covering all fields may overstate the access gains. Panel D restricts the reallocation to STEM researchers---those in the life sciences, physical sciences, and engineering---and again uses the college-educated population as the reallocation target and access base. The results are nearly identical to Panel C, consistent with STEM fields accounting for most research volume in our data.

In conclusion, across all alternatives, the break-even access elasticities remain below the 0.08 benchmark, reinforcing the conclusion that a moderate reallocation toward large underserved MSAs can plausibly raise aggregate output.

\section{Conclusion}
\label{conclusion}

This paper studies the geography of research in the U.S. and its implications for the broader economy. Research activity generates highly localized spillovers, providing a rationale for dispersing researchers toward population and economic activity so that a broader segment of the economy can access frontier knowledge. Yet dispersion may move researchers away from productive institutions and shrink the size of productive research clusters, reducing the aggregate stock of knowledge. We formally characterize and quantify this trade-off between the access benefit of dispersion and the production benefit of concentration.

Using detailed bibliographic data, we first show that research activity is highly concentrated across the U.S. and has become increasingly misaligned with population. As population shifted toward the South and West, research activity remained anchored in long-established institutions, making geographic access to research increasingly unequal across metropolitan areas. 

We then estimate an individual-level research production function to quantify the potential production cost of dispersion. Location effects account for meaningful differences in research output across researchers, suggesting that researchers’ productivity can change when they move between research environments. Our IV estimates also provide robust evidence that own-institution cluster size raises research output, while evidence of external-cluster effects is weaker. These estimates imply that researcher reallocation can affect knowledge production both directly through movers’ location productivity and indirectly through changes in the cluster sizes faced by researchers who remain.

Finally, we combine these estimates on the research-production side in a spatial framework to evaluate counterfactual reallocations and calculate the knowledge-access elasticity required for access gains to offset production losses. This provides a first quantitative assessment of whether alternative spatial allocations of researchers can generate net gains in aggregate output. The results suggest that marginally reallocating researchers toward large metropolitan areas with relatively limited research capacity would lead to modest production losses and require relatively small local access gains to raise aggregate output. 

\clearpage

\begin{figure}[!h]
     \centering
        \caption{County-Level Access to Research Papers, 2014--2019} 
     \begin{subfigure}[b]{0.9\linewidth}
         \centering
         \includegraphics[width=\linewidth]{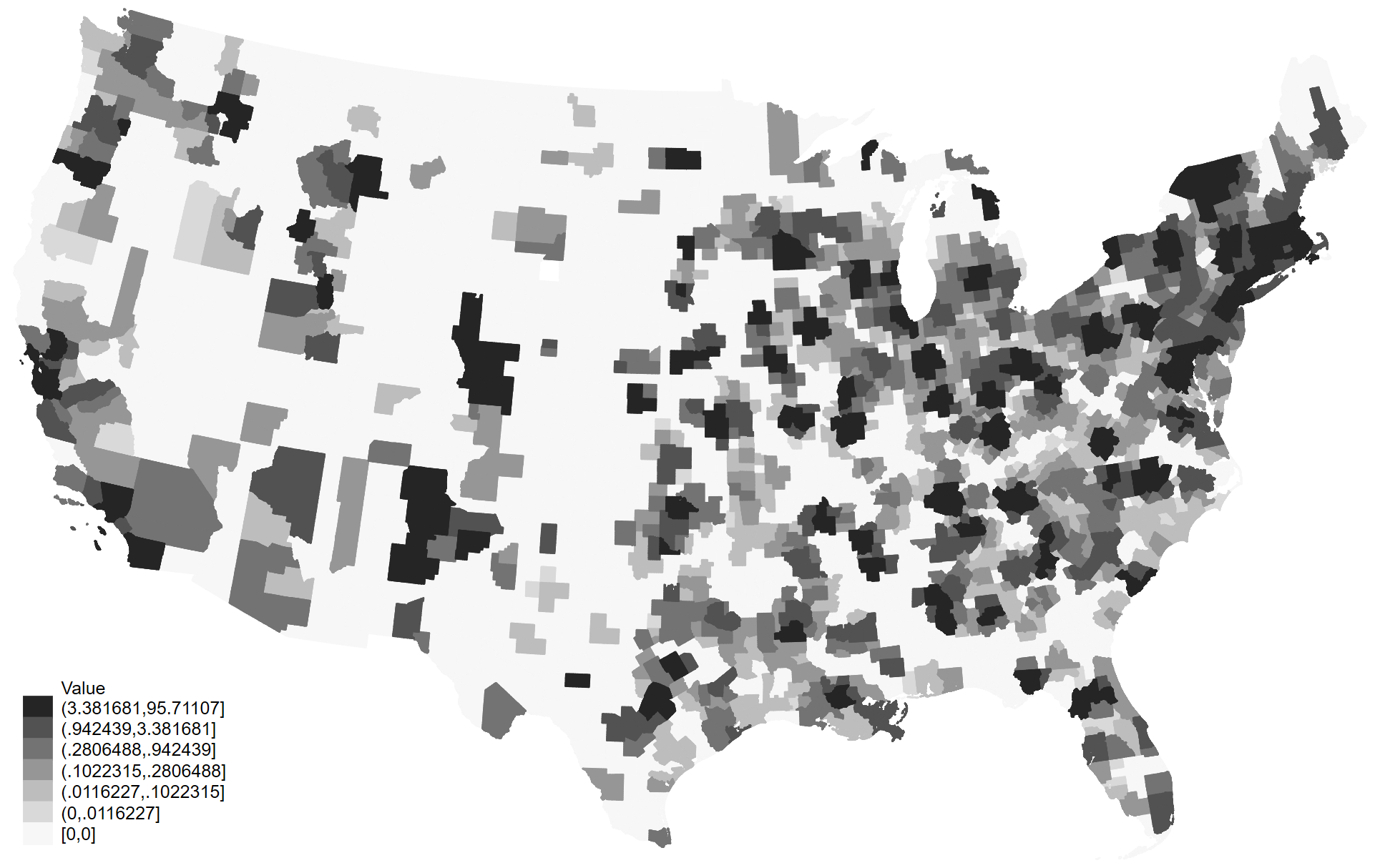}
     \end{subfigure}
        \label{fig:map_access}
\begin{minipage}{\textwidth}
\footnotesize{{\it Notes:} The map shows county-level access to nearby research. For each census tract, we calculate the number of research papers produced between 2014 and 2019 by institutions within 50 miles per 1,000 residents living within the same radius. County-level values are population-weighted averages of the tract-level measures.
} 
\end{minipage}
\end{figure}

\begin{figure}[!h]
     \centering
        \caption{Trends in Access to Research Papers: Selected MSA Case Studies} 
     \begin{subfigure}[b]{0.7\linewidth}
         \centering
         \includegraphics[width=\linewidth]{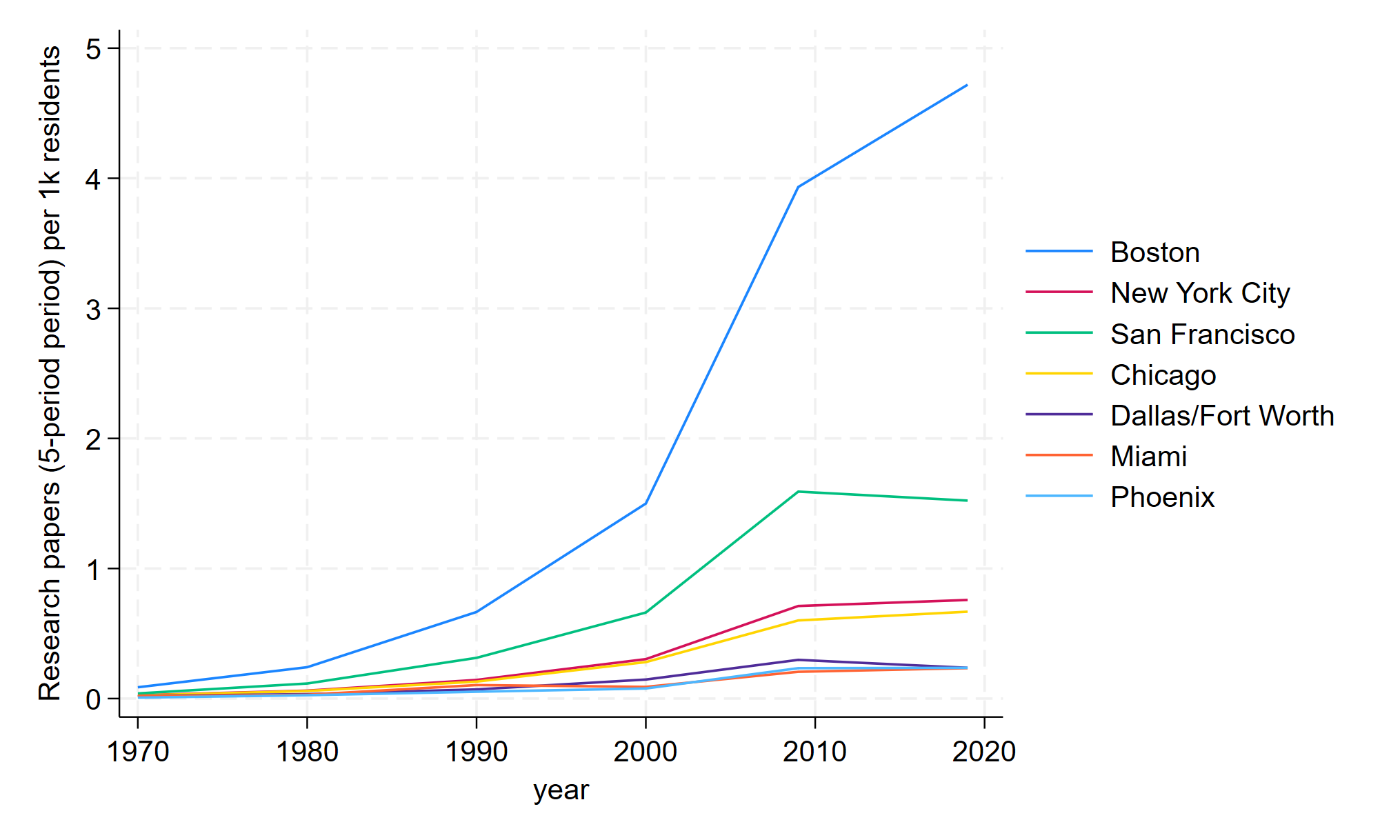}
      \caption{\centering Research Papers per 1,000 Residents}
     \end{subfigure}
   \\
        \begin{subfigure}[b]{0.7\linewidth}
         \centering
         \includegraphics[width=\linewidth]{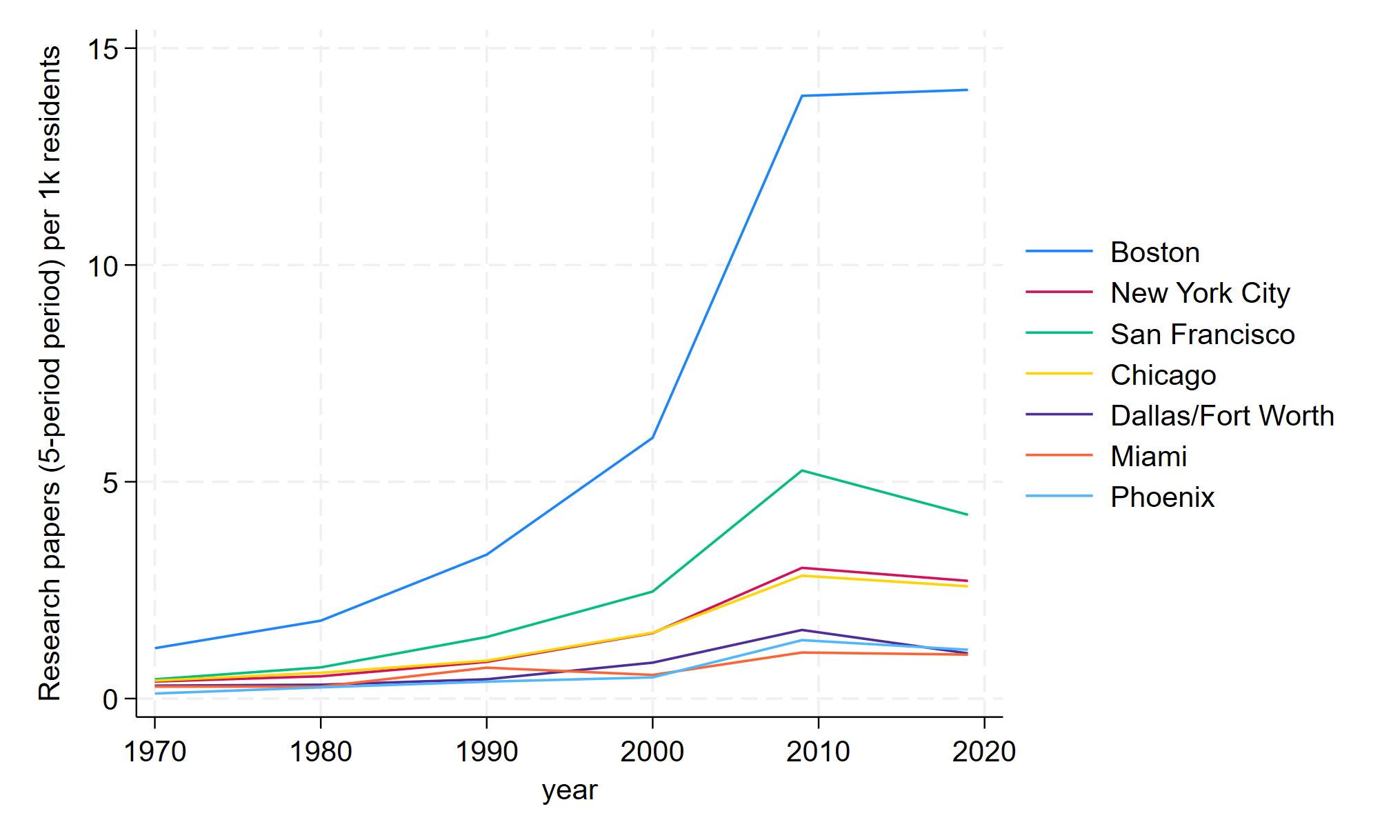}
     \caption{\centering Research Papers per 1,000 College Graduates}
     \end{subfigure}
\begin{minipage}{\textwidth}
\footnotesize{{\it Notes:} Panel (a) reports research papers per 1,000 residents, and Panel (b) reports research papers per 1,000 college graduates. For each MSA, publication counts include papers produced by authors affiliated with institutions in that MSA and are divided by the corresponding population measure. Publication counts are aggregated over 1965–1970, 1975–1980, 1985–1990, 1995–2000, 2005–2010, and 2014–2019. Population data come from the 1970, 1980, 1990, and 2000 Censuses and the 2005–2009 and 2015–2019 ACS.}
\end{minipage}
        \label{fig:case_studies}
\end{figure}

\clearpage

\begin{figure}[!h]
     \centering
        \caption{Rising Inequality in Access to Research Papers: 1970--2019} 
     \begin{subfigure}[b]{0.7\linewidth}
         \centering
         \includegraphics[width=\linewidth]{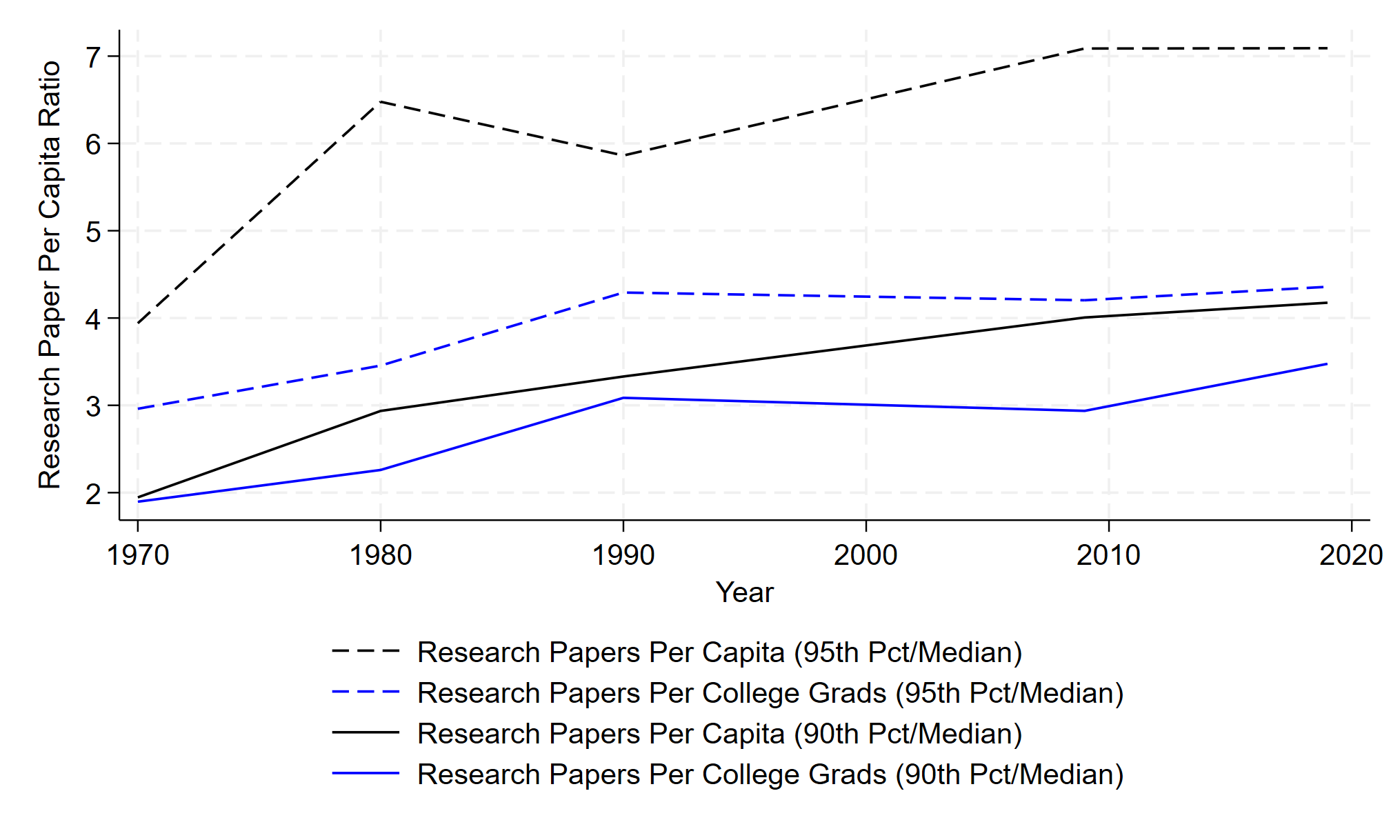}
      \caption{\centering 90/50 and 95/50 Ratios}
             \label{fig:ratio90_50}
     \end{subfigure}
   \\
        \begin{subfigure}[b]{0.7\linewidth}
         \centering
         \includegraphics[width=\linewidth]{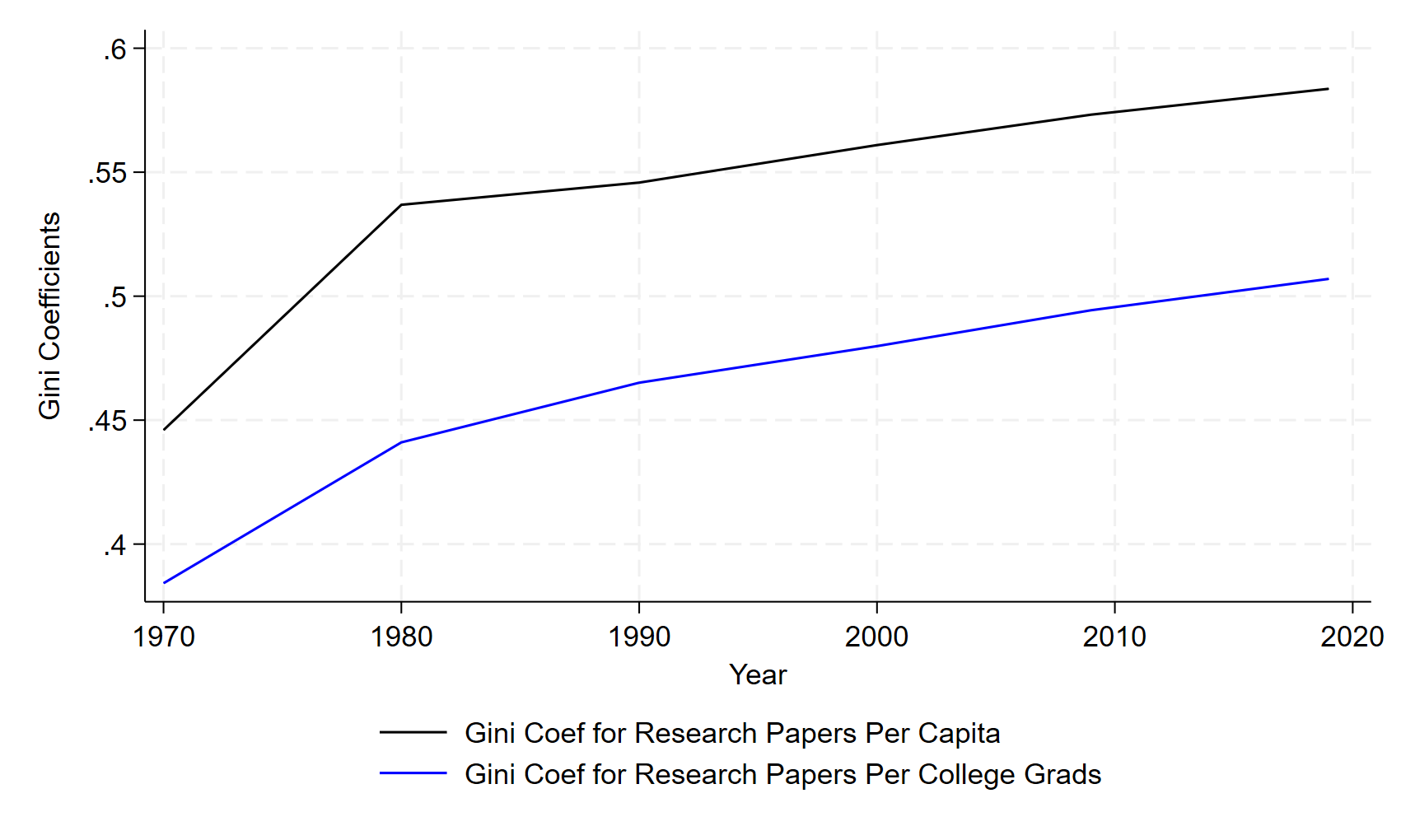}
     \caption{\centering Gini Coefficients}
        \label{fig:gini}
     \end{subfigure}
\begin{minipage}{\textwidth}
\footnotesize{Notes: Panel (a) plots the 90/50 and 95/50 ratios of access to research papers across MSAs from 1970 to 2019, and Panel (b) plots the corresponding Gini coefficients. The black series use research papers per resident, while the blue series use research papers per college graduate.} 
\end{minipage}
        \label{fig:trends}
\end{figure}

\clearpage

\begin{figure}[!h]
     \centering
    \caption{Event Study: Changes in Research Output Around a Move} 
     \begin{subfigure}[b]{0.47\linewidth}
         \centering
         \includegraphics[width=\linewidth]{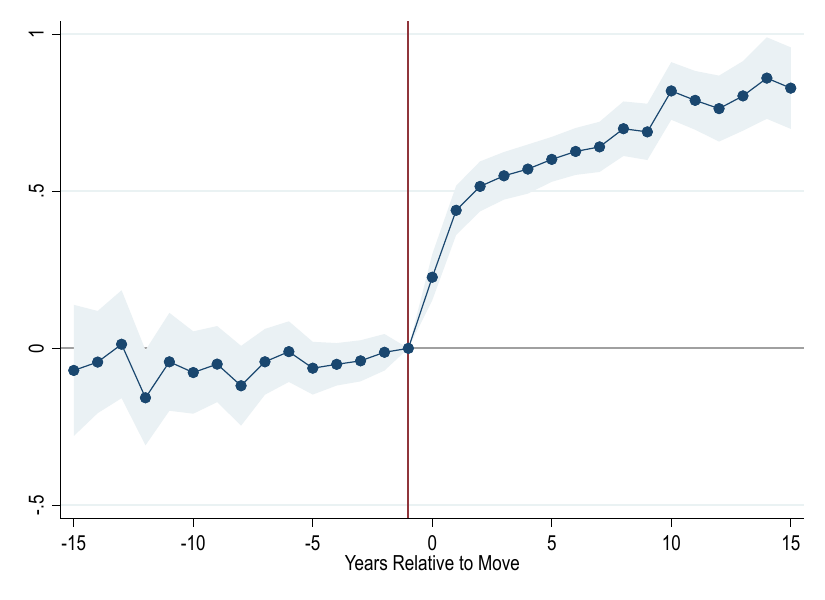}
         \caption{\centering Publications }
         \label{fig:event_location_publication}
     \end{subfigure}
     \qquad
    \begin{subfigure}[b]{0.47\linewidth}
         \centering
         \includegraphics[width=\linewidth]{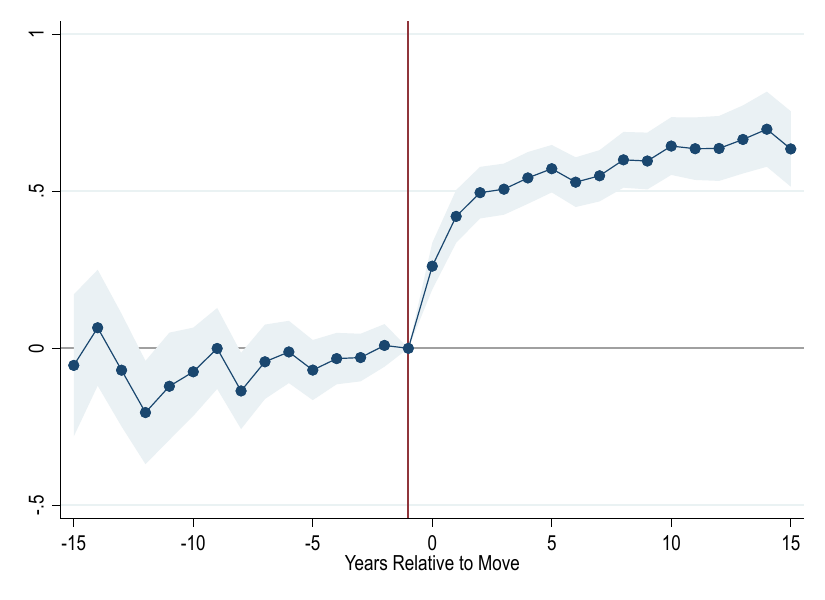}
         \caption{\centering Citations}
         \label{fig:event_location_citation}
     \end{subfigure}
    \\
     \vspace{0.25cm}    
     \begin{subfigure}[b]{0.47\linewidth}
         \centering
         \includegraphics[width=\linewidth]{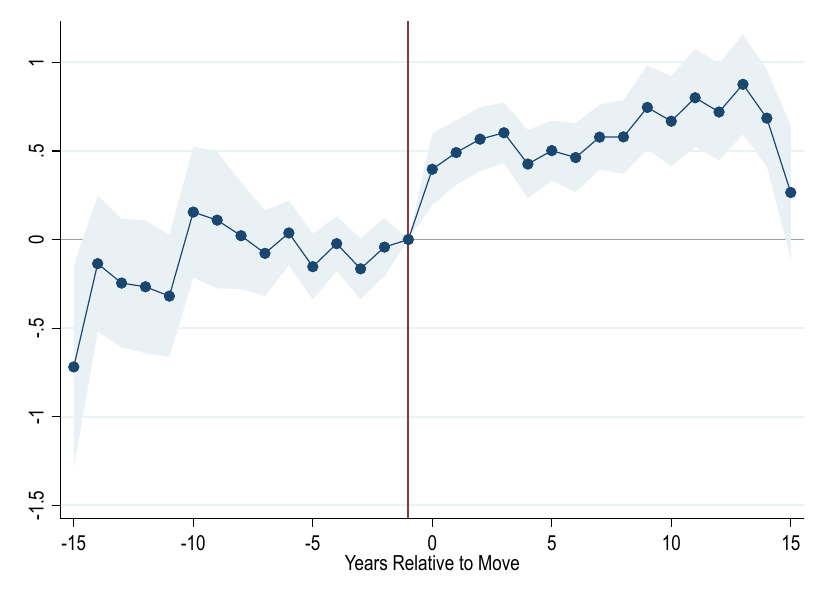}
         \caption{\centering Patent Citations }
         \label{fig:event_location_patent}
     \end{subfigure}
    \vspace{0.25cm}
    \label{fig:event_location}
    \vspace{0.25cm}
\begin{minipage}{\textwidth}
\footnotesize{Notes: Each panel plots the estimated coefficients $\delta_s$ for $-15\leq s \leq 15$ ($s \neq -1$) from Equation \ref{eq:event} using the sample of researchers who changed institutions once. The dependent variable is IHS publications (Panel (a)), IHS citations (Panel (b)), and IHS patent citations (Panel (c)). The $x$-axis denotes years relative to the move. The relative year $-1$ is omitted and normalized to zero, and the vertical line marks the year before the move. The shaded bands show 95\% confidence intervals based on standard errors clustered at the institution level.}
\end{minipage}
\end{figure}

\clearpage

\begin{table}[h]
\centering
\captionsetup{justification=centering}
\caption{Additive Decomposition of Research Outcomes Across Researchers: \\ By Researcher Productivity Group}
\begin{tabular}{
    >{\raggedright\arraybackslash}p{3.5cm}
    *{3}{>{\centering\arraybackslash}p{2.4cm}}
}\toprule
                     & \multicolumn{3}{c}{Researcher Productivity Groups} \\ \cmidrule(lr){2-4}
Top vs. Bottom       & 50\%        & 25\%        & 10\%        \\ \midrule
\multicolumn{4}{l}{{\it Panel A: Publications}}                \\
\multicolumn{4}{l}{Differences in IHS Publications}            \\ 
\quad Overall        & 1.15        & 1.74        & 2.38        \\
\quad Individuals    & 0.77        & 1.15        & 1.50        \\
\quad Locations      & 0.09        & 0.13        & 0.18        \\
\quad Unexplained    & 0.29        & 0.46        & 0.69        \\ \addlinespace
\multicolumn{4}{l}{Share of Difference Attributable to}        \\ 
\quad Individuals    & 0.67        & 0.66        & 0.63        \\
\quad Locations      & 0.08        & 0.08        & 0.08        \\
\quad Unexplained    & 0.26        & 0.27        & 0.29        \\ \midrule
\multicolumn{4}{l}{{\it Panel B: Citations}}                   \\
\multicolumn{4}{l}{Differences in IHS Citations}               \\ 
\quad Overall        & 2.68        & 3.97        & 5.17        \\
\quad Individuals    & 2.03        & 2.92        & 3.68        \\
\quad Locations      & 0.19        & 0.28        & 0.36        \\
\quad Unexplained    & 0.48        & 0.78        & 1.14        \\ \addlinespace
\multicolumn{4}{l}{Share of Difference Attributable to}        \\ 
\quad Individuals    & 0.76        & 0.74        & 0.71        \\
\quad Locations      & 0.07        & 0.07        & 0.07        \\
\quad Unexplained    & 0.18        & 0.20        & 0.22        \\ \midrule
\multicolumn{4}{l}{{\it Panel C: Patent Citations}}            \\
\multicolumn{4}{l}{Differences in IHS Patent Citations}        \\ 
\quad Overall        & 0.67        & 0.67        & 1.01        \\
\quad Individuals    & 0.51        & 0.51        & 0.75        \\
\quad Locations      & 0.07        & 0.07        & 0.05        \\
\quad Unexplained    & 0.11        & 0.11        & 0.21        \\ \addlinespace
\multicolumn{4}{l}{Share of Difference Attributable to}        \\ 
\quad Individuals    & 0.75        & 0.75        & 0.75        \\
\quad Locations      & 0.10        & 0.10        & 0.05        \\
\quad Unexplained    & 0.16        & 0.16        & 0.21        \\
\bottomrule  \noalign{\vskip 0.1in}
\multicolumn{4}{l}{%
\begin{minipage}{12.5cm}%
\footnotesize{\textit{Note:} This table reports the additive decomposition of differences in IHS publications, citations, and patent citations across researcher productivity groups during 2010--2015. Researchers are ranked within their field by their average research output for the corresponding outcome. The 50\%, 25\%, and 10\% columns compare researchers above and below the median, in the top and bottom quartiles, and in the top and bottom deciles, respectively. The researcher component includes author fixed effects and field-specific academic-age effects. The location component is predicted using lagged observable characteristics of the institution--field and external MSA--field research environments. The unexplained component includes differences in residuals. The field-year component is not reported because its contribution to the overall difference is negligible. Component shares are calculated as the difference in each reported component divided by the overall difference in the corresponding outcome; because the field-year component is omitted from the table, the reported shares may not sum exactly to one.}
\end{minipage}}%
\end{tabular}
\label{table:decomp_output}
\end{table}

\clearpage

\begin{table}[h]
\centering
\captionsetup{justification=centering}
\caption{Additive Decomposition of Research Outcomes Across Researchers: \\ By Institution and MSA Size Group}
\begin{tabular}{
    >{\raggedright\arraybackslash}p{3.2cm}
    *{6}{>{\centering\arraybackslash}p{1.6cm}}
}\toprule                         & \multicolumn{3}{c}{Institution Size Group}  & \multicolumn{3}{c}{MSA Size Group} \\ \cmidrule(lr){2-4} \cmidrule(lr){5-7}
Top vs. Bottom           & 50\%             & 25\%            & 10\%         & 50\%       & 25\%      & 10\%      \\ \midrule
\multicolumn{7}{l}{{\it Panel A: Publications}}                                                                   \\
\multicolumn{7}{l}{Differences in IHS Publications}                                                               \\
\quad Overall            & 0.40             & 0.62            & 0.77         & 0.11       & 0.32      & 0.56      \\
\quad Individuals        & 0.10             & 0.23            & 0.33         & 0.01       & 0.10      & 0.27      \\
\quad Locations          & 0.26             & 0.32            & 0.34         & 0.09       & 0.18      & 0.21      \\
\quad Unexplained        & 0.03             & 0.07            & 0.09         & 0.01       & 0.03      & 0.07      \\ \addlinespace
\multicolumn{7}{l}{Share of Difference Attributable to}                                                           \\
\quad Individuals        & 0.25             & 0.36            & 0.42         & 0.09       & 0.32      & 0.48      \\
\quad Locations          & 0.66             & 0.51            & 0.44         & 0.80       & 0.57      & 0.38      \\
\quad Unexplained        & 0.07             & 0.10            & 0.12         & 0.06       & 0.08      & 0.13      \\ \midrule
\multicolumn{7}{l}{{\it Panel B: Citations}}                                                                      \\
\multicolumn{7}{l}{Differences in IHS Citations}                                                                  \\
\quad Overall            & 0.95             & 1.45            & 1.82         & 0.31       & 0.80      & 1.30      \\
\quad Individuals        & 0.42             & 0.77            & 1.07         & 0.15       & 0.43      & 0.84      \\
\quad Locations          & 0.53             & 0.65            & 0.71         & 0.20       & 0.39      & 0.44      \\
\quad Unexplained        & -0.01            & 0.05            & 0.07         & -0.03      & -0.02     & 0.06      \\ \addlinespace
\multicolumn{7}{l}{Share of Difference Attributable to}                                                           \\
\quad Individuals        & 0.45             & 0.53            & 0.58         & 0.47       & 0.55      & 0.64      \\
\quad Locations          & 0.56             & 0.45            & 0.39         & 0.64       & 0.49      & 0.34      \\
\quad Unexplained        & -0.01            & 0.03            & 0.04         & -0.10      & -0.02     & 0.04      \\ \midrule
\multicolumn{7}{l}{{\it Panel C: Patent Citations}}                                                               \\
\multicolumn{7}{l}{Differences in IHS Patent Citations}                                                           \\
\quad Overall            & 0.10             & 0.14            & 0.17         & 0.07       & 0.11      & 0.13      \\
\quad Individuals        & 0.02             & 0.04            & 0.06         & 0.02       & 0.03      & 0.05      \\
\quad Locations          & 0.09             & 0.12            & 0.14         & 0.06       & 0.10      & 0.11      \\
\quad Unexplained        & -0.01            & -0.02           & -0.02        & -0.01      & -0.02     & -0.02     \\ \addlinespace
\multicolumn{7}{l}{Share of Difference Attributable to}                                                           \\
\quad Individuals        & 0.23             & 0.31            & 0.35         & 0.28       & 0.32      & 0.40      \\
\quad Locations          & 0.94             & 0.86            & 0.82         & 0.92       & 0.90      & 0.83      \\
\quad Unexplained        & -0.13            & -0.11           & -0.10        & -0.13      & -0.14     & -0.13     \\
\bottomrule  \noalign{\vskip 0.1in}
\multicolumn{7}{l}{%
\begin{minipage}{15.3cm}%
\footnotesize{\textit{Note:} This table reports the additive decomposition of differences in IHS publications, citations, and patent citations across researchers by institution- or MSA-size group during 2010--2015. Institutions and MSAs are ranked separately within each field by their number of active researchers. The 50\%, 25\%, and 10\% columns compare researchers at institutions or MSAs above and below the median, in the top and bottom quartiles, and in the top and bottom deciles, respectively. The researcher component includes author fixed effects and field-specific academic-age effects. The location component is predicted using lagged observable characteristics of the institution-field and external MSA-field research environments. The field-year component is not reported because its contribution to the overall difference is negligible. Component shares are calculated as the difference in each reported component divided by the overall difference in the corresponding outcome; because the field-year component is omitted from the table, the reported shares may not sum exactly to one.}
\end{minipage}}%
\end{tabular}
\label{table:decomp_size}
\end{table}

\begin{table}[h]
\centering
\captionsetup{justification=centering}
\caption{First-Stage Estimates for the Bartik Instruments}
\begin{tabular}{lcccccc}
\toprule
                               & \multicolumn{3}{c}{Ln(Institution Size)} & \multicolumn{3}{c}{Ln(External Size)} \\ \cmidrule(lr){2-4} \cmidrule(lr){5-7}
                               & (1)          & (2)         & (3)         & (4)         & (5)        & (6)        \\ \midrule
$\ln \widehat{N}^{I,Bartik}$   & 0.851***     & 0.816***    & 0.675***    & 0.604***    & 0.590***   & 0.623***   \\
                               & (0.162)      & (0.161)     & (0.145)     & (0.208)     & (0.207)    & (0.174)    \\
$\ln \widehat{N}^{E,Bartik}$   & -0.026       & -0.034      & -0.000      & 0.387***    & 0.385***   & 0.190***   \\
                               & (0.024)      & (0.023)     & (0.019)     & (0.028)     & (0.028)    & (0.036)    \\
Author Control                 &              & 0.175***    & 0.140***    &             & 0.003      & -0.023     \\
                               &              & (0.019)     & (0.018)     &             & (0.038)    & (0.038)    \\
Coauthor Control               &              & 0.004       & 0.005       &             & 0.036**     & 0.030*    \\
                               &              & (0.016)     & (0.012)     &             & (0.017)    & (0.016)    \\ \midrule
Observations                   & 798,548      & 798,548     & 797,434     & 798,548     & 798,548    & 797,434    \\
R-squared                      & 0.964        & 0.964       & 0.967       & 0.937       & 0.937      & 0.940      \\
Baseline FEs                   & Yes          & Yes         & Yes         & Yes         & Yes        & Yes        \\
Author/Coauthor Controls       & No           & Yes         & Yes         & No          & Yes        & Yes        \\
MSA $\times$ Year FE           & No           & No          & Yes         & No          & No         & Yes        \\
\bottomrule  \noalign{\vskip 0.1in}
\multicolumn{7}{l}{%
\begin{minipage}{15.6cm}%
\footnotesize{\textit{Note:} This table reports first-stage estimates for the two endogenous variables in Equation \ref{eq:agg_reg}. The dependent variable is log own-institution size in Columns 1--3 and log external-cluster size in Columns 4--6. The excluded instruments are the Bartik-predicted own-institution size and Bartik-predicted external-cluster size, constructed from base-period subfield shares and worldwide subfield growth as described in Equation \ref{eq:bartik}. The sample consists of author-year observations in three periods: 1995, 2005, and 2015, where each period pools the focal year and the two preceding years. Columns 1 and 4 include the baseline fixed effects from Equation \ref{eq:agg_reg}, including author fixed effects, academic-age-by-field fixed effects, field-by-year fixed effects, and affiliation-by-field fixed effects. Columns 2 and 5 add the author and coauthor direct-exposure controls. Columns 3 and 6 further add MSA-by-year fixed effects. Standard errors are in parentheses and clustered at the institution level: *** $p< 0.01$, ** $p< 0.05$, * $p< 0.1$.}
\end{minipage}}%
\end{tabular}
\label{table:result_first_stage}
\end{table}

\begin{table}[h]
\centering
\captionsetup{justification=centering}
\caption{OLS and IV Estimates of Agglomeration Elasticities}
\begin{tabular}{lcccccc}
\toprule
                               & \multicolumn{3}{c}{OLS}          & \multicolumn{3}{c}{IV}           \\ \cmidrule(lr){2-4} \cmidrule(lr){5-7}
                               & (1)      & (2)       & (3)       & (4)      & (5)       & (6)       \\ \midrule
\multicolumn{7}{l}{\textit{Panel A: Publications}}                                                   \\
Ln(Institution Size)           & 0.174*** & 0.173***  & 0.175***  & 0.234*** & 0.199***  & 0.300***  \\
                               & (0.007)  & (0.007)   & (0.007)   & (0.072)  & (0.070)   & (0.103)   \\
Ln(External Size)              & 0.003    & 0.003     & 0.002     & 0.082*** & 0.074***  & 0.037     \\
                               & (0.002)  & (0.002)   & (0.002)   & (0.020)  & (0.020)   & (0.052)   \\
Author Control                 &          & 0.001     & -0.004    &          & -0.005    & -0.021    \\
                               &          & (0.007)   & (0.006)   &          & (0.015)   & (0.017)   \\
Coauthor Control               &          & 0.096***  & 0.101***  &          & 0.091***  & 0.097***  \\
                               &          & (0.010)   & (0.010)   &          & (0.011)   & (0.011)   \\ \addlinespace
\multicolumn{7}{l}{\textit{Panel B: Citations}}                                                      \\
Ln(Institution Size)           & 0.375*** & 0.379***  & 0.384***  & 0.898*** & 0.828***  & 1.028***  \\
                               & (0.016)  & (0.017)   & (0.016)   & (0.192)  & (0.190)   & (0.271)   \\
Ln(External Size)              & 0.002    & 0.000     & -0.001    & 0.054    & 0.056     & -0.022    \\
                               & (0.004)  & (0.004)   & (0.004)   & (0.062)  & (0.060)   & (0.134)   \\
Author Control                 &          & -0.157*** & -0.158*** &          & -0.238*** & -0.248*** \\
                               &          & (0.015)   & (0.016)   &          & (0.038)   & (0.044)   \\
Coauthor Control               &          & 0.281***  & 0.296***  &          & 0.269***  & 0.288***  \\
                               &          & (0.027)   & (0.024)   &          & (0.031)   & (0.027)   \\ \addlinespace
\multicolumn{7}{l}{\textit{Panel C: Patent Citations}}                                               \\
Ln(Institution Size)           & 0.022*** & 0.022***  & 0.024***  & 0.153*** & 0.148***  & 0.192***  \\
                               & (0.004)  & (0.004)   & (0.004)   & (0.042)  & (0.043)   & (0.062)   \\
Ln(External Size)              & -0.002** & -0.002**  & -0.002**  & -0.013   & -0.011    & -0.030    \\
                               & (0.001)  & (0.001)   & (0.001)   & (0.012)  & (0.012)   & (0.029)   \\
Author Control                 &          & -0.011*** & -0.007*   &          & -0.033*** & -0.031*** \\
                               &          & (0.003)   & (0.004)   &          & (0.008)   & (0.010)   \\
Coauthor Control               &          & 0.025***  & 0.026***  &          & 0.023***  & 0.025***  \\
                               &          & (0.005)   & (0.005)   &          & (0.006)   & (0.006)   \\ \midrule
Observations                   & 798,548  & 798,548   & 797,434   & 798,548  & 798,548   & 797,434   \\
Baseline FEs                   & Yes      & Yes       & Yes       & Yes      & Yes       & Yes       \\
Author/Coauthor Controls       & No       & Yes       & Yes       & No       & Yes       & Yes       \\
MSA $\times$ Year FE           & No       & No        & Yes       & No       & No        & Yes       \\
Kleibergen–Paap Wald F         & N/A      & N/A       & N/A       & 15.88    & 14.90     & 12.97     \\
\bottomrule  \noalign{\vskip 0.1in}
\multicolumn{7}{l}{%
\begin{minipage}{15.8cm}%
\footnotesize{
\textit{Note:} This table reports OLS and IV estimates of Equation \ref{eq:agg_reg} using the IV sample. The dependent variable is IHS publications in Panel A, IHS citations in Panel B, and IHS patent citations in Panel C. The sample consists of author-year observations in three periods: 1995, 2005, and 2015, where each period pools the focal year and the two preceding years. Columns 1--3 report OLS estimates. Columns 4--6 report IV estimates that instrument log own-institution size and log external-cluster size using the Bartik instruments in Equation \ref{eq:bartik}. Columns 1 and 4 include the baseline fixed effects from Equation \ref{eq:agg_reg}, including author fixed effects, academic-age-by-field fixed effects, field-by-year fixed effects, and affiliation-by-field fixed effects. Columns 2 and 5 add the author and coauthor direct-exposure controls. Columns 3 and 6 further add MSA-by-year fixed effects. The Kleibergen--Paap rk Wald F statistic is reported for the IV specifications. Standard errors are in parentheses and clustered at the institution level: *** $p< 0.01$, ** $p< 0.05$, * $p< 0.1$.
}
\end{minipage}}%
\end{tabular}
\label{table:result_IV_main}
\end{table}

\begin{table}[h]
\centering
\captionsetup{justification=centering}
\caption{IV Estimates of Agglomeration Elasticities by Field Group}
\begin{tabular}{lccccc}
\toprule
                       & Social Sciences & Humanities & Natural Sciences & Engineering       & Life Sciences \\
                       & (1)             & (2)        & (3)              & (4)               & (5)           \\ \midrule
\multicolumn{6}{l}{\textit{Panel A: Publications}}                                                           \\
Ln(Institution Size)   & 0.118           & 0.123      & 0.216            & 0.210***          & 0.292         \\
                       & (0.121)         & (0.198)    & (0.143)          & (0.066)           & (0.186)       \\
Ln(External Size)      & 0.082***        & -0.018     & 0.074            & 0.019             & 0.118***      \\
                       & (0.022)         & (0.033)    & (0.058)          & (0.046)           & (0.040)       \\ \addlinespace
\multicolumn{6}{l}{\textit{Panel B: Citations}}                                                              \\
Ln(Institution Size)   & 0.920**         & -0.637     & 0.685*           & 0.654***          & 1.516***      \\
                       & (0.377)         & (0.589)    & (0.352)          & (0.192)           & (0.515)       \\
Ln(External Size)      & 0.019           & -0.142     & 0.119            & -0.160            & 0.161         \\
                       & (0.074)         & (0.125)    & (0.135)          & (0.126)           & (0.117)       \\ \addlinespace
\multicolumn{6}{l}{\textit{Panel C: Patent Citations}}                                                       \\
Ln(Institution Size)   & 0.012           & -0.013     & 0.138*           & 0.003             & 0.499***      \\
                       & (0.025)         & (0.035)    & (0.071)          & (0.028)           & (0.169)       \\
Ln(External Size)      & -0.004          & -0.007     & -0.008           & -0.024            & -0.003         \\
                       & (0.004)         & (0.008)    & (0.028)          & (0.026)           & (0.038)       \\ \midrule
Observations           & 118,826         & 21,133     & 173,316          & 152,814           & 332,459       \\
Kleibergen-Paap Wald F & 8.48            & 2.37       & 4.01             & 20.45             & 4.66          \\
\bottomrule  \noalign{\vskip 0.1in}
\multicolumn{6}{l}{%
\begin{minipage}{16.5cm}%
\footnotesize{
\textit{Note:} This table reports IV estimates of Equation \ref{eq:agg_reg} using the IV sample separately for each broad field group. The dependent variable is IHS publications in Panel A, IHS citations in Panel B, and IHS patent citations in Panel C. The sample consists of author-year observations in three periods: 1995, 2005, and 2015, where each period pools the focal year and the two preceding years. All columns include the baseline fixed effects from Equation \ref{eq:agg_reg} and the author and coauthor direct-exposure controls. Standard errors are in parentheses and clustered at the institution level: *** $p< 0.01$, ** $p< 0.05$, * $p< 0.1$.
}
\end{minipage}}%
\end{tabular}
\label{table:result_IV_field}
\end{table}

\clearpage 

\begin{table}[htbp]
\centering
\begin{threeparttable}
\caption{Counterfactual Reallocations of Researchers: Bilateral Cases}
\label{table:case_studies}
\begin{tabular}{cccc ccc}
\toprule
\multicolumn{4}{c}{Knowledge Production Changes (\%)} & & & \\
\cmidrule(lr){1-4}
Origin & Dest. & & Net & Access & Break-Even & Implied Access\ \\
Stayers & Incumb. & Movers & Change & Calibration & $\theta^{\ast}$ & Elast.\ $s^{N}\theta^{\ast}\bar{\lambda}$ \\
\midrule
\multicolumn{7}{l}{\textit{Panel A: Boston $\rightarrow$ Dallas--Fort Worth}}\\[2pt]
$-1.5\%$ & $+8.65\%$ & $-4.91\%$ & $-0.59\%$ & \quad $\bar\lambda=0.2$ & 0.482 & 0.0193 \\
 & & & & \quad $\bar\lambda=0.5$ & 0.186 & 0.0186 \\
 & & & & \quad $\bar\lambda=0.7$ & 0.132 & 0.0185 \\
\addlinespace
\multicolumn{7}{l}{\textit{Panel B: San Francisco $\rightarrow$ Las Vegas}}\\[2pt]
$-0.78\%$ & $+15.54\%$ & $-5.33\%$ & $-0.56\%$ & \quad $\bar\lambda=0.2$ & 0.462 & 0.0185 \\
 & & & & \quad $\bar\lambda=0.5$ & 0.178 & 0.0178 \\
 & & & & \quad $\bar\lambda=0.7$ & 0.126 & 0.0176 \\
\bottomrule
\end{tabular}
\begin{tablenotes}[flushleft]
\footnotesize
\item \emph{Notes:} Each panel reallocates researchers within one origin–destination MSA pair, with $\kappa=0.1$. Researchers are assigned following the procedure in Section \ref{section:reallocation_rule}, and results are averaged across 50 simulations. The first three columns report percentage changes in output among researchers remaining in the origin, incumbent researchers in the destination, and movers, each relative to that group’s baseline output. ``Net Change" reports the percentage change in aggregate flow research output. The change in knowledge production depends only on research output, so the first four columns are invariant to the access calibration. For each value of $\bar{\lambda}$, the break-even $\theta^{\ast}$ solves $\Delta \ln Y=0$, and the final column reports the corresponding access elasticity $s^{N}\theta^{\ast}\bar{\lambda}$. Parameter calibrations are described in Section \ref{section:calibration}: $s^{N}=0.2$, $\eta=0.5$, $\delta=0.15$, $\alpha^{I}=0.199$, $\alpha^{E}=0.074$.
\end{tablenotes}
\end{threeparttable}
\end{table}

\clearpage

\begin{table}[htbp]
\centering
\begin{threeparttable}
\caption{Counterfactual Reallocations of Researchers: Across Metropolitan Areas}
\label{table:systematic}
\begin{tabular}{cccc ccc}
\toprule
\multicolumn{4}{c}{Knowledge Production Changes (\%)} & & & \\
\cmidrule(lr){1-4}
Origin & Dest. & & Net & Access & Break-Even & Implied Access \\
Stayers & Incumb. & Movers & Change & Calibration & $\theta^{\ast}$ & Elast.\ $s^{N}\theta^{\ast}\bar{\lambda}$ \\
\midrule
\multicolumn{7}{l}{\textit{Panel A: Top 50 MSAs (10 Highest- to 10 Lowest-Researchers-per-Capita MSAs)}}\\[2pt]
$-0.99\%$ & $+8.78\%$ & $-1.42\%$ & $-0.32\%$ & \quad $\bar\lambda=0.2$ & 0.245 & 0.0098 \\
 & & & & \quad $\bar\lambda=0.5$ & 0.096 & 0.0096 \\
 & & & & \quad $\bar\lambda=0.7$ & 0.069 & 0.0096 \\
 & & & & \quad $\bar\lambda=0.5$, w/ migration & 0.120 & 0.0120 \\
\addlinespace
\multicolumn{7}{l}{\textit{Panel B: All MSAs with Research Institutions}}\\[2pt]
$-1.94\%$ & $+6.13\%$ & $-14.54\%$ & $-1.14\%$ & \quad $\bar\lambda=0.2$ & 1.075 & 0.0430 \\
 & & & & \quad $\bar\lambda=0.5$ & 0.393 & 0.0393 \\
 & & & & \quad $\bar\lambda=0.7$ & 0.277 & 0.0388 \\
\bottomrule
\end{tabular}
\begin{tablenotes}[flushleft]
\footnotesize
\item \emph{Notes:} Panel A reallocates researchers from the ten MSAs with the highest researchers per capita to the ten with the lowest among the 50 most populous MSAs. Panel B extends the reallocation to all MSAs that host research institutions. In both panels, $\kappa = 0.1$; assignments follow Section \ref{section:reallocation_rule} and results are averaged across 50 simulations. The first three columns report percentage changes in output among researchers remaining in the origin, incumbent researchers in the destination, and movers, each relative to that group’s baseline output. ``Net Change" reports the percentage change in aggregate flow research output. The change in knowledge production depends only on research output, so the first four columns are invariant to the access calibration. For each value of $\bar{\lambda}$, the break-even $\theta^{\ast}$ solves $\Delta \ln Y=0$, and the final column reports the corresponding access elasticity $s^{N}\theta^{\ast}\bar{\lambda}$. The migration row in Panel A allows non-researchers to relocate as described in Appendix \ref{app:nonresearcher}.  Parameter calibrations for Panel A are the same as for Table \ref{table:case_studies}. Panel B re-estimates and applies the agglomeration coefficients using $\ln(N+5)$ to make the transformation less curved near zero, reducing the
influence of changes involving only a few researchers: $\alpha^{I}=0.537$, $\alpha^{E}=0.043$ (Table \ref{table:result_IV_zero} Column 2).
\end{tablenotes}
\end{threeparttable}
\end{table}

\begin{table}[htbp]
\centering
\begin{threeparttable}
\caption{Counterfactual Reallocations of Researchers: Alternative Specifications ($\bar{\lambda}=0.5$)}
\label{tab:alternatives}
\begin{tabular}{cccc lcc}
\toprule
\multicolumn{4}{c}{Knowledge Production Changes (\%)} & & & \\
\cmidrule(lr){1-4}
Origin & Dest. & & Net &   & Break-even & Implied Access \\
Stayers & Incumb. & Movers & Change & \multicolumn{1}{c}{Access Measure} & $\theta^{\ast}$ & Elast.\ $s^{N}\theta^{\ast}\bar{\lambda}$ \\
\midrule
\multicolumn{7}{l}{\textit{Panel A: Alternative Measures of Knowledge Access}}\\[2pt]
$-0.99\%$ & $+8.78\%$ & $-1.42\%$ & $-0.32\%$ & \quad Researchers per capita     & 0.096 & 0.0096 \\
 & & & & \quad Publications per capita & 0.068 & 0.0068 \\
 & & & & \quad Citations per capita       & 0.034 & 0.0034 \\
\addlinespace
\multicolumn{7}{l}{\textit{Panel B: Publication-Based Agglomeration}}\\[2pt]
$-0.55\%$ & $+6.24\%$ & $+1.07\%$ & $-0.10\%$ & \quad Researchers per capita     & 0.031 & 0.0031 \\
 & & & & \quad Publications per capita & 0.024 & 0.0024 \\
 & & & & \quad Citations per capita       & 0.008 & 0.0008 \\
\addlinespace
\multicolumn{7}{l}{\textit{Panel C: Reallocation Toward the College-Educated Population}}\\[2pt]
$-0.87\%$ & $+6.33\%$ & $-3.97\%$ & $-0.29\%$ & \quad Researchers per capita     & 0.117 & 0.0117 \\
 & & & & \quad Publications per capita & 0.086 & 0.0086 \\
 & & & & \quad Citations per capita       & 0.042 & 0.0042 \\
\addlinespace
\multicolumn{7}{l}{\textit{Panel D: Reallocation of STEM Researchers Toward the College-Educated Population}}\\[2pt]
$-0.85\%$ & $+6.40\%$ & $-4.04\%$ & $-0.30\%$ & \quad Researchers per capita     & 0.121 & 0.0121 \\
 & & & & \quad Publications per capita & 0.090 & 0.0090 \\
 & & & & \quad Citations per capita       & 0.043 & 0.0043 \\
\bottomrule
\end{tabular}
\begin{tablenotes}[flushleft]
\footnotesize
\item \emph{Notes:} All panels reallocate researchers from the ten MSAs with the highest researchers per capita to the ten with the lowest among the 50 most populous MSAs, with $\kappa = 0.1$ and $\bar{\lambda}=0.5$. Panel A considers the same reallocation as Table \ref{table:systematic} Panel A, but allows local access to depend on publications or citations per capita, in addition to researchers per capita. Panel B measures cluster size by paper counts (rather than author counts) and applies the corresponding agglomeration estimates reported in Table \ref{table:result_IV_cluster}. Panel C defines researcher surpluses, deficits, and knowledge access relative to the college-educated population. Panel D repeats the exercise in Panel C but reallocates STEM researchers only. Within each panel, the rows measure knowledge access using researcher counts, publications, or citations per capita. The first three columns report percentage changes in output among researchers remaining in the origin, incumbent researchers in the destination, and movers, each relative to that group’s baseline output. ``Net Change" reports the percentage change in aggregate flow research output. The change in knowledge production depends only on research output, so the first four columns are invariant to the access calibration. The break-even $\theta^{\ast}$ solves $\Delta \ln Y=0$, and the final column reports the corresponding access elasticity $s^{N}\theta^{\ast}\bar{\lambda}$. Parameter calibrations for Panels A, C, and D are the same as for Table \ref{table:case_studies}. For Panel B: $\alpha^{I}=0.101$, $\alpha^{E}=0.049$.
\end{tablenotes}
\end{threeparttable}
\end{table}

\clearpage
\begin{singlespace}
\bibliography{reference}

@article{chandra2025person,
  title={{Person and Place Effects in Scientific Discovery}},
  author={Chandra, Amitabh and Xu, Connie},
  year={2025},
  journal={Working Paper},
  publisher={National Bureau of Economic Research Working Paper}
}

@article{chen2024logs,
  title={{Logs with Zeros? Some Problems and Solutions}},
  author={Chen, Jiafeng and Roth, Jonathan},
  journal={The Quarterly Journal of Economics},
  volume={139},
  number={2},
  pages={891--936},
  year={2024},
  publisher={Oxford University Press}
}

@article{wang2020microsoft,
  title={{Microsoft Academic Graph: {When} Experts are Not Enough}},
  author={Wang, Kuansan and Shen, Zhihong and Huang, Chiyuan and Wu, Chieh-Han and Dong, Yuxiao and Kanakia, Anshul},
  journal={Quantitative Science Studies},
  volume={1},
  number={1},
  pages={396--413},
  year={2020},
  publisher={MIT Press One Rogers Street, Cambridge, MA 02142-1209, USA}
}

@article{visser2021large,
  title={{Large-Scale Comparison of Bibliographic Data Sources: {Scopus}, {Web} of {Science}, {Dimensions}, {Crossref}, and {Microsoft Academic}}},
  author={Visser, Martijn and Van Eck, Nees Jan and Waltman, Ludo},
  journal={Quantitative Science Studies},
  volume={2},
  number={1},
  pages={20--41},
  year={2021},
  publisher={MIT Press One Rogers Street, Cambridge, MA 02142-1209, USA}
}

@article{emanuel2023power,
  title={{The Power of Proximity}},
  author={Emanuel, Natalia and Harrington, Emma and Pallais, Amanda},   
  journal={Working Paper},
  year={2023}
}

@article{belenzon2013spatial,
  title={{Spatial Heterogeneity in the Diffusion of Technologies: Evidence from Patent Citations}},
  author={Belenzon, Sharon and Schankerman, Mark},
  journal={Journal of Industrial Economics},
  volume={61},
  number={3},
  pages={733--758},
  year={2013},
  publisher={Wiley Online Library},
  doi={10.1111/joie.12025}
}

@article{belenzon2013spreading,
  title   = {Spreading the Word: Geography, Policy, and Knowledge Spillovers},
  author  = {Belenzon, Sharon and Schankerman, Mark},
  journal = {The Review of Economics and Statistics},
  year    = {2013},
  volume  = {95},
  number  = {3},
  pages   = {884--903},
  doi     = {10.1162/REST_a_00334}
}

@article{carlino2015agglomeration,
  title={{Agglomeration and Innovation}},
  author={Carlino, Gerald and Kerr, William R},
  journal={Handbook of Regional and Urban economics},
  volume={5},
  pages={349--404},
  year={2015},
  publisher={Elsevier}
}

@article{guzman2024go,
  title={{Go West Young Firm: {T}he Impact of Startup Migration on the Performance of Migrants}},
  author={Guzman, Jorge},
  journal={Management Science},
  volume={70},
  number={7},
  pages={4824--4846},
  year={2024},
  publisher={INFORMS}
}

@article{roca_puga2017,
  title={{Learning by Working in Big Cities}},
  author={De La Roca, Jorge and Puga, Diego},
  journal={Review of Economic Studies},
  volume={84},
  number={1},
  pages={106–142},
  year={2017}
}

@article{glaeser2001,
  title={{Cities and Skills}},
  author={Glaeser, Edward and Mare, David},
  journal={Journal of Labor Economics},
  volume={19},
  number={2},
  pages={316-342},
  year={2001}
}

@article{milliondollarplant,
  title={{Identifying Agglomeration Spillovers: Evidence from Winners and Losers of Large Plant Openings}},
  author={Greenstone, Michael and Hornbeck, Richard and Moretti, Enrico},
  journal={Journal of Political Economy},
  volume={118},
  number={3},
  pages={536-598},
  year={2010}
}

@article{dcosta_overman2014,
  title={{The Urban Wage Growth Premium: Sorting or Learning?} },
  author={D'Costa, Sabine and Overman, Henry},
  journal={Regional Science and Urban Economics},
  volume={48},
  pages={168-179},
  year={2014}
}

@article{eckertwalsh2022,
  title={{The Return to Big City Experience: Evidence from Refugees in Denmark}},
  author={Eckert, Fabian and Hejlesen, Mads and Walsh, Conor},
  journal={Journal of Urban Economics},
  number={103454},
  year={2022}
}

@article{rosenthal_strange2003,
  title={{Geography, Industrial Organization, and Agglomeration}},
  author={Rosenthal, Stuart and Strange, William},
  journal={Review of Economics and Statistics},
  volume={85},
  number={2},
  pages={377–393},
  year={2003}
}

@article{rosenthal_strange2008,
  title={{The Attenuation of Human Capital Spillovers}},
  author={Rosenthal, Stuart and Strange, William},
  journal={Journal of Urban Economics},
  volume={64},
  pages={373-389},
  year={2008}
}

@article{ellisonglaeserkerr2010,
  title={{What Causes Industry Agglomeration? Evidence from Coagglomeration Patterns}},
  author={Ellison, Glenn and Glaeser, Edward and Kerr, William},
  journal={American Economic Review},
  volume={100},
  number={3},
  pages={1195-1213},
  year={2010}
}

@article{diamond2016,
  title={{The Determinants and Welfare Implications of US Workers' Diverging Location Choices by Skill: 1980-2000}},
  author={Diamond, Rebecca},
  journal={American Economic Review},
  volume={106},
  number={3},
  pages={479-524},
  year={2016}
}

@article{glaeser1999,
  title={{Learning in Cities}},
  author={Glaeser, Edward},
  journal={Journal of Urban Economics},
  volume={46},
  number={2},
  pages={254-277},
  year={1999}
}

@article{charlot_duranton2004,
  title={{Communication Externalities in Cities}},
  author={Charlot, Sylvie and Duranton, Gilles},
  journal={Journal of Urban Economics},
  volume={56},
  number={3},
  pages={581–613},
  year={2004}
}

@article{jaffe1993,
  title={{Geographic Localization of Knowledge Spillovers as Evidenced by Patent Citations}},
  author={Jaffe, Adam and Trajtenberg, Manuel and Henderson, Rebecca},
  journal={Quarterly Journal of Economics},
  volume={108},
  number={3},
  pages={577-598},
  year={1993}
}

@article{gaubert2018,
  title={{Firm Sorting and Agglomeration}},
  author={Gaubert, Cecile},
  journal={American Economic Review},
  volume={108},
  number={11},
  pages={3117-3153},
  year={2018}
}

@article{baumsnow_pavan2024,
  title={Local {P}roductivity {S}pillovers},
  author={Baum-Snow, Nathaniel and Gendron-Carrier, Nicolas and Pavan, Ronni},
  journal={American Economic Review},
  year={2024},
  volume={114},
  number={4},
  pages={1030--69}
}

@article{dingel_davis_2019,
  title={{A Spatial Knowledge Economy}},
  author={Davis, Donald and Dingel, Jonathan},
  journal={American Economic Review},
  volume={109},
  number={1},
  pages={153-170},
  year={2019}
}

@article{moretti2021,
  title={{The Effect of High-Tech Clusters on the Productivity of Top Inventors}},
  author={Enrico Moretti},
  journal={American Economic Review},
  volume={111},
  number={10},
  pages={3328–3375},
  year={2021}
}

@article{atkin2022,
  title={The {R}eturns to {F}ace-to-{F}ace {I}nteractions: {K}nowledge {S}pillovers in {S}ilicon {V}alley},
  author={David Atkin and Keith Chen and Anton Popov},
  journal={Working Paper},
  year={2022}
}

@article{Arzaghi2008,
  title={Networking off {M}adison {A}venue},
  author={Mohammad Arzaghi and Vernon Henderson},
  journal={Review of Economic Studies},
  volume={75},
  number={4},
  pages={1011–1038},
  year={2008}
}

@article{dingel_medical2023,
  title={Market {S}ize and {T}rade in {M}edical {S}ervices},
  author={Jonathan Dingel and Josh Gottlieb and Maya Lozinski and Pauline Mourot},
  journal={Working Paper},
  year={2023}
}

@article{biasi2023,
  title={The {E}ducation-{I}nnovation {G}ap},
  author={Barbara Biasi and Song Ma},
  journal={Working Paper},
  year={2023}
}

@article{mas_compile,
  title={{An Overview of Microsoft Academic Service (MAS) and Applications}},
  author={Arnab Sinha and Zhihong Shen and Yang Song and Hao Ma and Darrin Eide and Bo-June Hsu and Kuansan Wang},
  journal={WWW'15 Companion: Proceedings of the 24th International Conference on World Wide Web},
  pages={243-246},
  year={2015}
}

@article{emakg,
  title={{EMAKG: An Enriched Version of the Microsoft Academic Knowledge Graph}},
  author={Pollacci, Laura},
  journal={https://github.com/LauraPollacci/EMAKG},
  year={2022}
}

@article{reliance_on_science,
  title={{Reliance on Science: Worldwide Front-Page Patent Citations to Scientific Articles}},
  author={Matt Marx and Aaron Fuegi},
  journal={Strategic Management Journal},
  volume={41},
  number={9},
  pages={1572--1594},
  year={2020}
}

@incollection{card1995distance,
  title={{Using Geographic Variation in College Proximity to Estimate the Return to Schooling}},
  author={David Card},
  booktitle={Aspects of Labour Economics: Essays in Honour of John Vanderkamp, edited by Louis Christofides, E. Kenneth Grant and Robert Swindinsky},
  year={1995},
  publisher={University of Toronto Press}
}

@article{rouse1995cc,
  title={{Democratization or Diversion? {T}he Effect of Community Colleges on Educational Attainment}},
  author={Cecilia Rouse},
  journal={Journal of Business and Economic Statistics},
  volume={13},
  number={2},
  pages={217-224},
  year={1995},
}

@article{black2020access,
  title={{Apply Yourself: {R}acial and Ethnic Differences in College Application}},
  author={Sandra Black and Kalena Cortes and Jane Arnold Lincove},
  journal={Education Finance and Policy},
  volume={15},
  number={2},
  pages={209-240},
  year={2020},
}

@article{fu2022access,
  title={{Students’ Heterogeneous Preferences and the Uneven Spatial Distribution of Colleges}},
  author={Chao Fu and Junjie Guo and Adam Smith and Alan Sorensen},
  journal={Journal of Monetary Economics},
  volume={129},
  pages={49-64},
  year={2022},
}

@article{mountjoy2022cc,
  title={{Community Colleges and Upward Mobility}},
  author={Jack Mountjoy},
  journal={American Economic Review},
  volume={112},
  number={8},
  pages={2580-2630},
  year={2022},
}

@inproceedings{acton2024access,
  title={{Distance to Opportunity: {H}igher Education Deserts and College Enrollment Choices}},
  author={Riley Acton and Kalena Cortes and Camila Morales},
  booktitle={Financing Institutions of Higher Education, Chapter 2},
  year={2024}
}

@article{ishimaru2024geographic,
  title={{Geographic Mobility of Youth and Spatial Gaps in Local College and Labor Market Opportunities}},
  author={Ishimaru, Shoya},
  journal={Journal of Labor Economics},
  volume={43},
  number={4},
  pages={1251--1294},
  year={2025},
  publisher={The University of Chicago Press Chicago, IL}
}

@article{fabre2023geography,
  title={{The Geography of Higher Education and Spatial Inequalities}},
  author={Fabre, Ana{\i}s},
  year={2023},
  journal={Working Paper}
}

@article{duranton2005testing,
  title={{Testing for Localization using Micro-Geographic Data}},
  author={Duranton, Gilles and Overman, Henry G},
  journal={The Review of Economic Studies},
  volume={72},
  number={4},
  pages={1077--1106},
  year={2005},
  publisher={Wiley-Blackwell}
}

@article{kantor2014knowledge,
  title={{Knowledge Spillovers from Research Universities: {E}vidence from Endowment Value Shocks}},
  author={Kantor, Shawn and Whalley, Alexander},
  journal={Review of Economics and Statistics},
  volume={96},
  number={1},
  pages={171--188},
  year={2014},
  publisher={The MIT Press}
}

@article{kantor2019research,
  title={{Research Proximity and Productivity: {L}ong-Term Evidence from Agriculture}},
  author={Kantor, Shawn and Whalley, Alexander},
  journal={Journal of Political Economy},
  volume={127},
  number={2},
  pages={819--854},
  year={2019},
  publisher={The University of Chicago Press Chicago, IL}
}

@article{hausman2022university,
  title={{University Innovation and Local Economic Growth}},
  author={Hausman, Naomi},
  journal={Review of Economics and Statistics},
  volume={104},
  number={4},
  pages={718--735},
  year={2022},
  publisher={MIT Press One Rogers Street, Cambridge, MA 02142-1209, USA journals-info~…}
}

@techreport{lerner2024wandering,
  title={{The Wandering Scholars: {U}nderstanding the Heterogeneity of University Commercialization}},
  author={Lerner, Josh and Manley, Henry and Stein, Carolyn and Williams, Heidi},
  year={2024},
  institution={National Bureau of Economic Research}
}

@article{jaff1989academic,
  title={{Real Effects of Academic Research}},
  author={Adam Jaffe},
  journal={American Economic Review},
  volume={79},
  number={5},
  pages={957--970},
  year={1989},
  publisher={American Economic Association 2014 Broadway, Suite 305, Nashville, TN 37203}
}

@article{waldinger2012peer,
  title={{Peer Effects in Science: {E}vidence from the Dismissal of Scientists in {N}azi {G}ermany}},
  author={Waldinger, Fabian},
  journal={Review of Economic Studies},
  volume={79},
  number={2},
  pages={838--861},
  year={2012},
  publisher={Oxford University Press}
}

@article{finkelstein2016sources,
  title={{Sources of Geographic Variation in Health Care: {E}vidence from Patient Migration}},
  author={Finkelstein, Amy and Gentzkow, Matthew and Williams, Heidi},
  journal={Quarterly Journal of Economics},
  volume={131},
  number={4},
  pages={1681--1726},
  year={2016},
  publisher={MIT Press}
}

@article{rossi_hansberg2023,
  title={{Cognitive Hubs and Spatial Redistribution}},
  author={Esteban Rossi-Hansberg and Sarte, Pierre-Daniel and Schwartzman, Felipe},
  journal={American Economic Journal: Macroeconomics},
  volume={18},
  number={2},
  pages={72--111},
  year={2026},
  publisher={American Economic Association}
}

@article{mertens2024,
  title={{The Returns to Government R\&D: Evidence from U.S. Appropriations Shocks}},
  author={Andrew Fieldhouse and Karel Mertens},
  journal={Working Paper},
  year={2024}
}

@article{moretti_return2004,
  title={{Estimating the Social Return to Higher Education: Evidence from Longitudinal and Repeated Cross-Sectional Data}},
  author={Moretti, Enrico},
  journal={Journal of Econometrics},
  volume={121},
  number={1-2},
  pages={175--212},
  year={2004},
  publisher={Elsevier}
}

@article{van_reenen2019,
  title={{The Economic Impact of Universities: Evidence from Across the Globe}},
  author={Valero, Anna  and Van Reenen, John },
  journal={Economics of Education Review},
  volume={68},
  pages={53-67},
  year={2019}
}

@article{aghion2009causal,
  title={{The Causal Impact of Education on Economic Growth: Evidence from the United States}},
  author={Aghion, Philippe and Boustan, Leah and Hoxby, Caroline and Vandenbussche, J{\'e}r{\^o}me},
  journal={Brookings Papers on Economic Activity},
  volume={2009},
  number={2},
  pages={1--73},
  year={2009},
  publisher={Brookings Institution Press}
}

@article{helmers2017my,
  title={{My Precious! {T}he Location and Diffusion of Scientific Research: {E}vidence from the Synchrotron Diamond Light Source}},
  author={Helmers, Christian and Overman, Henry G},
  journal={The Economic Journal},
  volume={127},
  number={604},
  pages={2006--2040},
  year={2017},
  publisher={Oxford University Press Oxford, UK}
}

@article{bosquet2017sorting,
  title={{Sorting and Agglomeration Economies in {F}rench Economics Departments}},
  author={Bosquet, Cl{\'e}ment and Combes, Pierre-Philippe},
  journal={Journal of Urban Economics},
  volume={101},
  pages={27--44},
  year={2017},
  publisher={Elsevier}
}

@article{azoulay2010superstar,
  title={{Superstar Extinction}},
  author={Azoulay, Pierre and Graff Zivin, Joshua S and Wang, Jialan},
  journal={The Quarterly Journal of Economics},
  volume={125},
  number={2},
  pages={549--589},
  year={2010},
  publisher={MIT Press}
}

@misc{nhgis,
  author       = {Manson, Steven and Schroeder, Jonathan and Van Riper, David and Kugler, Tracy and Ruggles, Steven},
  title        = {{IPUMS National Historical Geographic Information System: Version 17.0}},
  year         = {2022},
  publisher    = {IPUMS},
  address      = {Minneapolis, MN},
  howpublished = {\url{https://www.nhgis.org}},
  doi          = {10.18128/D050.V17.0}
}

@misc{cbp2021,
  author       = {{U.S. Census Bureau}},
  title        = {{County Business Patterns, 2021}},
  year         = {2023},
  howpublished = {\url{https://www.census.gov/programs-surveys/cbp.html}},
  note         = {U.S. Department of Commerce}
}

@misc{nara_morrill_1862,
  author       = {{National Archives}},
  title        = {Morrill Act (1862)},
  howpublished = {U.S. National Archives, Milestone Documents},
  year         = {1862},
  url          = {https://www.archives.gov/milestone-documents/morrill-act}
}

@misc{aplu_landgrant_tradition,
  author       = {{Association of Public and Land-Grant Universities}},
  title        = {{The Land-Grant Tradition}},
  howpublished = {APLU report},
  year         = {2012},
  url          = {https://www.aplu.org/wp-content/uploads/the-land-grant-tradition.pdf}
}

@misc{nea_landgrant_overview,
  author       = {{National Education Association}},
  title        = {{Land Grant Institutions: An Overview}},
  howpublished = {NEA brief},
  year         = {2022},
  url          = {https://www.nea.org/sites/default/files/2022-03/Land%20Grant%20Institutions%20-%20An%20Overview.pdf}
}

@article{hesse1993returns,
  title={{Returns to Science: {C}omputer Networks in Oceanography}},
  author={Hesse, Bradford W and Sproull, Lee S and Kiesler, Sara B and Walsh, John P},
  journal={Communications of the ACM},
  volume={36},
  number={8},
  pages={90--101},
  year={1993},
  publisher={ACM New York, NY, USA}
}

@article{cohen1996computer,
  title={{Computer Mediated Communication and Publication Productivity Among Faculty}},
  author={Cohen, Joel},
  journal={Internet Research},
  volume={6},
  number={2/3},
  pages={41--63},
  year={1996},
  publisher={MCB UP Ltd}
}

@article{van1996could,
  title={{Could the {I}nternet Balkanize Science?}},
  author={Van Alstyne, Marshall and Brynjolfsson, Erik},
  journal={Science},
  volume={274},
  number={5292},
  pages={1479--1480},
  year={1996},
  publisher={American Association for the Advancement of Science}
}

@article{walsh1996computer,
  title={{Computer Networks and Scientific Work}},
  author={Walsh, John P and Bayma, Todd},
  journal={Social Studies of Science},
  volume={26},
  number={3},
  pages={661--703},
  year={1996},
  publisher={Sage Publications London}
}

@article{kaminer1998bibliometric,
  title={{Bibliometric Analysis of the Impact of Internet Use on Scholarly Productivity}},
  author={Kaminer, Noam and Braunstein, Yale M},
  journal={Journal of the American Society for Information Science},
  volume={49},
  number={8},
  pages={720--730},
  year={1998},
  publisher={Wiley Online Library}
}

@article{walsh2000connecting,
  title={{Connecting Minds: {C}omputer-Mediated Communication and Scientific Work}},
  author={Walsh, John P and Kucker, Stephanie and Maloney, Nancy G and Gabbay, Shaul},
  journal={Journal of the American Society for Information Science},
  volume={51},
  number={14},
  pages={1295--1305},
  year={2000},
  publisher={Wiley Online Library}
}

@article{hamermesh2002tools,
  title={{Tools or Toys? {T}he Impact of High Technology on Scholarly Productivity}},
  author={Hamermesh, Daniel S and Oster, Sharon M},
  journal={Economic Inquiry},
  volume={40},
  number={4},
  pages={539--555},
  year={2002},
  publisher={Wiley Online Library}
}

@article{rosenblat2004getting,
  title={{Getting Closer or Drifting Apart?}},
  author={Rosenblat, Tanya S and Mobius, Markus M},
  journal={The Quarterly Journal of Economics},
  volume={119},
  number={3},
  pages={971--1009},
  year={2004},
  publisher={MIT Press}
}

@article{adams2005scientific,
  title={Scientific teams and institutional collaborations: {E}vidence from {US} universities, 1981--1999},
  author={Adams, James D and Black, Grant C and Clemmons, J Roger and Stephan, Paula E},
  journal={Research Policy},
  volume={34},
  number={3},
  pages={259--285},
  year={2005},
  publisher={Elsevier}
}

@article{wuchty2007increasing,
  title={{The Increasing Dominance of Teams in Production of Knowledge}},
  author={Wuchty, Stefan and Jones, Benjamin F and Uzzi, Brian},
  journal={Science},
  volume={316},
  number={5827},
  pages={1036--1039},
  year={2007},
  publisher={American Association for the Advancement of Science}
}

@article{agrawal2008restructuring,
  title={{Restructuring Research: {C}ommunication Costs and the Democratization of University Innovation}},
  author={Agrawal, Ajay and Goldfarb, Avi},
  journal={American Economic Review},
  volume={98},
  number={4},
  pages={1578--1590},
  year={2008},
  publisher={American Economic Association}
}

@article{butler2008equalizing,
  title={{The Equalizing Effect of the Internet on Access to Research Expertise in Political Science and Economics}},
  author={Butler, Daniel M and Butler, Richard J and Rich, Jesse T},
  journal={PS: Political Science \& Politics},
  volume={41},
  number={3},
  pages={579--584},
  year={2008},
  publisher={Cambridge University Press}
}

@article{jones2008multi,
  title={{Multi-University Research Teams: {S}hifting Impact, Geography, and Stratification in Science}},
  author={Jones, Benjamin F and Wuchty, Stefan and Uzzi, Brian},
  journal={Science},
  volume={322},
  number={5905},
  pages={1259--1262},
  year={2008},
  publisher={American Association for the Advancement of Science}
}

@article{kim2009elite,
  title={{Are Elite Universities Losing Their Competitive Edge?}},
  author={Kim, E Han and Morse, Adair and Zingales, Luigi},
  journal={Journal of Financial Economics},
  volume={93},
  number={3},
  pages={353--381},
  year={2009},
  publisher={Elsevier}
}

@article{ding2010impact,
  title={{The Impact of Information Technology on Academic Scientists' Productivity and Collaboration Patterns}},
  author={Ding, Waverly W and Levin, Sharon G and Stephan, Paula E and Winkler, Anne E},
  journal={Management Science},
  volume={56},
  number={9},
  pages={1439--1461},
  year={2010},
  publisher={INFORMS}
}

@article{winkler2010diffusion,
  title={{The Diffusion of {IT} in Higher Education: {P}ublishing Productivity of Academic Life Scientists}},
  author={Winkler, Anne E and Levin, Sharon G and Stephan, Paula E},
  journal={Economics of Innovation and New Technology},
  volume={19},
  number={5},
  pages={481--503},
  year={2010},
  publisher={Taylor \& Francis}
}

@article{adams2011role,
  title={{The Role of Search in University Productivity: {I}nside, Outside, and Interdisciplinary Dimensions}},
  author={Adams, James D and Clemmons, J Roger},
  journal={Industrial and Corporate Change},
  volume={20},
  number={1},
  pages={215--251},
  year={2011},
  publisher={Oxford University Press}
}

@article{goldstein2024communication,
  title={{Communication Costs in Science: {E}vidence from the {National Science Foundation Network}}},
  author={Goldstein, Ezra G},
  journal={Industrial and Corporate Change},
  volume={33},
  number={4},
  pages={785--807},
  year={2024},
  publisher={Oxford University Press UK}
}

@article{krugman1980scale,
  title={{Scale Economies, Product Differentiation, and the Pattern of Trade}},
  author={Krugman, Paul R},
  journal={The American Economic Review},
  volume={70},
  number={5},
  pages={950--959},
  year={1980},
  publisher={JSTOR}
}

@techreport{hoxby2012endowment,
  author = {Hoxby, Caroline M.},
  title = {{Endowment Management Based on a Positive Model of the University}},
  year = {2012},
  institution = {National Bureau of Economic Research},
  type = {Working Paper},
  number = {18626}
}

@techreport{blair2021elite,
  author = {Blair, Peter Q. and Smetters, Kent},
  title = {{Why Don't Elite Colleges Expand Supply?}},
  year = {2021},
  institution = {National Bureau of Economic Research},
  type = {Working Paper},
  number = {29309}
}

@techreport{ehrenberg2003startup,
  author = {Ehrenberg, Ronald G. and Rizzo, Michael J. and Condie, Scott S.},
  title = {{Start-Up Costs in American Research Universities}},
  year = {2003},
  institution = {Cornell Higher Education Research Institute},
  type = {Working Paper},
  number = {33}
}

@article{howard2024universities,
  title={{Do Universities Improve Local Economic Resilience?}},
  author={Howard, Greg and Weinstein, Russell and Yang, Yuhao},
  journal={Review of Economics and Statistics},
  volume={106},
  number={4},
  pages={1129--1145},
  year={2024},
  publisher={MIT Press},
  doi={10.1162/rest_a_01212}
}

@article{jones1995,
  author  = {Jones, Charles I.},
  title   = {{R\&D-Based Models of Economic Growth}},
  journal = {Journal of Political Economy},
  year    = {1995},
  volume  = {103},
  number  = {4},
  pages   = {759--784},
  doi     = {10.1086/262002}
}

@article{romer1990,
  author  = {Romer, Paul M.},
  title   = {{Endogenous Technological Change}},
  journal = {Journal of Political Economy},
  year    = {1990},
  volume  = {98},
  number  = {5, Part 2},
  pages   = {S71--S102},
  doi     = {10.1086/261725}
}

@article{bloometal2020,
  author  = {Bloom, Nicholas and Jones, Charles I. and Van Reenen, John and Webb, Michael},
  title   = {{Are Ideas Getting Harder to Find?}},
  journal = {American Economic Review},
  year    = {2020},
  volume  = {110},
  number  = {4},
  pages   = {1104--1144},
  doi     = {10.1257/aer.20180338}
}

@article{goldsmith2020bartik,
  author  = {Goldsmith-Pinkham, Paul and Sorkin, Isaac and Swift, Henry},
  title   = {{Bartik Instruments: What, When, Why, and How}},
  journal = {American Economic Review},
  year    = {2020},
  volume  = {110},
  number  = {8},
  pages   = {2586--2624},
  doi     = {10.1257/aer.20181047}
}

@article{borusyak2022quasi,
  author  = {Borusyak, Kirill and Hull, Peter and Jaravel, Xavier},
  title   = {{Quasi-Experimental Shift-Share Research Designs}},
  journal = {Review of Economic Studies},
  year    = {2022},
  volume  = {89},
  number  = {1},
  pages   = {181--213},
  doi     = {10.1093/restud/rdab030}
}

@incollection{hall2010measuring,
  author    = {Hall, Bronwyn H. and Mairesse, Jacques and Mohnen, Pierre},
  title     = {Measuring the Returns to {R\&D}},
  booktitle = {Handbook of the Economics of Innovation},
  editor    = {Hall, Bronwyn H. and Rosenberg, Nathan},
  publisher = {North-Holland},
  address   = {Amsterdam},
  year      = {2010},
  volume    = {2},
  chapter   = {24},
  pages     = {1033--1082},
  doi       = {10.1016/S0169-7218(10)02008-3}
}

@article{coe_helpman1995,
  author  = {Coe, David T. and Helpman, Elhanan},
  title   = {International {R\&D} Spillovers},
  journal = {European Economic Review},
  volume  = {39},
  number  = {5},
  pages   = {859--887},
  year    = {1995}
}

@article{hsieh_moretti2019,
  author  = {Hsieh, Chang-Tai and Moretti, Enrico},
  title   = {{Housing Constraints and Spatial Misallocation}},
  journal = {American Economic Journal: Macroeconomics},
  volume  = {11},
  number  = {2},
  pages   = {1--39},
  year    = {2019},
  doi     = {10.1257/mac.20170388}
}

@book{pelfrey2004,
  author    = {Pelfrey, Patricia A.},
  title     = {A Brief History of the University of California},
  edition   = {2nd},
  publisher = {University of California Press},
  address   = {Berkeley},
  year      = {2004}
}

@book{merritt_lawrence2007,
  editor    = {Merritt, Karen and Lawrence, Jane Fiori},
  title     = {From Rangeland to Research University: The Birth of {University of California, Merced}},
  series    = {New Directions for Higher Education},
  number    = {139},
  publisher = {Jossey-Bass},
  address   = {San Francisco},
  year      = {2007}
}

@techreport{lao2024merced,
  author      = {{Legislative Analyst's Office}},
  title       = {{{UC Merced} at 20: Campus Developments and Key State-Level Takeaways}},
  institution = {California Legislative Analyst's Office},
  address     = {Sacramento},
  year        = {2024},
  month       = {November}
}

@article{thelwall2018does,
  title={{Does {M}icrosoft {A}cademic Find Early Citations?}},
  author={Thelwall, Mike},
  journal={Scientometrics},
  volume={114},
  number={1},
  pages={325--334},
  year={2018},
  publisher={Springer}
}

@article{martin2021google,
  title={{Google Scholar, Microsoft Academic, Scopus, Dimensions, Web of Science, and OpenCitations’ COCI: A Multidisciplinary Comparison of Coverage via Citations}},
  author={Mart{\'\i}n-Mart{\'\i}n, Alberto and Thelwall, Mike and Orduna-Malea, Enrique and Delgado L{\'o}pez-C{\'o}zar, Emilio},
  journal={Scientometrics},
  volume={126},
  number={1},
  pages={871--906},
  year={2021},
  publisher={Springer}
}

@techreport{airoldi2024inequality,
  title={{Inequality in Science: {W}ho Becomes a Star?}},
  author={Airoldi, Anna and Moser, Petra},
  year={2024},
  institution={National Bureau of Economic Research}
}

@article{kim2025women,
  title={{Women in Science. {L}essons from the {B}aby {B}oom}},
  author={Kim, Scott and Moser, Petra},
  journal={Econometrica},
  volume={93},
  number={5},
  pages={1521--1560},
  year={2025},
  publisher={Wiley Online Library}
}

@article{truffa2025undergraduate,
  title={{Undergraduate Gender Diversity and the Direction of Scientific Research}},
  author={Truffa, Francesca and Wong, Ashley},
  journal={American Economic Review},
  volume={115},
  number={7},
  pages={2414--2448},
  year={2025},
  publisher={American Economic Association 2014 Broadway, Suite 305, Nashville, TN 37203}
}

@article{koffi2026cassatts,
  title={{Cassatts in the Attic: {I}s There a Gender Gap in the Commercialization of Science}},
  author={Koffi, Marl{\`e}ne and Marx, Matt},
  journal={American Economic Journal: Applied Economics},
  volume={18},
  number={2},
  pages={266--298},
  year={2026},
  publisher={American Economic Association 2014 Broadway, Suite 305, Nashville, TN 37203-2425}
}
\bibliographystyle{apalike}
\end{singlespace}

\clearpage

\clearpage
\appendix
\setcounter{table}{0}
\setcounter{figure}{0}

\renewcommand{\thesection}{A\arabic{section}}
\renewcommand{\thetable}{A\arabic{table}}
\renewcommand{\thefigure}{A\arabic{figure}}

\section*{\hfil Appendix \hfil}

\section{Model Extension: Endogenous Location Choices}
\label{model_extension}

\subsection{Researchers' Location Choice}

For simplicity, researchers have homogeneous location preferences and value locations according to the following reduced-form indirect utility function:\footnote{In the stylized model, researchers have homogeneous location preferences, while non-researchers draw idiosyncratic location tastes that act as migration frictions. The asymmetry is deliberate: researcher location is the policy instrument we manipulate, not a response margin, whereas non-researchers' migration is the endogenous response at the center of the analysis. Productivity heterogeneity among researchers is consistent with this simplification: since individual productivity $D_i$ is separable from location productivity, it enters indirect utility as a location-invariant constant, so every researcher evaluates locations identically. We can therefore treat quality-adjusted researchers as solving a common location-choice problem.}
\begin{equation*}
    V^N_{j}=a_j(N_j) + b_j(N_j, M_j) - \beta^N_r r_j(M_j),
\end{equation*}
where 
\begin{equation*}
a_j(N_j) \equiv \ln A_j(N_j)=a_{0j} + \alpha \ln N_j, ~~~~a_{0j} \equiv \ln A_{0j}.
\end{equation*}
The term $b_j(N_j, M_j)$ captures local amenities for researchers, which include any non-productivity-related factors that affect researchers' valuation of a location; $r_j(M_j)$ is the local log rent, which reflects local cost of living. Research productivity may affect utility through several channels: More productive locations may offer higher salaries, greater funding opportunities, or stronger intrinsic research value and thus attract researchers through higher $a_j(N_j)$. The amenity term $b_j(N_j,M_j)$ captures both attractive and unattractive features of institutions and surrounding areas, such as institutional support and the availability of research positions. For simplicity, we assume that the local log rent $r_j(M_j)$ depends only on the non-research population, reflecting the assumption that researchers are a small share of the total population and therefore do not materially affect local living costs.

To capture spatial variation in the availability of research positions, we specify 
\begin{equation*}
  b_j(N_j,M_j)  = \beta^N_{0j} + \beta^N_{b} \ln\left( \frac{M_j}{N_j}\right),
\end{equation*}
where $\beta^N_{0j}$ captures exogenous variation in the availability of research positions and the desirability of research environments, driven by factors such as the historical presence of research universities, support for local institutions, and traditions of scholarly activity. The parameter $\beta^N_{b}$ governs how strongly the availability of research positions responds to local population. It is a reduced-form parameter that captures, for example, the responsiveness of local faculty positions to the population served by local universities. A larger $\beta^N_{b}$ implies that institutional demand for researchers is more responsive to the size of the local population. To ensure the existence of an interior equilibrium, we assume $\beta_b^N>\alpha$, so that researchers face net decreasing returns to scale in each location. Because researchers are assumed to move frictionlessly, their spatial distribution is determined by equalized utility across locations.

\subsection{Non-Researchers' Location Choice}

Non-researchers choose locations based on local wages, non-researcher amenities, and living costs. Based on the production function in each location, the local wage equals the marginal product of labor, which depends partially on local access to frontier knowledge. This is the key channel through which the spatial distribution of researchers affects non-researchers’ location choices.

Non-researchers are imperfectly mobile across locations. Their indirect utility is
\begin{equation*}
    V^M_{ij} = \left(1-s^N\right)s_{0j}+ s^{N} \ln\left[ S_0 + \left(\frac{N_j}{M_j} \right)^{\theta}\right] + b_j^M(M_j) - \beta^M_r r_{j}(M_j) + \varepsilon_{ij}, 
\end{equation*}
where the first two terms capture the location-varying component of log wages. Here, $s_{0j}\equiv \ln S_{0j}$ is the component of local productivity that is unrelated to exposure to frontier knowledge, while $S_0 + \left(N_j/M_j \right)^{\theta}$ measures access to frontier knowledge. The elasticity of local productivity with respect to $N_j/M_j$ is $s^{N} \theta \lambda_j$. The term $b_j^M(M_j)$ is location $j$’s amenity value for non-researchers, which may depend on the size of the local population (largely represented by the non-research population). As with researchers, non-researchers dislike higher log rent, with sensitivity governed by $\beta^M_r$. The term $\varepsilon_{ij}$ is an individual-specific idiosyncratic preference for location $j$, which equals $\sigma$ times a random variable drawn from the Type-I Extreme Value distribution. The parameter $\sigma$ governs the strength of idiosyncratic attachment to locations and therefore the degree of migration frictions.

The population of non-researchers in location $j$ is therefore
\begin{equation*}
    M_j=M \frac{\exp\left(\bar{V}_{j}/\sigma\right)}{\sum_{j'} \exp\left(\bar{V}_{j'}/\sigma \right)},
\end{equation*}
where $M$ is the national non-researcher population, and $\bar{V}_{j}$ is the deterministic component of indirect utility. 

The spatial equilibrium consists of the distributions $\{N_j, M_j\}$ and rents $r_j$ that clear local housing markets. With $J$ locations, the equilibrium and its comparative statics generally do not admit closed-form solutions. To illustrate how shocks to the spatial distribution of researchers affect knowledge access and knowledge production, Section \ref{two_locations} studies a tractable two-location case.

\subsection{Spatial Effects on Access to Frontier Knowledge: A Two-Location Analysis}
\label{two_locations}

To illustrate the model's implications, we consider a setting with two locations, 1 and 2. Location 1 resembles a city such as Boston: a long-established center of scholarship with a dense cluster of universities built up over centuries. Location 2 resembles a city such as Dallas: a fast-growing regional economic hub that has experienced substantial productivity and population growth in recent decades but lacks the same historical depth in higher education and research. Accordingly, we assume that location 1 has a higher location-specific amenity value for researchers, $\beta^N_{01}>\beta^N_{02}$, and that this advantage is sufficiently large to imply $N_1>N_2$. 

For tractability, we assume that local log rents take the form $r(M_j)=r_0 + e \ln (M_j)$, where $e$ is the inverse housing supply elasticity and is common across locations. To further simplify the equilibrium characterization, we set researchers' housing cost sensitivity to zero, $\beta^N_r=0$.\footnote{In this case, researchers' housing costs can be interpreted as being absorbed into the location-specific amenity value $\beta_{0j}^N$.} We also shut down endogenous amenities for non-researchers by setting $b^M_j(M_j)=0$.\footnote{We reintroduce these mechanisms in the quantitative evaluation of the trade-off later in the paper.}

Finally, because the local access shares $\lambda_j$ generally differ across locations, we evaluate these comparative statics at a common value $\bar\lambda$ when solving the system. This is exact at a symmetric equilibrium and a first-order approximation otherwise. 

\subsubsection{``Home-Market'' Effect}

A location with higher non-research productivity attracts more non-researchers. Through the endogenous provision of research positions, captured by $\beta^N_b$, a larger non-research population in turn attracts more researchers. In addition, because researcher productivity rises with local cluster size, a larger non-research population can generate a more than proportional increase in the number of researchers. As a result, the more productive and more populous location may also enjoy greater researcher access per capita.
This result echoes the home-market effect in trade models with agglomeration economies \citep{krugman1980scale} and is also related to the spatial concentration of medical services studied by \cite{dingel_medical2023}.\footnote{Places with larger non-research populations attract disproportionately more researchers, similar to the classic home-market result in which producers concentrate in locations with larger consumer markets. The mechanism here differs, however, because researchers generate local productivity spillovers rather than traded goods.} 

We summarize the effect of local non-research productivity, $s_{01}$, on the spatial distribution of researcher access in the following proposition.
\begin{proposition}
\label{prop_1}
Under Assumption \ref{assumption_1}, an increase in non-research productivity in location 1, $s_{01}$, raises $M_1$ and increases $N_1/M_1$ relative to $N_2/M_2$.
\end{proposition}
The proof of Proposition \ref{prop_1} is provided in Appendix \ref{appendix_s01}.

An increase in $s_{01}$ attracts non-researchers to location 1 through higher wages. The larger local population expands research positions, while agglomeration economies further attract researchers. Researcher employment therefore rises more than proportionally, increasing researchers per capita in location 1 relative to location 2. The model predicts that large population and industry centers have both more researchers and more researchers per capita.
 
\subsubsection{Spatial Inertia of the Historical Legacy of Universities and Research Institutions} 

In the postwar period, especially after the 1970s, U.S. economic activity shifted substantially away from legacy cities in the Northeast and Midwest toward the Sun Belt. Based on Proposition \ref{prop_1}, this shift in productivity and population toward booming cities in the South and West should, all else equal, also attract researchers and gradually narrow differences in researcher access per capita across locations.\footnote{In the model, a large increase in non-research productivity in location 2, such as Dallas, raises its non-research population, which then attracts more researchers. As researchers move in, agglomeration effects further increase local researcher productivity, reinforcing this process. All else equal, per-capita access to researchers in location 2 should therefore increase through the home-market effect.} Yet 
universities and research institutions remained remarkably spatially inert and did not follow population and economic activity into these expanding locations.

In the model, this persistence is represented by a higher researcher amenity $\beta^N_{01}$ relative to $\beta^N_{02}$, reflecting the anchoring of research positions in legacy cities beyond what population or productivity alone would predict. 
Several forces may contribute to this anchoring. Alumni networks and institutional ties deepen over time and create persistent advantages for existing universities: donor relationships and academic networks are difficult to replicate in new locations \citep{hoxby2012endowment, blair2021elite}. In addition, the fixed costs of building a modern research university from scratch are substantial compared with the land-grant era, when federal land transfers and infrastructure requirements were smaller in scale \citep{ehrenberg2003startup}.

We summarize the effect of a higher value of $\beta^N_{01}$ on spatial inequality in research access as follows:
\begin{proposition}
\label{prop_2}
Under Assumption \ref{assumption_1}, an increase in the researcher amenity in location 1, $\beta_{01}^N$, lowers $N_2/M_2$ relative to $N_1/M_1$.
\end{proposition}
The proof of Proposition \ref{prop_2} is provided in Appendix \ref{appendix_beta01}. 

\subsubsection{Moving Research Access Toward Equality}
\label{equalizing_access}

We next consider a policy that reallocates research resources toward an allocation proportional to population. In the model, this operates through an increase in $\beta^N_b$, which governs how strongly the provision of local research positions responds to local population. The following proposition summarizes the effect of an increase in $\beta^N_b$.

\begin{proposition}
\label{prop_3}
Under Assumption \ref{assumption_1}, if researcher access per capita is initially higher in location 1 than in location 2, an increase in the linkage between the provision of local researcher positions and local population, $\beta_b^N$, narrows the researcher access gap by raising $N_2/M_2$ relative to $N_1/M_1$.
\end{proposition}
The proof of Proposition \ref{prop_3} is provided in Appendix \ref{appendix_betaNb}. 

The persistent attraction of legacy cities for researchers can lead to unequal access to frontier knowledge, while policies that direct research resources toward population centers can help reduce this inequality. However, such policies may shrink existing research clusters and weaken the research productivity gains generated by their scale, particularly when those clusters are large and highly productive. In addition, as knowledge access improves in location 2, non-researchers may adjust their own location choices endogenously. We now characterize the costs and benefits of increasing $\beta_b^N$, taking into account all of these adjustments.

\paragraph{Costs and Benefits of Reallocating Researchers in Proportion to Population} 
Taking the derivative of national output with respect to $\beta^N_b$, we show that the effect on national output is proportional to the following expression:\footnote{See Appendix \ref{appendix_y} for the full derivation.}
{\small \begin{align*}
   \frac{\partial y}{\partial \beta^N_b} \propto& P_1^N(1-P_1^N)N \bigg[\underbrace{s^N\left(\frac{Y_2}{Y}\frac{\theta\lambda_2}{N_2} - \frac{Y_1}{Y}\frac{\theta\lambda_1}{N_1}\right)}_{\text{Gain in Knowledge Access}} \underbrace{- \frac{s^N(1+\alpha)}{(1-\eta) \delta Q^{1-\eta}} \left(A_{01} N^\alpha_{1} - A_{02} N^\alpha_{2}\right)}_{\text{Loss in Knowledge Production}} \bigg] \\
    & \underbrace{-\;\frac{s^N\theta\bar\lambda\,\left(1 - s^N\theta\bar\lambda\right)}{B\left(\sigma + s^N\theta\bar\lambda + \beta^M_r e\right)}\, P_1^M\left(1-P_1^M\right)M \left(S_{01}^{1-s^N}\left(S_0 + \left(\frac{N_1}{M_1}\right)^{\theta}\right)^{s^N}-S_{02}^{1-s^N}\left(S_0 + \left(\frac{N_2}{M_2}\right)^{\theta}\right)^{s^N}
\right)}_{\text{Loss in Productivity due to Non-Researcher Migration}},
\end{align*}}
where $N$ and $M$ are national populations of researchers and non-researchers, respectively. $P_1^N=N_1/N$ and $P_1^M=M_1/M$.
The first two terms capture the effects of researcher migration. As researchers move from location 1 to location 2, access to frontier knowledge for non-researchers rises in location 2 and falls in location 1. Because knowledge access is concave in the number of researchers, holding frontier knowledge production fixed, equalizing access across locations raises aggregate output. The second term captures the offsetting effect on knowledge production. As researchers move out of location 1, they forgo the high research productivity of that location, and the contraction of the larger cluster further lowers productivity by weakening agglomeration effects. Because knowledge production is convex in the number of researchers, the productivity loss from shrinking location 1 exceeds the productivity gain from expanding location 2, so total knowledge production declines.

The third term captures the endogenous migration response of non-researchers. As researcher reallocation raises productivity in location 2 and lowers it in location 1, some non-researchers move toward location 2. Because location 2 remains the lower-productivity location overall, this migration can partially offset the gains from improved knowledge access.
\subsection{Proofs for the Comparative Statics in the Two-Location Case}

This section derives the comparative statics for the two-location case and proves Propositions \ref{prop_1}--\ref{prop_3}. 

\subsubsection{Effects of $s_{01}$ on the Spatial Distributions of Researchers and Non-Researchers: The ``Home-Market'' Effect}
\label{appendix_s01}

We derive the effect of $s_{01}$ on researcher and non-researcher populations. Because researchers are homogeneous and move frictionlessly, their equilibrium utility is equalized across the two locations. Researcher utility depends on research productivity and amenities rather than on knowledge access. The utility-equalization condition is:
\begin{align*}\label{utility_dif}
    0 &= (a_{01}-a_{02}) + (\beta^N_{01} - \beta^N_{02}) \\
    &+ \left(\alpha - \beta^N_{b} \right) \left[\ln N_1 -\ln (N-N_1)\right] + \beta^N_{b}\left[\ln M_1 -\ln (M-M_1)\right]. 
\end{align*}
Totally differentiating this equation yields 
\begin{equation*}\label{dif_equation_r}
(\alpha - \beta^N_b)\left(\frac{1}{N_1} + \frac{1}{N_2}\right)\frac{\partial N_1}{\partial s_{01}}  + \beta^N_{b}\left(\frac{1}{M_1} + \frac{1}{M_2}\right)\frac{\partial M_1}{\partial s_{01}} = 0.
\end{equation*}

We then totally differentiate non-researcher population at location 1. Because a non-researcher's real income depends on local access $S_j$, the access elasticities $\theta\lambda_j$ enter this row; evaluating them at $\bar\lambda$ gives 
\begin{equation*}\label{dif_equation_nr}
\left(\sigma + s^N\theta\bar\lambda + \beta^M_r e\right)\left(\frac{1}{M_1} + \frac{1}{M_2}\right)\frac{\partial M_1}{\partial s_{01}} - s^N\theta\bar\lambda\left(\frac{1}{N_1} + \frac{1}{N_2}\right)\frac{\partial N_1}{\partial s_{01}} = 1-s^N.
\end{equation*}
The right-hand side remains $1-s^N$ because the shock $s_{01}$ operates through the location-specific non-research factor $S_{01}$ and not through the access function, so it carries no $\bar\lambda$.

Solving this two-equation system yields:
\begin{equation*}
    \left(\frac{1}{N_1} + \frac{1}{N_2}\right)\frac{\partial N_1}{\partial s_{01}} = \frac{ \beta^N_b(1-s^N) }{(\beta^N_b - \alpha)(\sigma + \beta^M_r e) - \alpha s^N\theta\bar\lambda }
\end{equation*}
and
\begin{equation*}
    \left(\frac{1}{M_1} + \frac{1}{M_2}\right)\frac{\partial M_1}{\partial s_{01}} = \frac{(1-s^N)(\beta^N_b - \alpha)}{(\beta^N_b - \alpha)(\sigma + \beta^M_r e) - \alpha s^N\theta\bar\lambda }.
\end{equation*}
Since $\frac{\partial\ln\left(\frac{N_1}{N_2}\right)}{\partial s_{01}}  = \left(\frac{1}{N_1} + \frac{1}{N_2}\right) \frac{\partial N_1}{\partial s_{01}}$ and $\frac{\partial \ln\left(\frac{M_1}{M_2}\right)}{\partial s_{01}}  = \left(\frac{1}{M_1} + \frac{1}{M_2}\right) \frac{\partial M_1}{\partial s_{01}}$, we can rewrite the comparative statics as: 
\begin{equation}
\label{derivative_N_s01}
    \frac{\partial \ln\left(\frac{N_1}{N_2}\right)}{\partial s_{01}}  = \frac{ \beta^N_b(1-s^N) }{(\beta^N_b - \alpha)(\sigma + \beta^M_r e) - \alpha s^N\theta\bar\lambda }
\end{equation}
and
\begin{equation}
\label{derivative_M_s01}
     \frac{\partial \ln\left(\frac{M_1}{M_2}\right)}{\partial s_{01}} = \frac{(1-s^N)(\beta^N_b - \alpha)}{(\beta^N_b - \alpha)(\sigma + \beta^M_r e) - \alpha s^N\theta\bar\lambda }.
\end{equation}

Given these derivatives, we can characterize how a productivity shock in location 1 affects the distribution of per-capita access to researchers: 
\begin{equation}
\label{derivative_s01}
     \frac{\partial \ln\left(\frac{N_1/M_1}{N_2/M_2}\right)}{\partial s_{01}} = \frac{ \alpha(1-s^N)}{(\beta^N_b - \alpha)(\sigma + \beta^M_r e) - \alpha s^N\theta\bar\lambda }.
\end{equation}

A sufficient condition for a stable interior equilibrium is that the common denominator in Equations \ref{derivative_N_s01}--\ref{derivative_s01} be positive:
\begin{assumption}
\label{assumption_1}
The parameters satisfy the following condition:
\begin{equation*}
    (\beta^N_b - \alpha)(\sigma + \beta^M_r e) - \alpha s^N\theta\bar\lambda > 0.
\end{equation*}
\end{assumption} 
This condition imposes net decreasing returns to scale in each location. It holds as long as the agglomeration elasticity $\alpha$ is not too large, implying that $\beta^N_b - \alpha > 0$ must hold.
\setcounter{proposition}{0}
\begin{proposition}
Under Assumption \ref{assumption_1}, an increase in non-research productivity in location 1, $s_{01}$, raises $M_1$ and increases $N_1/M_1$ relative to $N_2/M_2$.
\end{proposition}
\begin{proof}
Under Assumption \ref{assumption_1}, the comparative static provided by equation \ref{derivative_M_s01} shows that an increase in $s_{01}$ leads to an increase in the population share of non-researchers in location 1 if $\beta^N_b >\alpha$. Since $\beta^N_b >\alpha$ is implied from Assumption \ref{assumption_1}, we know that an increase in $s_{01}$ must lead to an increase in $M_1$. 

Under Assumption \ref{assumption_1}, the comparative static provided by equation \ref{derivative_s01} shows that location 1's per-capita access to researchers by non-researchers increases relative to that in location 2 in response to an increase in $s_{01}$. 
\end{proof}

\subsubsection{Effects of $\beta^N_{01}$ on Research and Non-Research Population Distributions}
\label{appendix_beta01}

Next, we examine how a change in the location-specific amenity value for researchers, $\beta^N_{01}$, affects the spatial distribution of both researcher and non-researcher populations: $\frac{\partial\ln\left(\frac{N_1}{N_2}\right) }{\partial \beta^N_{01}}$ and $\frac{\partial\ln\left(\frac{M_1}{M_2}\right) }{\partial \beta^N_{01}}$. 

Totally differentiating the researchers’ utility-equalization condition yields:
\begin{equation*}
    (\alpha - \beta^N_b)\left(\frac{1}{N_1} + \frac{1}{N_2}\right)\frac{\partial N_1}{\partial \beta^N_{01}} + \beta^N_{b}\left(\frac{1}{M_1} + \frac{1}{M_2}\right)\frac{\partial M_1}{\partial \beta^N_{01}} + 1 =0.
\end{equation*}
Totally differentiating the equilibrium condition for the non-researcher population in location 1, again evaluating the access elasticities at $\bar\lambda$, yields: 
\begin{equation*}
    \left(\sigma + s^N\theta\bar\lambda + \beta^M_r e\right)\left(\frac{1}{M_1} + \frac{1}{M_2}\right)\frac{\partial M_1}{\partial \beta^N_{01}} - s^N\theta\bar\lambda\left(\frac{1}{N_1} + \frac{1}{N_2}\right)\frac{\partial N_1}{\partial \beta^N_{01}}=0.
\end{equation*}
Solving the two equations for $\frac{\partial M_1}{\partial \beta^N_{01}}$ and $\frac{\partial N_1}{\partial \beta^N_{01}}$ yields:

\begin{equation*}
    \left(\frac{1}{N_1} + \frac{1}{N_2}\right)\frac{\partial N_1}{\partial \beta^N_{01}} = \frac{\sigma + s^N\theta\bar\lambda + \beta^M_r e}{(\beta^N_b - \alpha)(\sigma + \beta^M_r e) - \alpha s^N\theta\bar\lambda }
\end{equation*}
and
\begin{equation*}
    \left(\frac{1}{M_1} + \frac{1}{M_2}\right)\frac{\partial M_1}{\partial \beta^N_{01}} = \frac{s^N\theta\bar\lambda}{(\beta^N_b - \alpha)(\sigma + \beta^M_r e) - \alpha s^N\theta\bar\lambda}.
\end{equation*}

Since $\frac{\partial\ln\left(\frac{N_1}{N_2}\right)}{\partial \beta^N_{01}}  = \left(\frac{1}{N_1} + \frac{1}{N_2}\right) \frac{\partial N_1}{\partial \beta^N_{01}}$ and $\frac{\partial \ln\left(\frac{M_1}{M_2}\right)}{\partial \beta^N_{01}}  = \left(\frac{1}{M_1} + \frac{1}{M_2}\right) \frac{\partial M_1}{\partial \beta^N_{01}}$, we arrive at the following: 

\begin{equation}
\label{derivative_N_beta01}
    \frac{\partial \ln\left(\frac{N_1}{N_2}\right)}{\partial \beta^N_{01}}  = \frac{\sigma + s^N\theta\bar\lambda + \beta^M_r e}{(\beta^N_b - \alpha)(\sigma + \beta^M_r e) - \alpha s^N\theta\bar\lambda }.
\end{equation}
and
\begin{equation}
\label{derivative_M_beta01}
     \frac{\partial \ln\left(\frac{M_1}{M_2}\right)}{\partial \beta^N_{01}} = \frac{s^N\theta\bar\lambda}{(\beta^N_b - \alpha)(\sigma + \beta^M_r e) - \alpha s^N\theta\bar\lambda }.
\end{equation}
Combining these derivatives gives the effect of $\beta^N_{01}$ on researcher access per capita in location 2 relative to location 1:
\begin{equation}
\label{derivative_beta01}
     \frac{\partial \ln\left(\frac{N_2/M_2}{N_1/M_1}\right)}{\partial \beta^N_{01}} = \frac{-(\sigma+\beta^M_r e)}{(\beta^N_b - \alpha)(\sigma + \beta^M_r e) - \alpha s^N\theta\bar\lambda}.
\end{equation}

\begin{proposition}
Under Assumption \ref{assumption_1}, an increase in the researcher amenity in location 1, $\beta_{01}^N$, lowers $N_2/M_2$ relative to $N_1/M_1$.
\end{proposition}

\begin{proof}
Under Assumption 1, we can easily see that  $\frac{\partial \ln\left(\frac{N_2/M_2}{N_1/M_1}\right)}{\partial \beta^N_{01}}<0$, which proves the proposition. 
\end{proof}

\subsubsection{Effects of $\beta^N_{b}$ on Research and Non-Research Population Distributions}
\label{appendix_betaNb}
Following the same steps as above, we obtain the effects of $\beta^N_b$ on $N_1$ and $M_1$:
\begin{equation*}
    \left(\frac{1}{N_1}+\frac{1}{N_2}\right)\frac{\partial N_1}{\partial \beta^N_b} = -\frac{(\sigma + s^N\theta\bar\lambda + \beta^M_r e)\ln\left(\frac{N_1/N_2}{M_1/M_2}\right)}{(\beta^N_b - \alpha)(\sigma + \beta^M_r e) - \alpha s^N\theta\bar\lambda }
\end{equation*}
and
\begin{equation*}
    \left(\frac{1}{M_1}+\frac{1}{M_2}\right)\frac{\partial M_1}{\partial \beta^N_b} = -\frac{s^N\theta\bar\lambda\ln\left(\frac{N_1/N_2}{M_1/M_2}\right)}{(\beta^N_b - \alpha)(\sigma + \beta^M_r e) - \alpha s^N\theta\bar\lambda }
\end{equation*}
Combining these expressions gives the effect of $\beta^N_b$ on researcher access per capita in location 1 relative to location 2:
\begin{equation}
\label{derivative_betaNb}
     \frac{\partial \ln\left(\frac{N_1/M_1}{N_2/M_2}\right)}{\partial \beta^N_{b}} = \frac{-(\sigma+\beta^M_r e)\ln\left(\frac{N_1/M_1}{N_2/M_2}\right)}{(\beta^N_b - \alpha)(\sigma + \beta^M_r e) - \alpha s^N\theta\bar\lambda }.
\end{equation}
Equation \ref{derivative_betaNb} is negative when location 1 initially has higher researcher access per capita. Thus, an increase in $\beta^N_b$ narrows the access gap between the two locations.

\begin{proposition}
Under Assumption \ref{assumption_1}, if researcher access per capita is initially higher in location 1 than in location 2, an increase in the linkage between the provision of local researcher positions and local population, $\beta_b^N$, narrows the researcher access gap by raising $N_2/M_2$ relative to $N_1/M_1$.
\end{proposition}

\begin{proof}
    Under Assumption \ref{assumption_1}, $\frac{\partial \ln\left(\frac{N_1/M_1}{N_2/M_2}\right)}{\partial \beta^N_{b}}<0$ if the per-capita access to researchers is higher in location 1, and vice versa if the per-capita access to researchers is higher in location 2. This means that $\beta^N_b$ equalizes per-capita access to researchers across the two locations. 
\end{proof}

\subsubsection{The Effect of $\beta_b^N$ on National Output}
\label{appendix_y}

Totally differentiating log national output gives $\frac{\partial y}{\partial \beta_b^N}$ as a function of $\frac{\partial N_1}{\partial \beta_b^N}$ and $\frac{\partial M_1}{\partial \beta_b^N}$. The marginal effect of a researcher on log output carries the access elasticity $\theta\lambda_j$, while the marginal effect of a non-researcher carries the marginal-to-average product wedge $1-s^N\theta\lambda_j$; we evaluate both at $\bar\lambda$. The access levels take the additive form, with per-non-researcher access productivity $S_{0j}^{1-s^N}S_j^{s^N}$.

The effect of $\beta^N_b$ on log national output can be written as follows: 
{\small \begin{align*}
    \frac{\partial y}{\partial\beta^N_b} &= \left[s^N\left(\frac{Y_1}{Y}\frac{\theta\lambda_1}{N_1}-\frac{Y_2}{Y}\frac{\theta\lambda_2}{N_2} \right) + \frac{s^N(1+\alpha)}{(1-\eta) \delta Q^{1-\eta}} \left(A_{01} N^\alpha_{1} - A_{02} N^\alpha_{2}\right)\right]\frac{\partial N_1}{\partial \beta^N_b} \\
    & + \frac{1 - s^N\theta\bar\lambda}{B}\left(S_{01}^{1-s^N} S_1^{s^N} - S_{02}^{1-s^N} S_2^{s^N}\right)\frac{\partial M_1}{\partial \beta^N_b}.
\end{align*}}
Substituting the comparative statics for researchers and non-researchers and rearranging yields:
{\small \begin{align*}
   \frac{\partial y}{\partial \beta^N_b} \propto& P_1^N(1-P_1^N)N \bigg[\underbrace{s^N\left(\frac{Y_2}{Y}\frac{\theta\lambda_2}{N_2} - \frac{Y_1}{Y}\frac{\theta\lambda_1}{N_1}\right)}_{\text{Gain in Knowledge Access}} \underbrace{- \frac{s^N(1+\alpha)}{(1-\eta) \delta Q^{1-\eta}} \left(A_{01} N^\alpha_{1} - A_{02} N^\alpha_{2}\right)}_{\text{Loss in Knowledge Production}} \bigg] \\
    & \underbrace{-\;\frac{s^N\theta\bar\lambda\,\left(1 - s^N\theta\bar\lambda\right)}{B\left(\sigma + s^N\theta\bar\lambda + \beta^M_r e\right)}\, P_1^M\left(1-P_1^M\right)M \left(S_{01}^{1-s^N}\left(S_0 + \left(\frac{N_1}{M_1}\right)^{\theta}\right)^{s^N}-S_{02}^{1-s^N}\left(S_0 + \left(\frac{N_2}{M_2}\right)^{\theta}\right)^{s^N}
\right)}_{\text{Loss in Productivity due to Non-Researcher Migration}}.
\end{align*}}

\section{Comparison of Microsoft Academic Graph with Other Bibliographic Databases}
\label{app:MAG}
We use the Microsoft Academic Graph (MAG) because it combines broad cross-disciplinary coverage with linked information on publications, citations, authors, institutional affiliations, and fields of study—all of which are necessary to construct our annual researcher panel. An additional advantage is that MAG permits free bulk access, making the construction of our sample more transparent and reproducible.

Existing comparisons support the comprehensiveness and reliability of MAG. \cite{visser2021large} compare MAG with other major bibliographic data sources, including Scopus, Web of Science, Dimensions, and Crossref, and find that MAG provides more comprehensive coverage of scientific publications than the other databases, with most records unique to MAG corresponding to scientific work. While Scopus includes more journal articles and conference proceedings in some subjects, MAG exhibits the greatest overlap with Scopus relative to other sources. 

In terms of citation measures, \cite{thelwall2018does} finds that Microsoft Academic is essentially equivalent to Scopus in terms of its citation counts, making it a robust source of free citation data. \cite{martin2021google} show that in most subjects, MAG retrieves more citations than Scopus and WoS, although it has coverage gaps in some fields, particularly physics and parts of the humanities.
\cite{visser2021large} show that across all databases, documents with a higher number of citations are disproportionately represented in the overlap between data sources. Thus, differences between data sources may not be consequential in terms of the relative ``importance'' of publications \citep{visser2021large}. These findings suggest that MAG provides broadly comparable citation measures to those of other bibliographic data sources.

\cite{wang2020microsoft} describe MAG as replicating ``the success of Google Scholar, which utilizes the massive document index from a web search engine to achieve comprehensive coverage of contemporary scholarly materials, many of which are not published and distributed through traditional channels and not assigned DOIs.'' 
MAG has consequently been used in recent economics research on scientific careers, inequality, and innovation, such as \cite{airoldi2024inequality}, \cite{kim2025women}, \cite{truffa2025undergraduate}, and \cite{koffi2026cassatts}. 

\section{Back-of-the-Envelope Calculation: Gains from Improved Access to Knowledge}
\label{app:back_of_envelope}

The evidence in Section 4 shows that research activity has become increasingly misaligned with population, suggesting potential gains from bringing research closer to where people live and work. This appendix uses the spatial framework in Section 2 to quantify the access channel alone. Specifically, we ask how much aggregate output would increase if researchers were distributed across metropolitan areas in proportion to population, holding the total stock of frontier knowledge fixed.

Using Equation \eqref{eq:aggregate_production}, we reallocate researchers so that each MSA’s share of researchers equals its share of the national population. We measure non-research productivity using residualized wages from the 2015–2019 ACS and hold it fixed at its baseline value. We set $\lambda_j=\bar{\lambda}=0.5$, so that local and non-local channels each account for half of knowledge access. We set $s^N=0.2$, consistent with estimates of the aggregate elasticity of productivity with respect to frontier knowledge. We calibrate the local access elasticity $s^N\theta\bar{\lambda}$ to $0.08$, matching the magnitude of local university spillovers estimated by \citet{kantor2014knowledge}. 

Under these assumptions, reallocating researchers in proportion to population raises aggregate output by 5.3\%. When geographic access is instead measured using research papers per capita, the implied gain rises to 8.8\%. These calculations illustrate that the access benefits of a more geographically balanced allocation of research activity could be economically meaningful.

The exercise, however, captures only the gross gain through improved knowledge access. It holds the stock of frontier knowledge fixed and therefore does not account for the possibility that reallocation moves researchers away from productive institutions, reduces research-cluster sizes, and lowers knowledge production. Section 6 incorporates these production-side effects and evaluates the net effect of researcher reallocation.

\section{Counterfactual Analysis: Migration of Non-Researchers}
\label{app:nonresearcher}
To simulate the migration of the non-research population, we assume a standard multinomial logit framework in which the change in per-capita researchers shifts the local marginal product of labor. We set the unobserved idiosyncratic location taste dispersion based on \cite{hsieh_moretti2019}. We set the counterfactual wage in each MSA equal to this counterfactual marginal product of labor, and, given this wage shock, recompute each MSA's population; the resulting population serves as the counterfactual non-research population used to compute the counterfactual access term $\ln \widehat{B}$. We hold housing rents fixed with respect to the change in local wages. Because a rent response would dampen migration, holding rents fixed makes our counterfactual migration response an upper bound: if even this upper-bound response generates no sizable efficiency cost, the true cost is likely smaller still.

\clearpage

\begin{figure}[!h]
     \centering
        \caption{Shares of Population, Researchers, and Publications in the Ten Largest Research Clusters} 
     \begin{subfigure}[b]{0.9\linewidth}
         \centering
         \includegraphics[width=\linewidth]{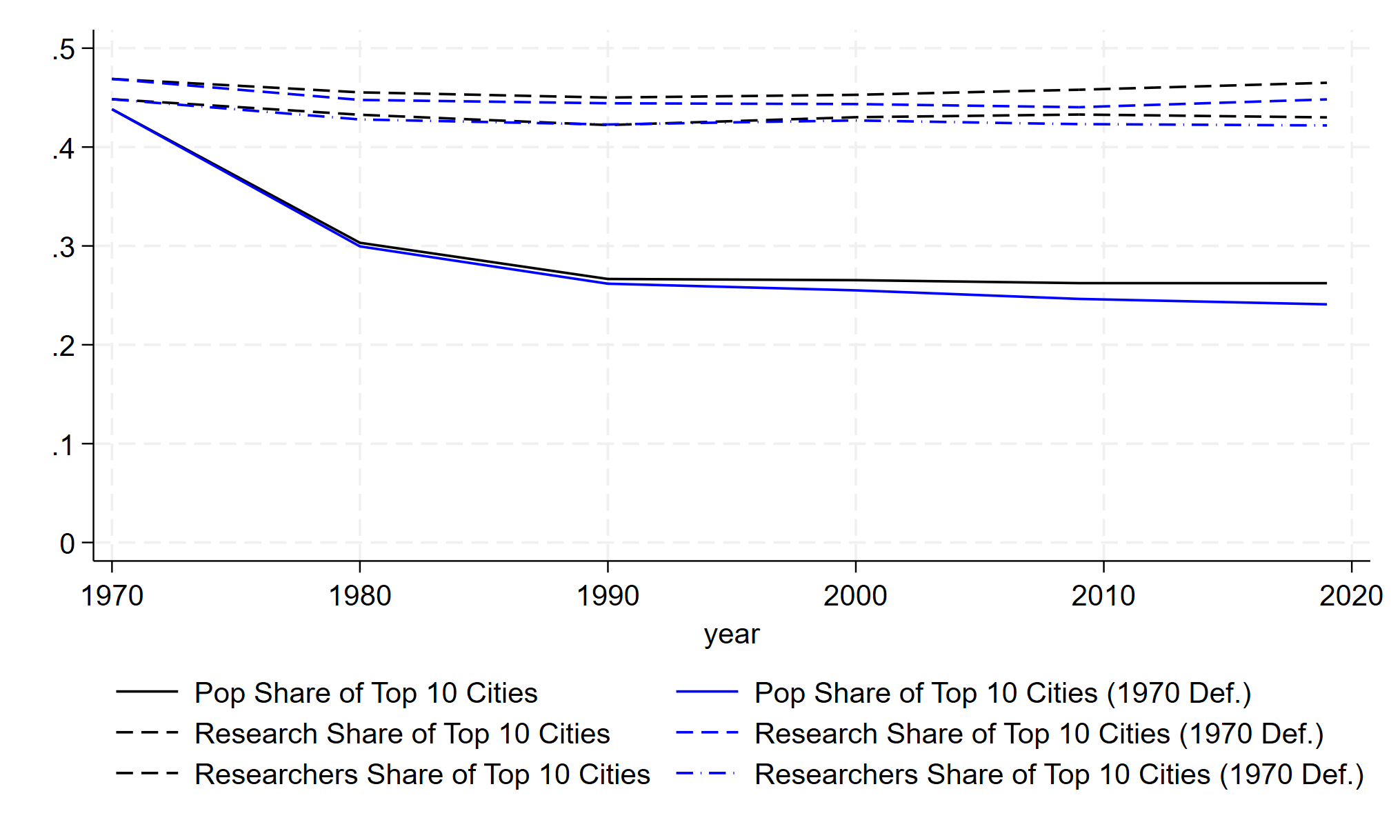}
     \end{subfigure}

\begin{minipage}{\textwidth}
\footnotesize{{\it Notes:} The figure plots the shares of U.S. population, researchers, and publications that the ten largest research clusters accounted for from 1970 to 2019. The black lines allow the set of ten largest clusters to change over time, while the blue lines hold the set fixed at the ten largest clusters in 1970.}
\end{minipage}
\label{fig:top_10_cities}
\end{figure}

\clearpage

\begin{figure}[!h]
     \centering
        \caption{Shares of Researcher Affiliations by Institution Type and Field} 
     \begin{subfigure}[b]{0.95\linewidth}
         \centering
         \includegraphics[width=\linewidth]{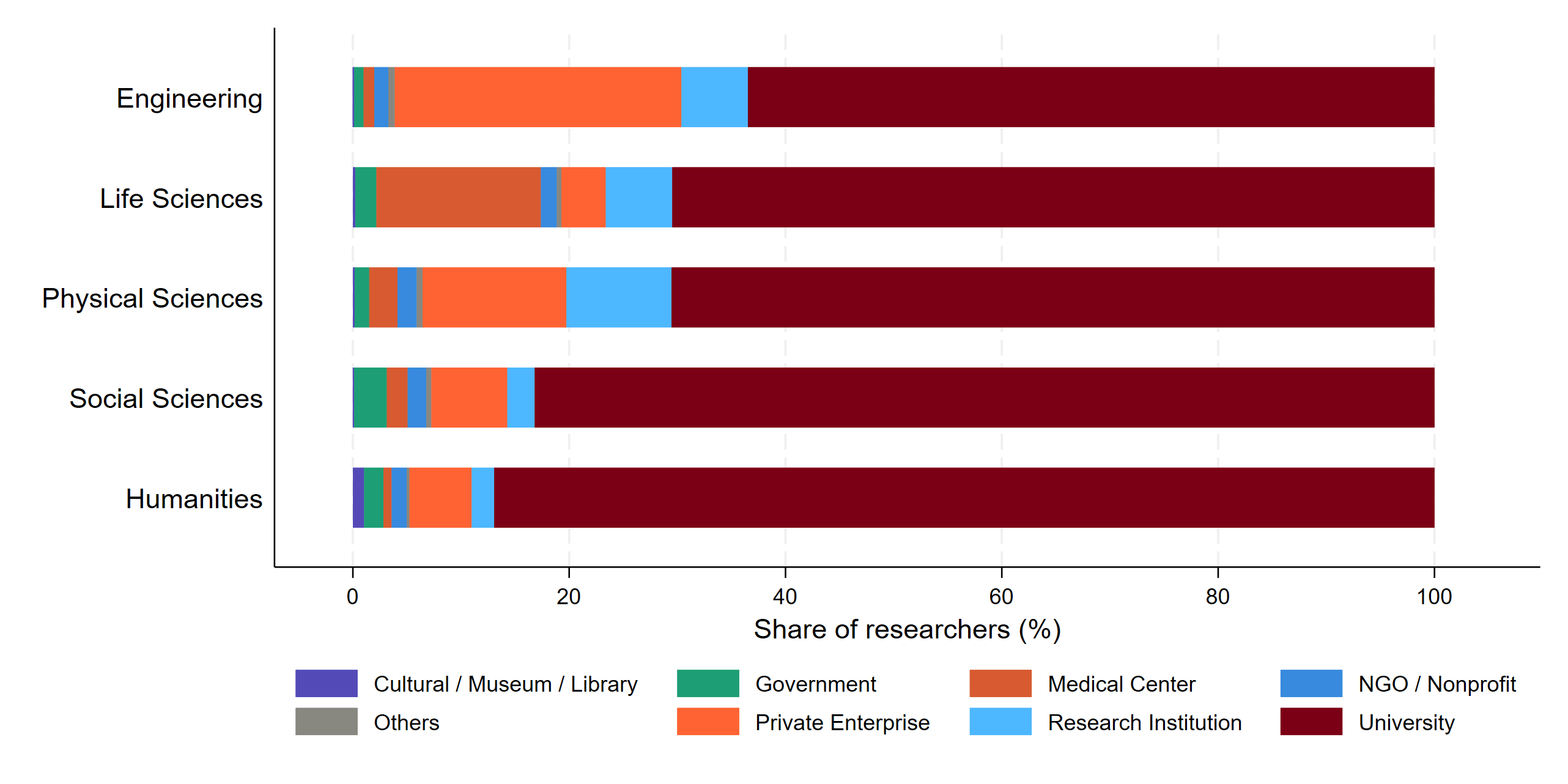}
      \caption{\centering Shares by Institution Type and Field}
             \label{fig:researcher_share}
     \end{subfigure}
   \\
        \begin{subfigure}[b]{0.8\linewidth}
         \centering
         \includegraphics[width=\linewidth]{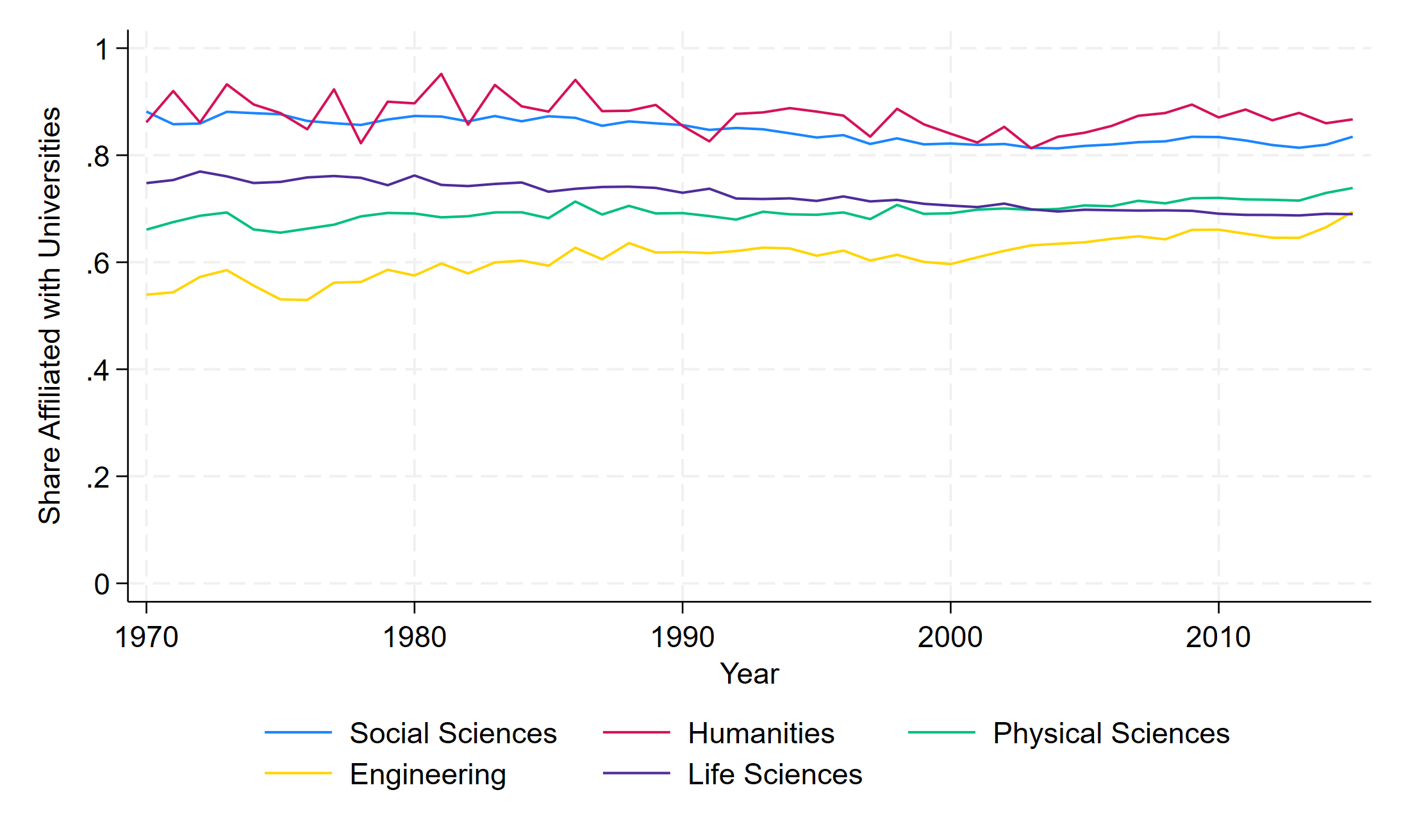}
     \caption{\centering Shares of University Affiliations over Time}
        \label{fig:univerity_share}
     \end{subfigure}
\begin{minipage}{\textwidth}
\footnotesize{{\it Notes:} We use Claude AI to classify research institutions into broad types based on the names of the institutions. Panel (a) shows, for each broad field, the share of researcher affiliations associated with each institution type, pooling publication data from 1970 to 2015. The unit of observation is a researcher–institution affiliation appearing in a publication in a given year. Panel (b) plots the annual shares of researcher affiliations associated with universities within each broad field.
}
\end{minipage}
        \label{fig:institution_type}
\end{figure}

\clearpage

\begin{figure}[!h]
     \centering
        \caption{Rising Inequality in Access to Research Papers: Universities vs. Non-Universities} 
     \begin{subfigure}[b]{0.7\linewidth}
         \centering
         \includegraphics[width=\linewidth]{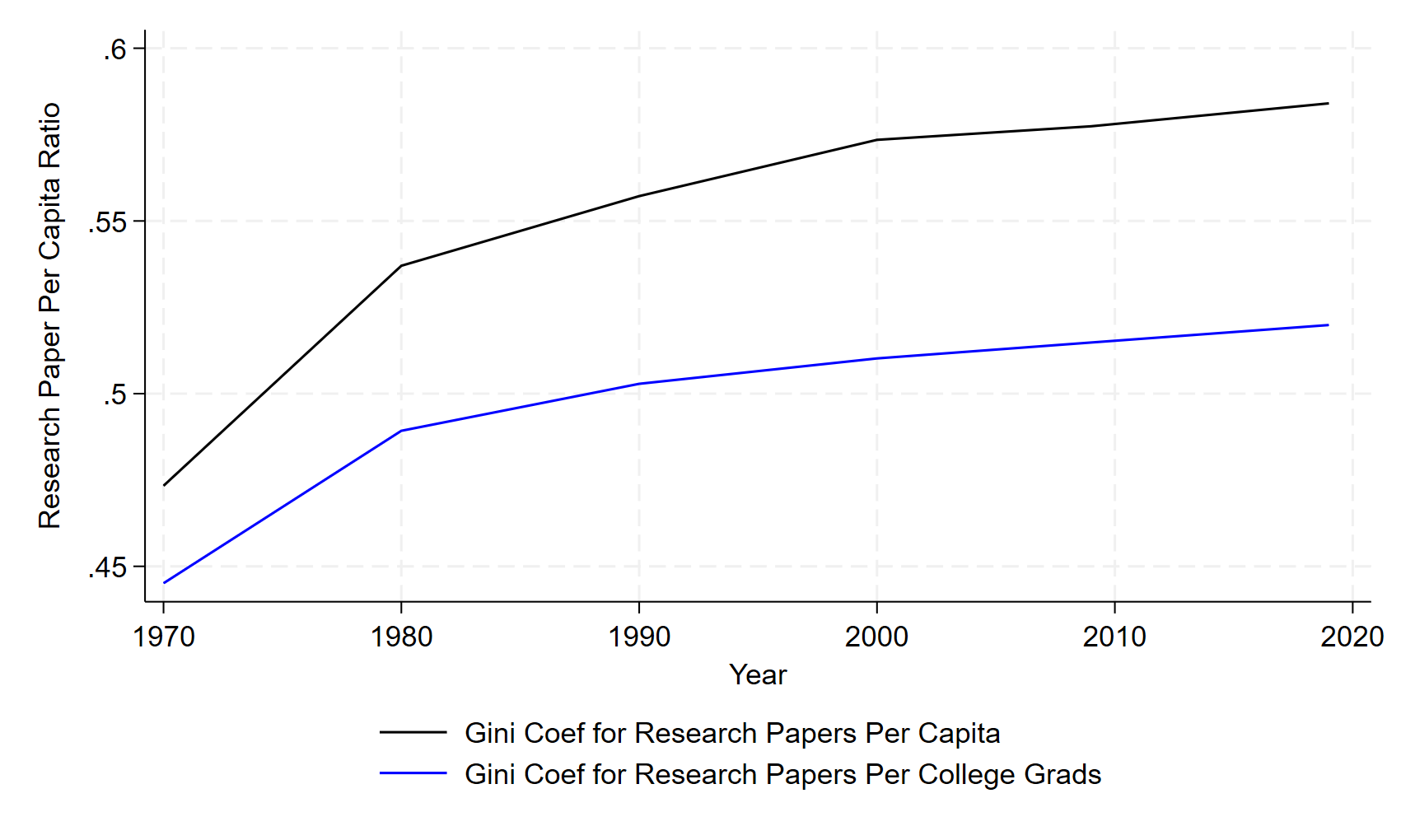}
      \caption{\centering University Research}
             \label{fig:gini_univ}
     \end{subfigure}
   \\
        \begin{subfigure}[b]{0.7\linewidth}
         \centering
         \includegraphics[width=\linewidth]{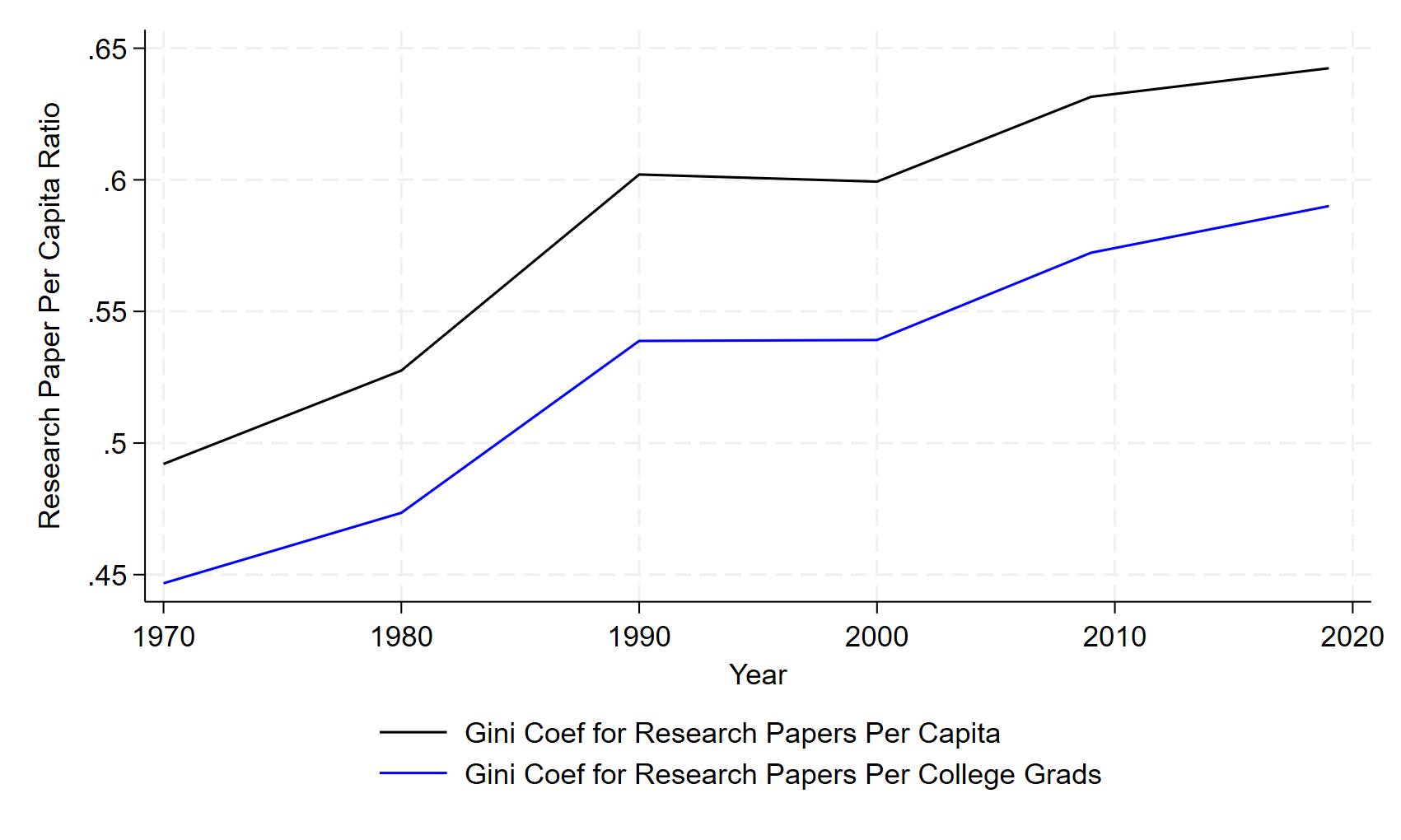}
     \caption{\centering Non-University Research}
        \label{fig:gini_nonuniv}
     \end{subfigure}
\begin{minipage}{\textwidth}
\footnotesize{{\it Notes:} Panel (a) plots population-weighted Gini coefficients of access to research papers produced by university-affiliated authors across MSAs, while Panel (b) reports the corresponding coefficients for non-university-affiliated authors. The black series use papers per resident, and the blue series use papers per college graduate.}
\end{minipage}
        \label{fig:gini_uni_nonuni}
\end{figure}

\clearpage 

\begin{figure}[!h]
    \centering
    \caption{Gini Coefficients of Research Papers per Resident by Field, 1970--2019} 
    \includegraphics[width=0.8\linewidth]{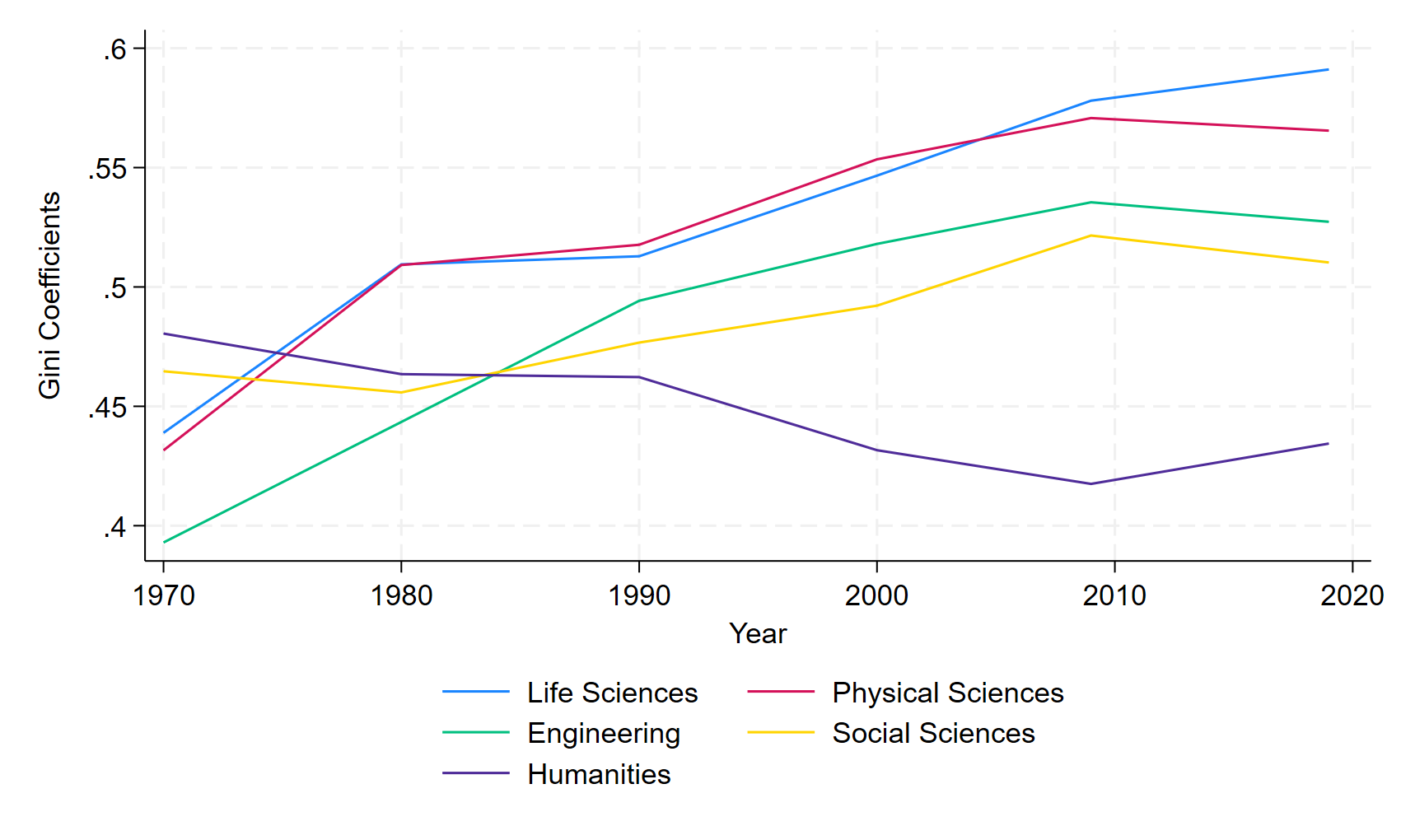}
    \label{fig:gini_field}
    \vspace{0.25cm}
\begin{minipage}{\textwidth}
\footnotesize{{\it Notes:} The figure plots population-weighted Gini coefficients of research papers per resident across MSAs from 1970 to 2019, separately for five broad fields.}
\end{minipage}
\end{figure}

\clearpage

\begin{figure}[!h]
     \centering
        \caption{Distributions of Access to Research Papers Across MSAs, 1970, 1990, and 2019} 
     \begin{subfigure}[b]{0.7\linewidth}
         \centering
         \includegraphics[width=\linewidth]{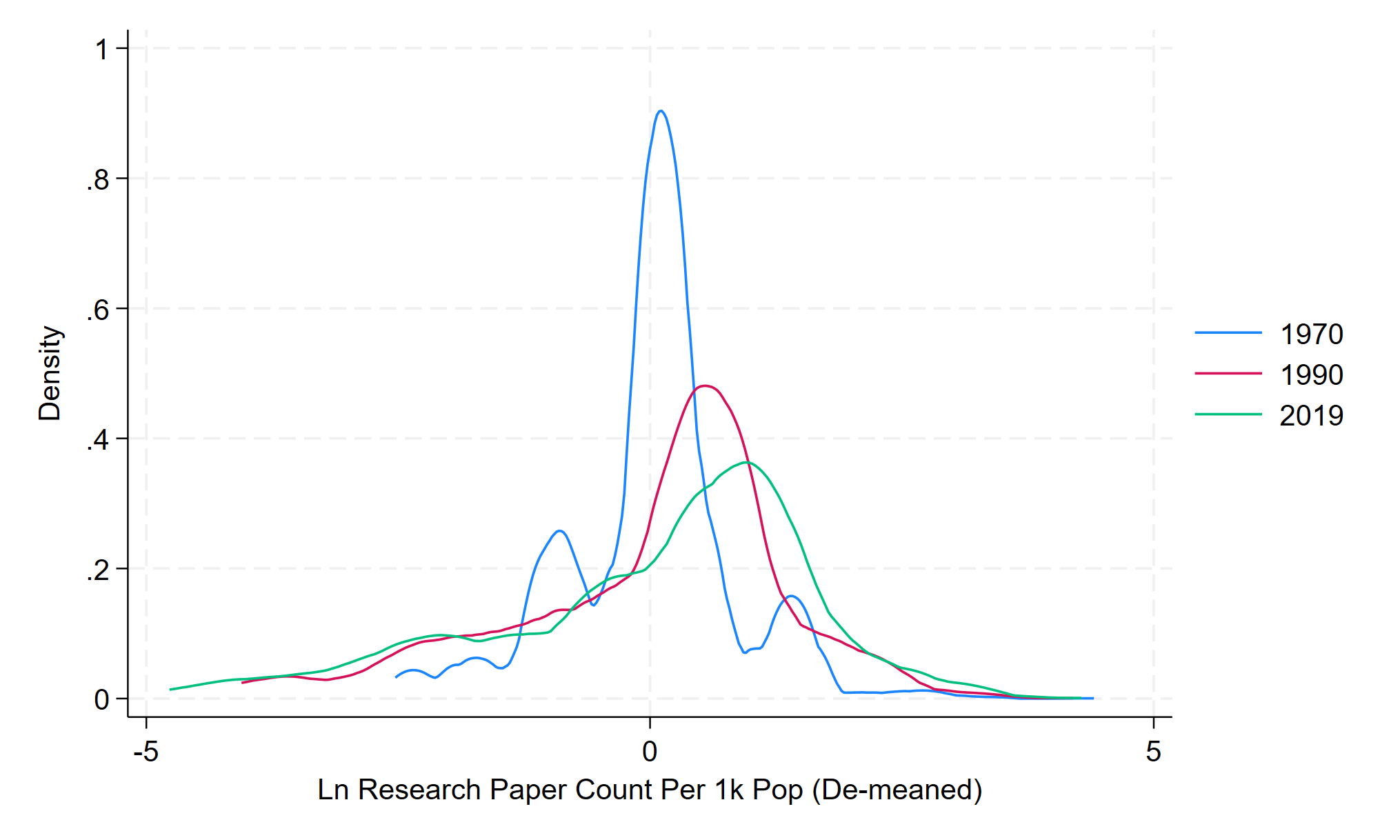}
      \caption{\centering Research Papers per 1,000 Residents}
             \label{fig:distribution_overall}
     \end{subfigure}
   \\
        \begin{subfigure}[b]{0.7\linewidth}
         \centering
         \includegraphics[width=\linewidth]{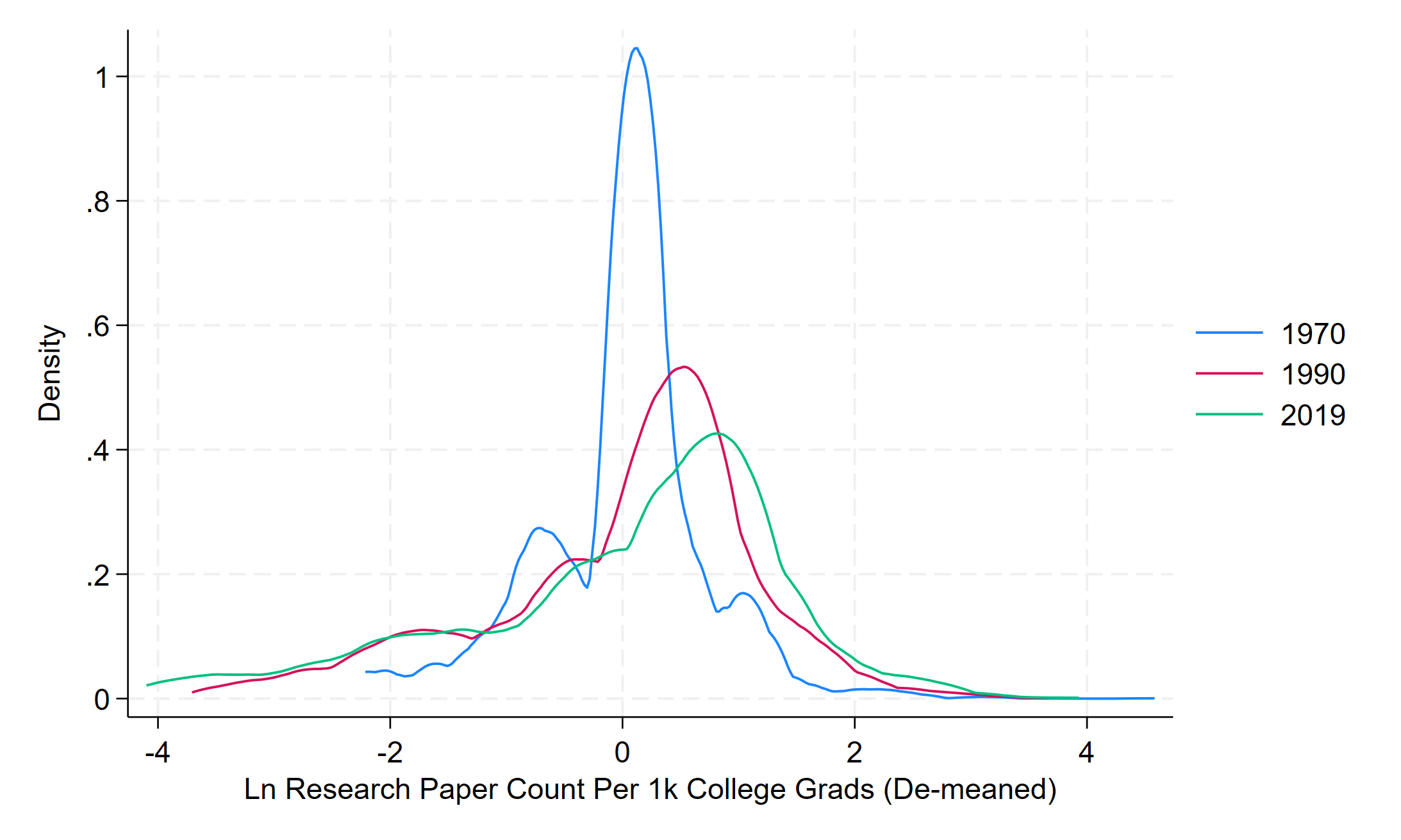}
     \caption{\centering Research Papers per 1,000 College Graduates}
        \label{fig:distribution_college}
     \end{subfigure}
\begin{minipage}{\textwidth}
\footnotesize{{\it Notes:} Panel (a) plots the population-weighted distribution of demeaned log research papers per 1,000 residents across MSAs. Panel (b) plots the corresponding distribution of research papers per 1,000 college graduates. Publication counts are aggregated over five-year periods for authors affiliated with institutions in each MSA.
}
\end{minipage}
        \label{fig:distribution}
\end{figure}

\begin{figure}[!h]
     \centering
        \caption{Decomposing Changes in the Distribution of Research Papers per 1,000 Residents} 
     \begin{subfigure}[b]{0.7\linewidth}
         \centering
         \includegraphics[width=\linewidth]{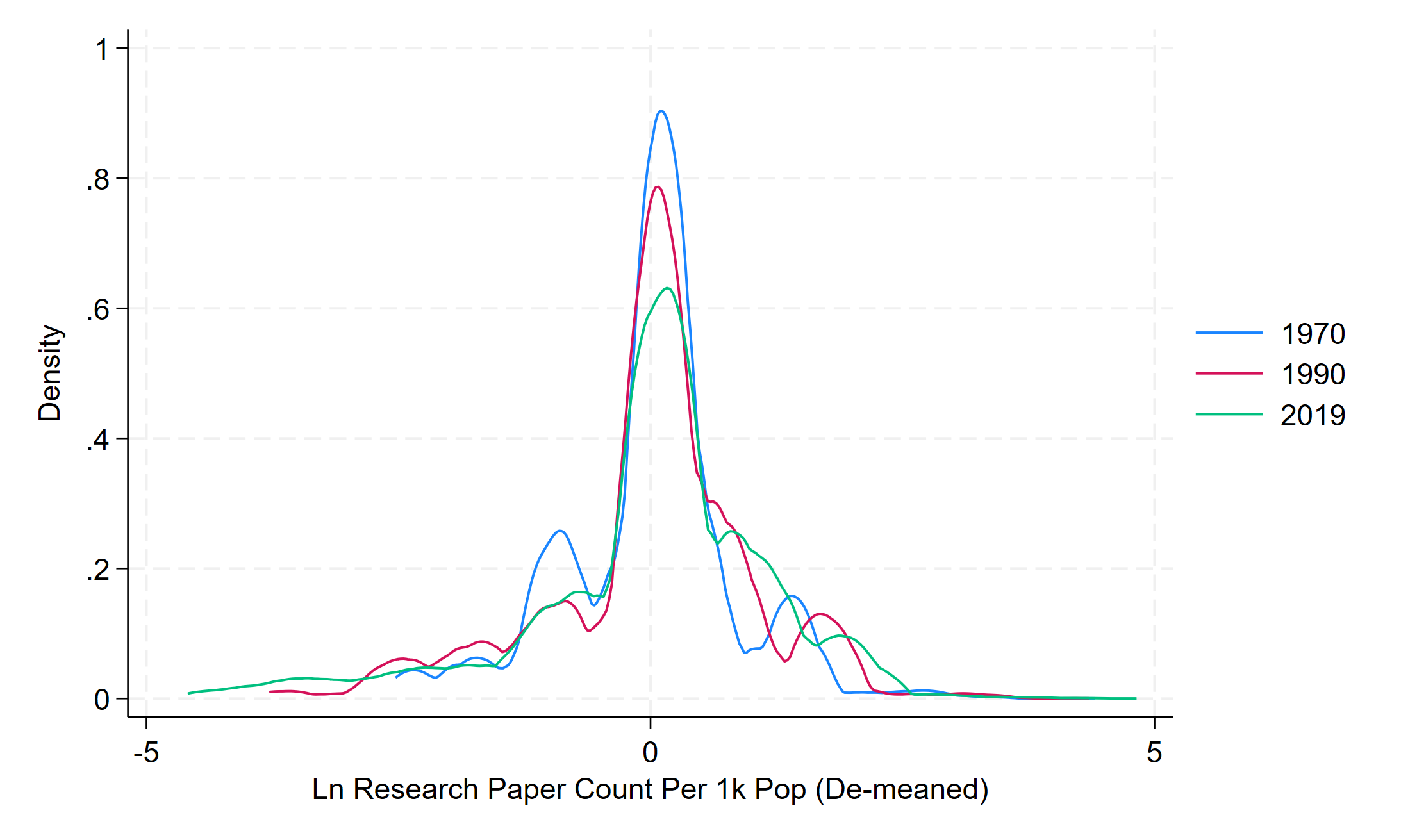}
      \caption{\centering Changing Research Activity (Population Fixed at 1970)}
        \label{fig:researcher_driven}
     \end{subfigure}
   \\
        \begin{subfigure}[b]{0.7\linewidth}
         \centering
    \includegraphics[width=\linewidth]{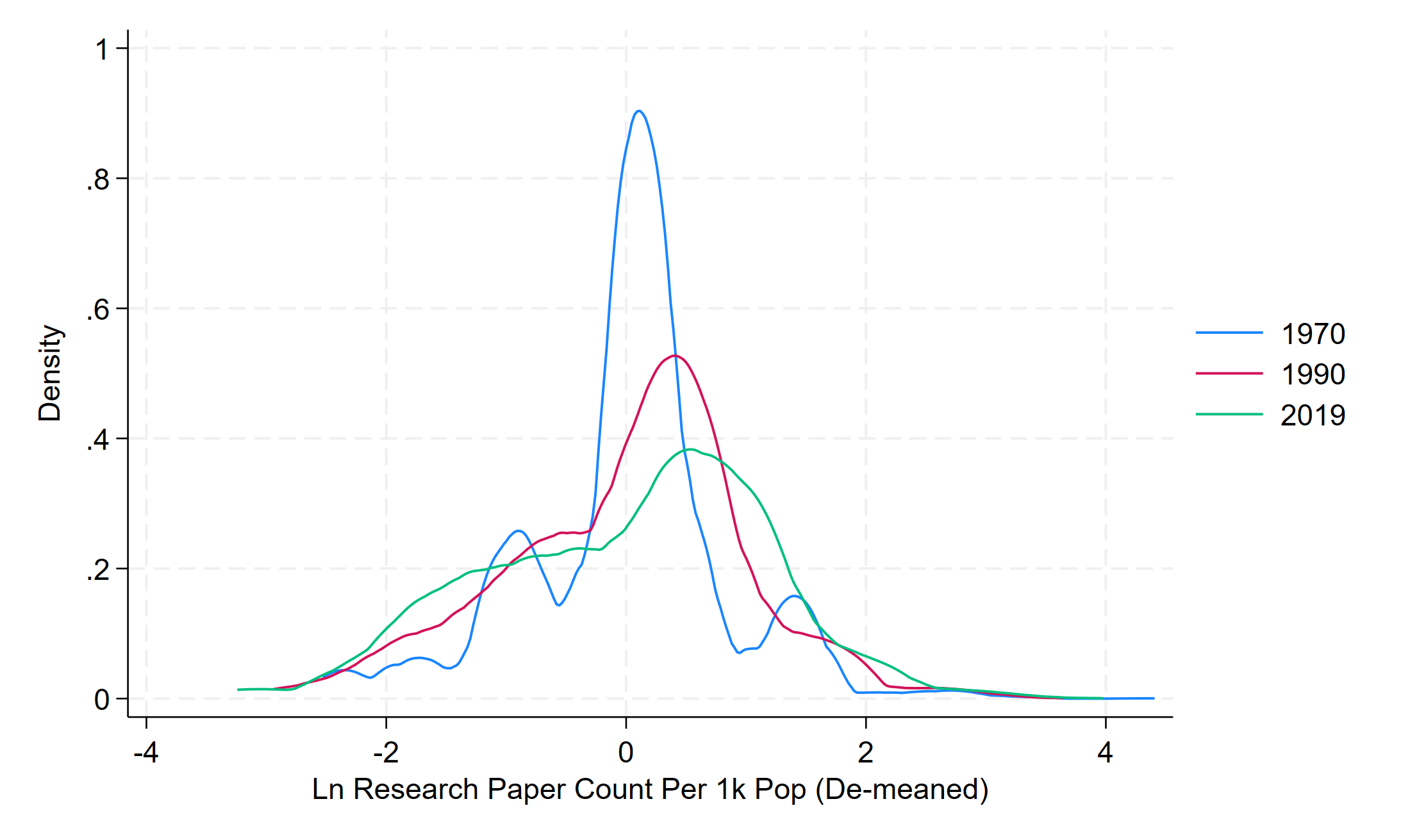}
     \caption{\centering Population Shifts (Publications Fixed at 1970)}
        \label{fig:population_driven}
     \end{subfigure}
\begin{minipage}{\textwidth}
\footnotesize{{\it Notes:} Panel (a) plots the population-weighted distribution of demeaned log research papers per 1,000 residents across MSAs, holding each MSA’s population fixed at its 1970 level while allowing publication counts to vary. Panel (b) holds publication counts fixed at their 1970 levels while allowing population to vary. The panels isolate the contributions of changing research activity and population shifts, respectively.}
\end{minipage}
        \label{fig:population_researchers}
\end{figure}

\clearpage

\begin{figure}[!h]
     \centering
        \caption{Distributions of Research Papers per Worker and Business Establishment} 
     \begin{subfigure}[b]{0.48\linewidth}
         \centering
         \includegraphics[width=\linewidth]{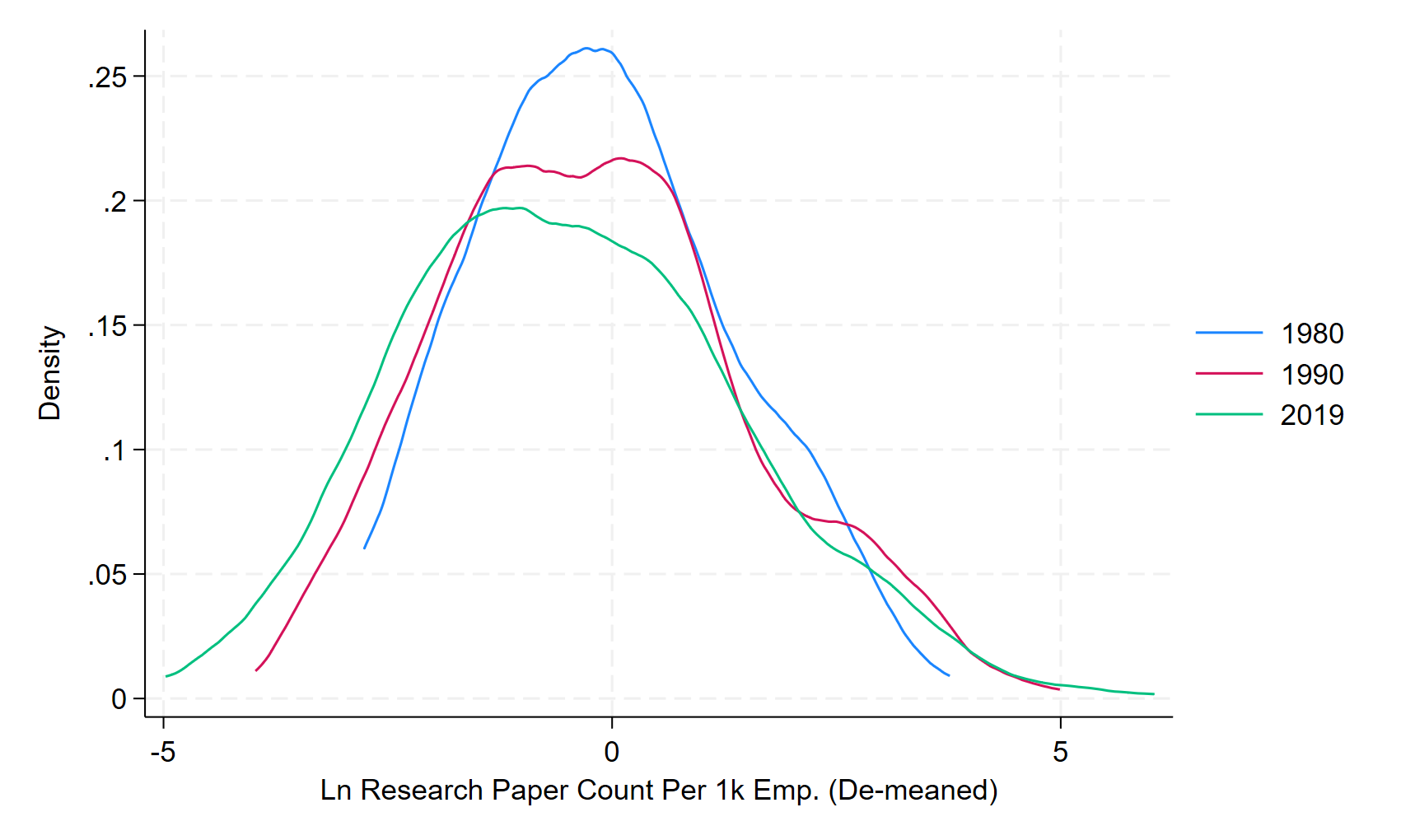}
      \caption{\centering Research Papers Per 1,000 Workers}
             \label{fig:dist_emp}
     \end{subfigure}
\hfill
        \begin{subfigure}[b]{0.48\linewidth}
         \centering
         \includegraphics[width=\linewidth]{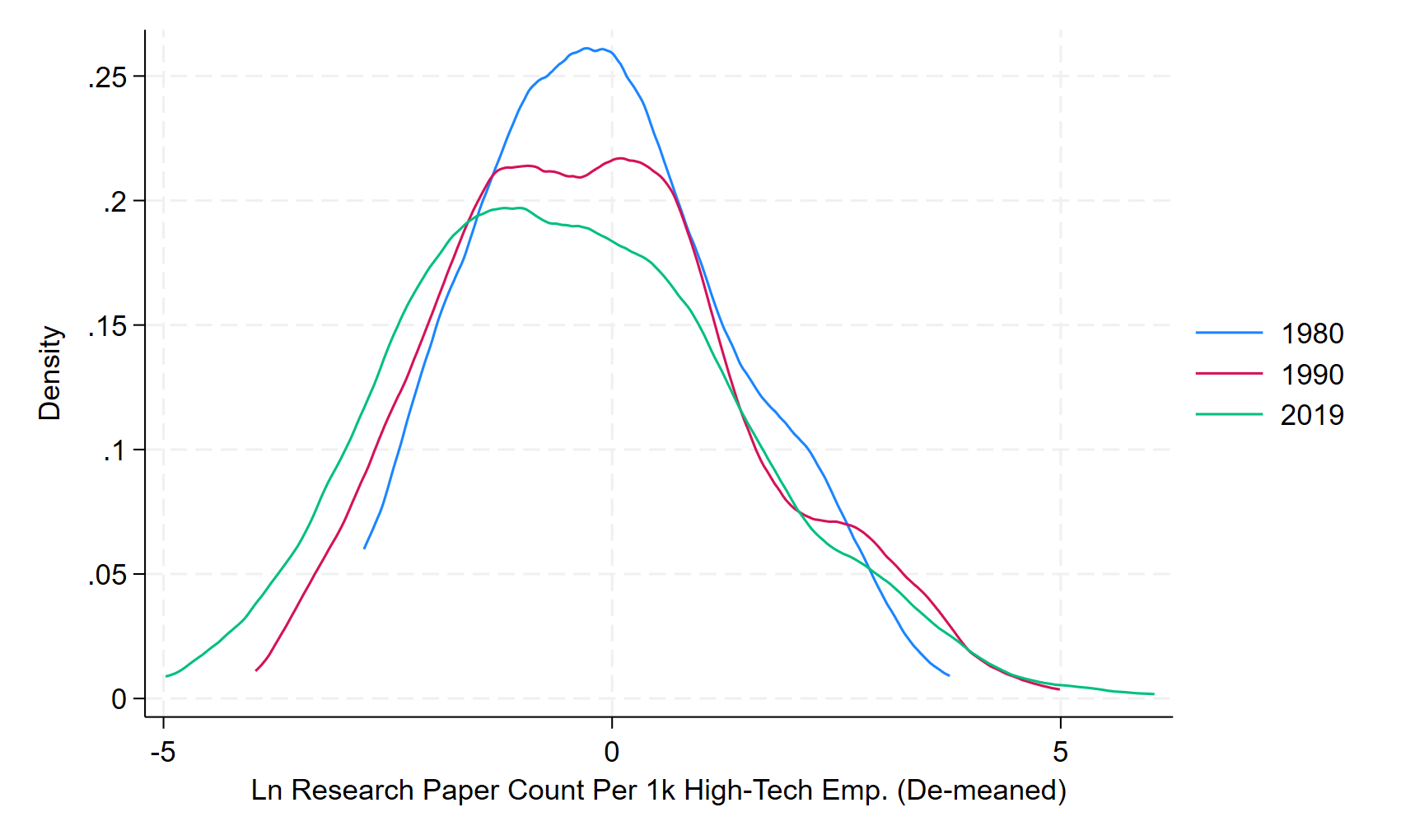}
     \caption{\centering Research Papers Per 1,000 High-Tech Workers}
        \label{fig:dist_high_tech_emp}
     \end{subfigure}

\begin{subfigure}[b]{0.48\linewidth}
         \centering
         \includegraphics[width=\linewidth]{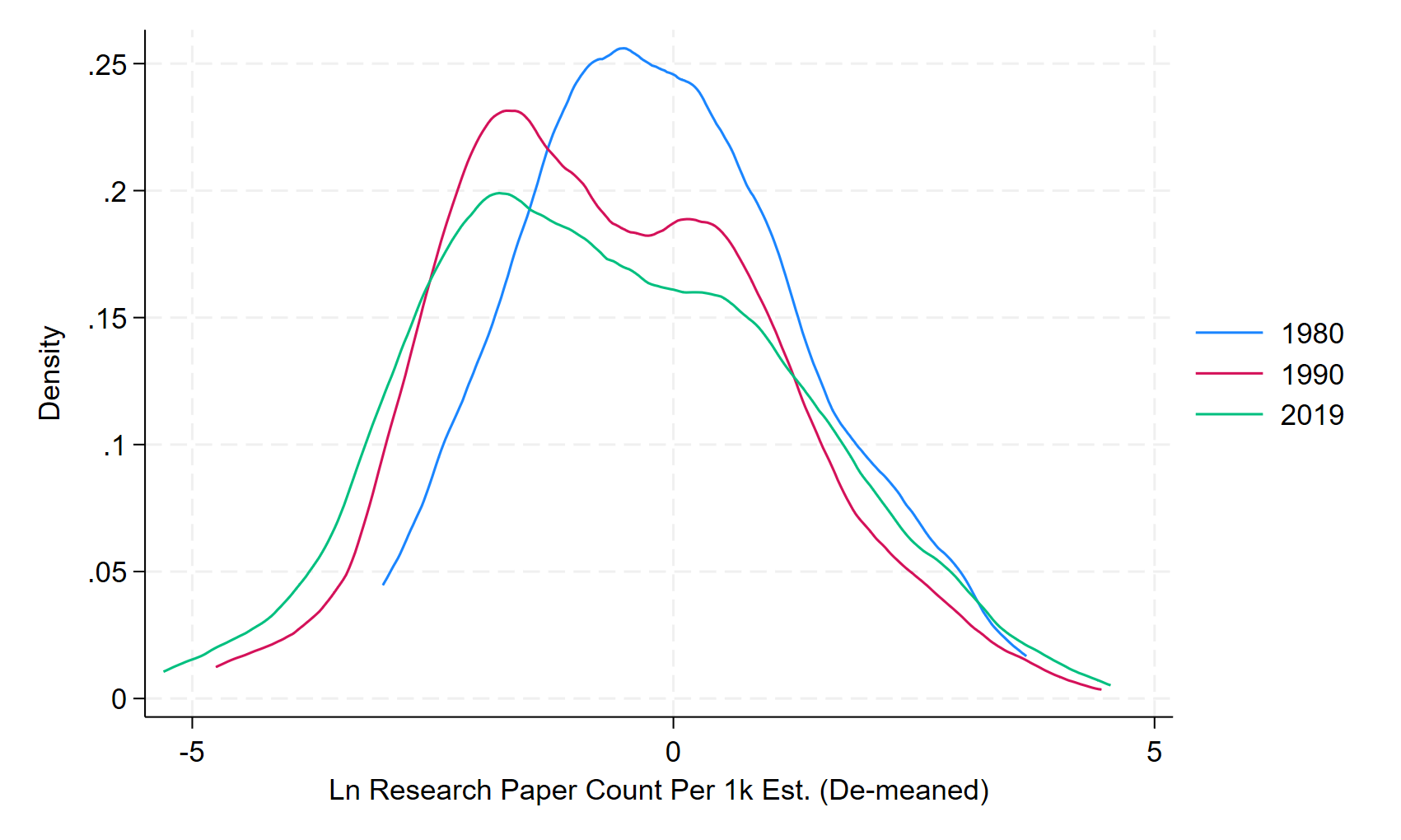}
      \caption{\centering Research Papers Per 1,000 Establishments}
             \label{fig:dist_est}
     \end{subfigure}
\hfill
        \begin{subfigure}[b]{0.48\linewidth}
         \centering
         \includegraphics[width=\linewidth]{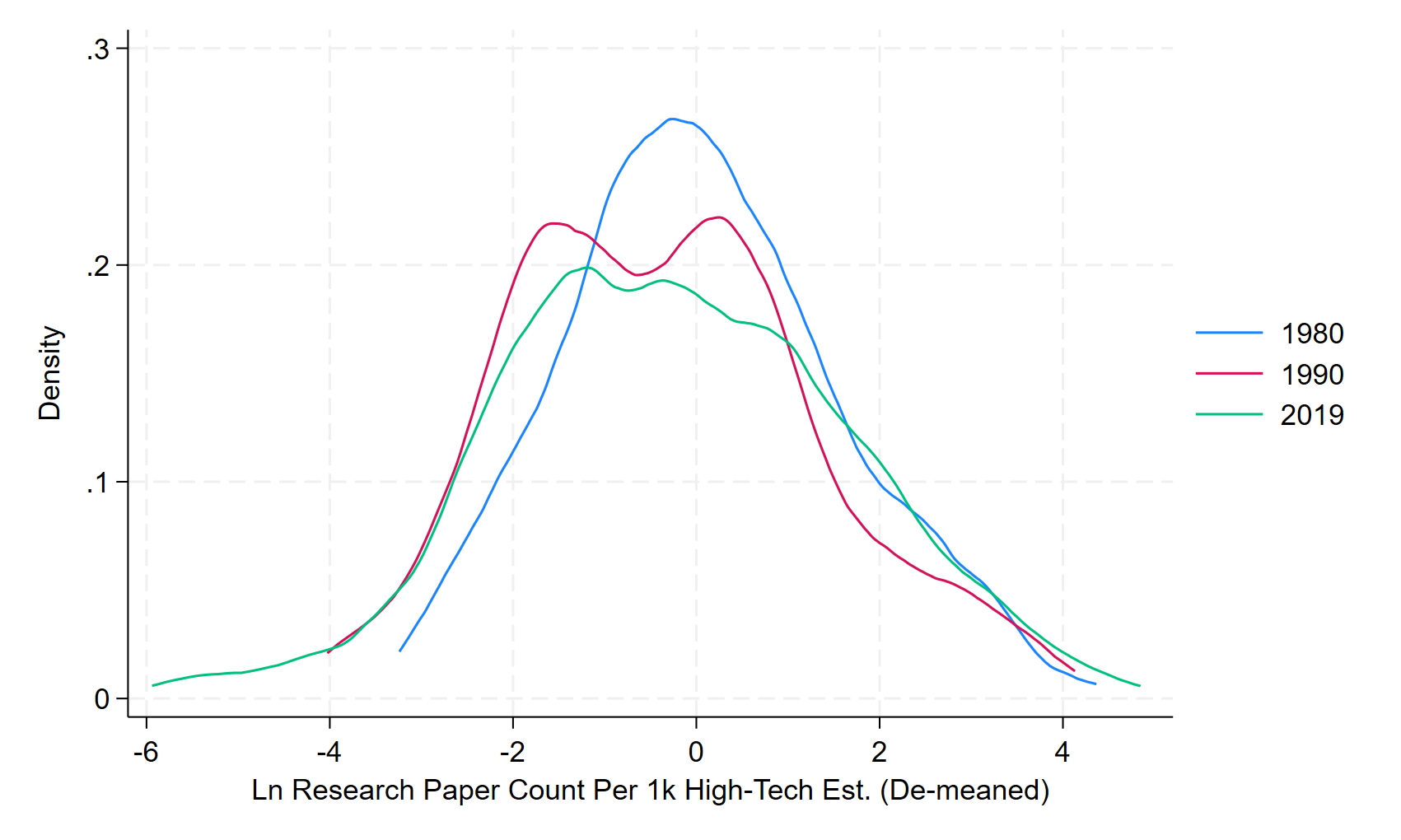}
     \caption{\centering Research Papers Per 1,000 High-Tech Establishments}
        \label{fig:dist_high_tech_est}
     \end{subfigure}
\begin{minipage}{\textwidth}
\footnotesize{{\it Notes:} The figure plots distributions of access to research papers across MSAs in 1980, 1990, and 2019. Panels (a) and (b) show demeaned log research papers per 1,000 workers and per 1,000 high-tech workers, respectively, weighted by the corresponding employment counts. Panels (c) and (d) show demeaned log research papers per 1,000 establishments and per 1,000 high-tech establishments, respectively, weighted by the corresponding establishment counts. Publication counts are aggregated over five-year periods. Employment and establishment data come from the QCEW and County Business Patterns, respectively. High-tech sectors are defined using 2017 three-digit NAICS codes, including oil and gas extraction (211), utilities (221), pipeline transportation (486), petroleum and coal products (324), chemicals (325), machinery (333), computer and electronic products (334), electrical equipment and components (335), transportation equipment (336), software (511), telecommunications (517), data processing and hosting (518), other information services (519), wholesale electronic markets and agents and brokers (425), professional, scientific, and technical services (541), and management of companies and enterprises (551).
}
\end{minipage}
        \label{fig:dist_emp_est}
\end{figure}

\clearpage

\begin{figure}[!h]
     \centering
        \caption{Distributions of Access to University and Non-University Research} 
     \begin{subfigure}[b]{0.48\linewidth}
         \centering
         \includegraphics[width=\linewidth]{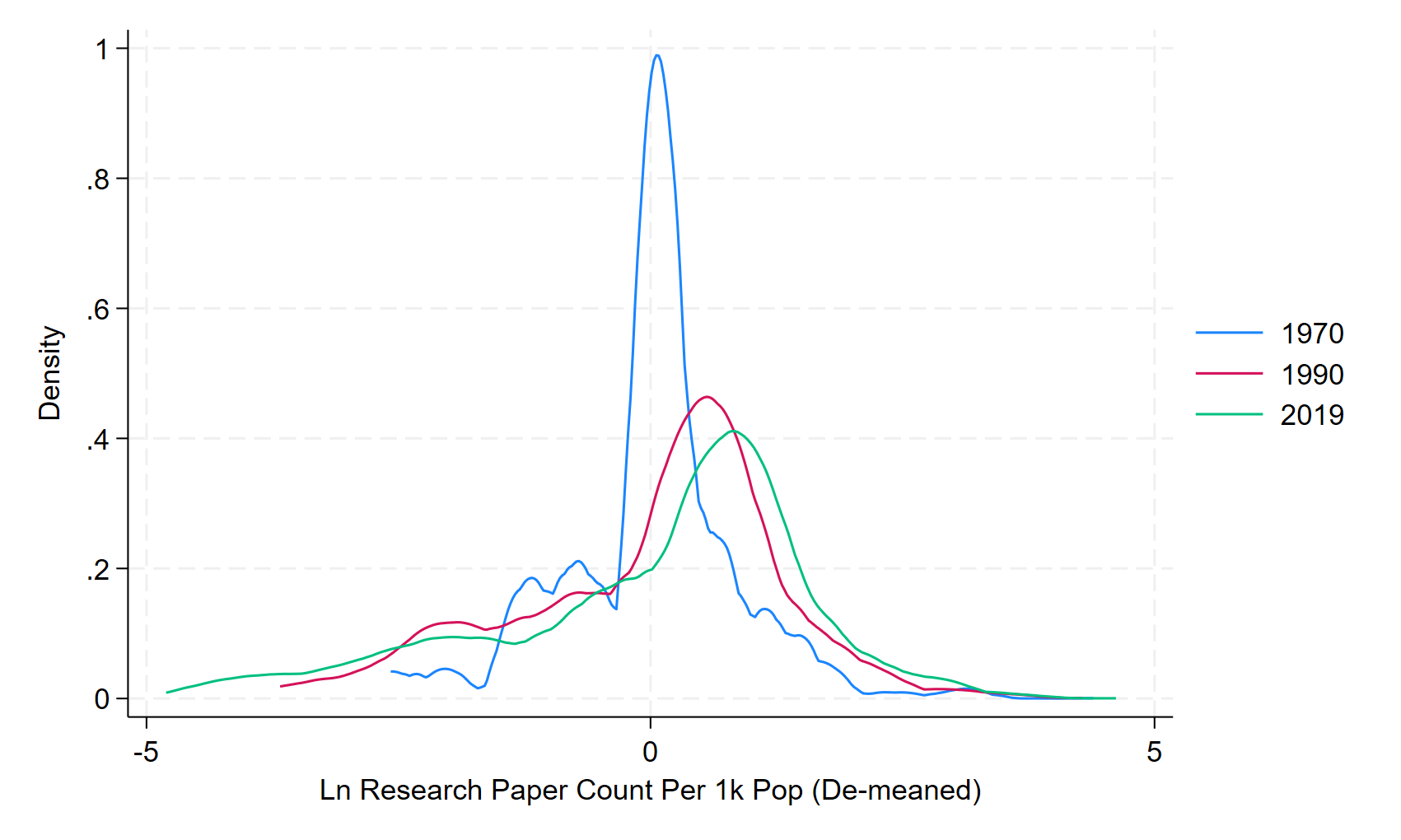}
      \caption{\centering University Research per 1,000 Residents}
             \label{fig:dist_univ}
     \end{subfigure}
\hfill
        \begin{subfigure}[b]{0.48\linewidth}
         \centering
         \includegraphics[width=\linewidth]{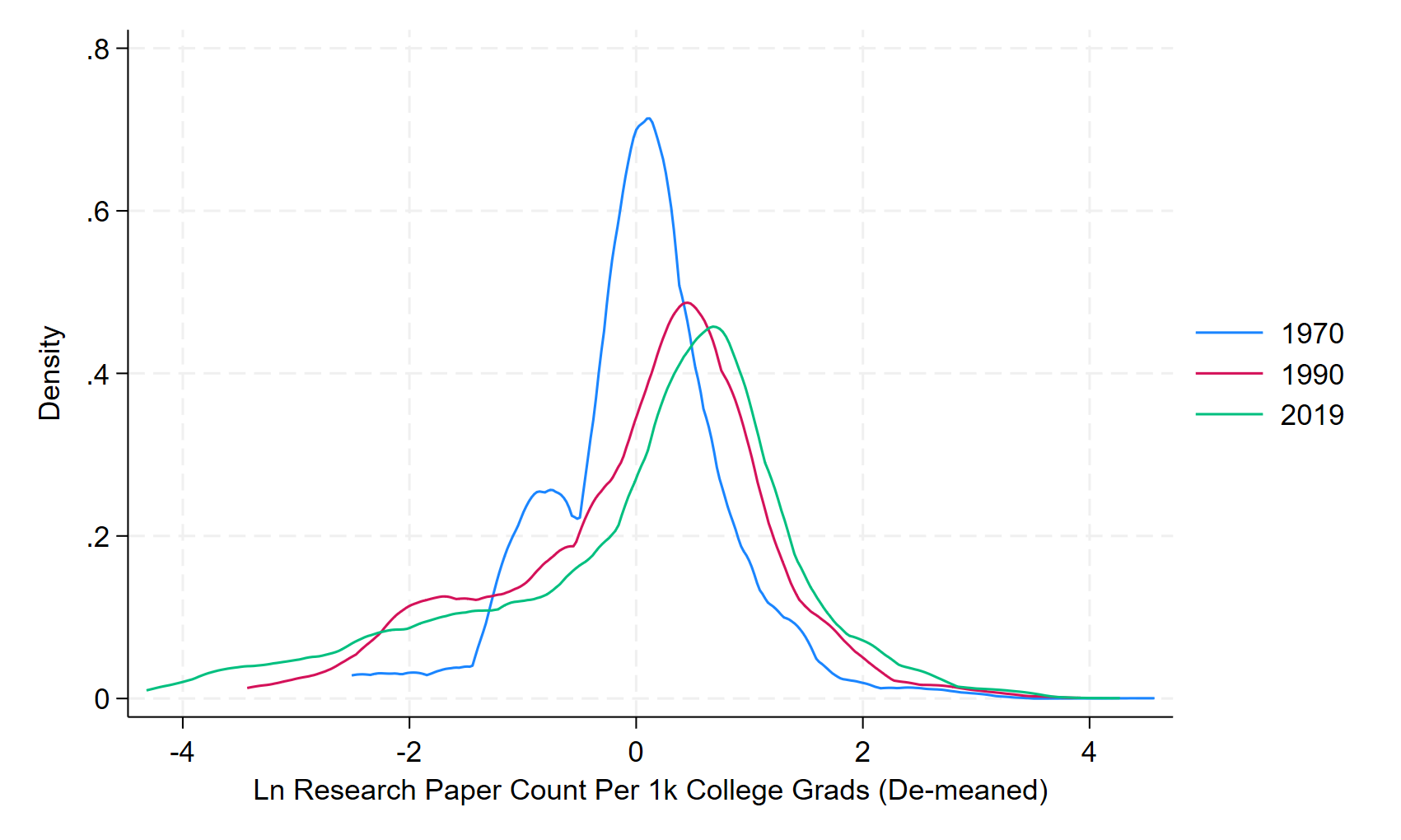}
     \caption{\centering University Research per 1,000 College Graduates}
        \label{fig:dist_univ_college}
     \end{subfigure}

\begin{subfigure}[b]{0.48\linewidth}
         \centering
         \includegraphics[width=\linewidth]{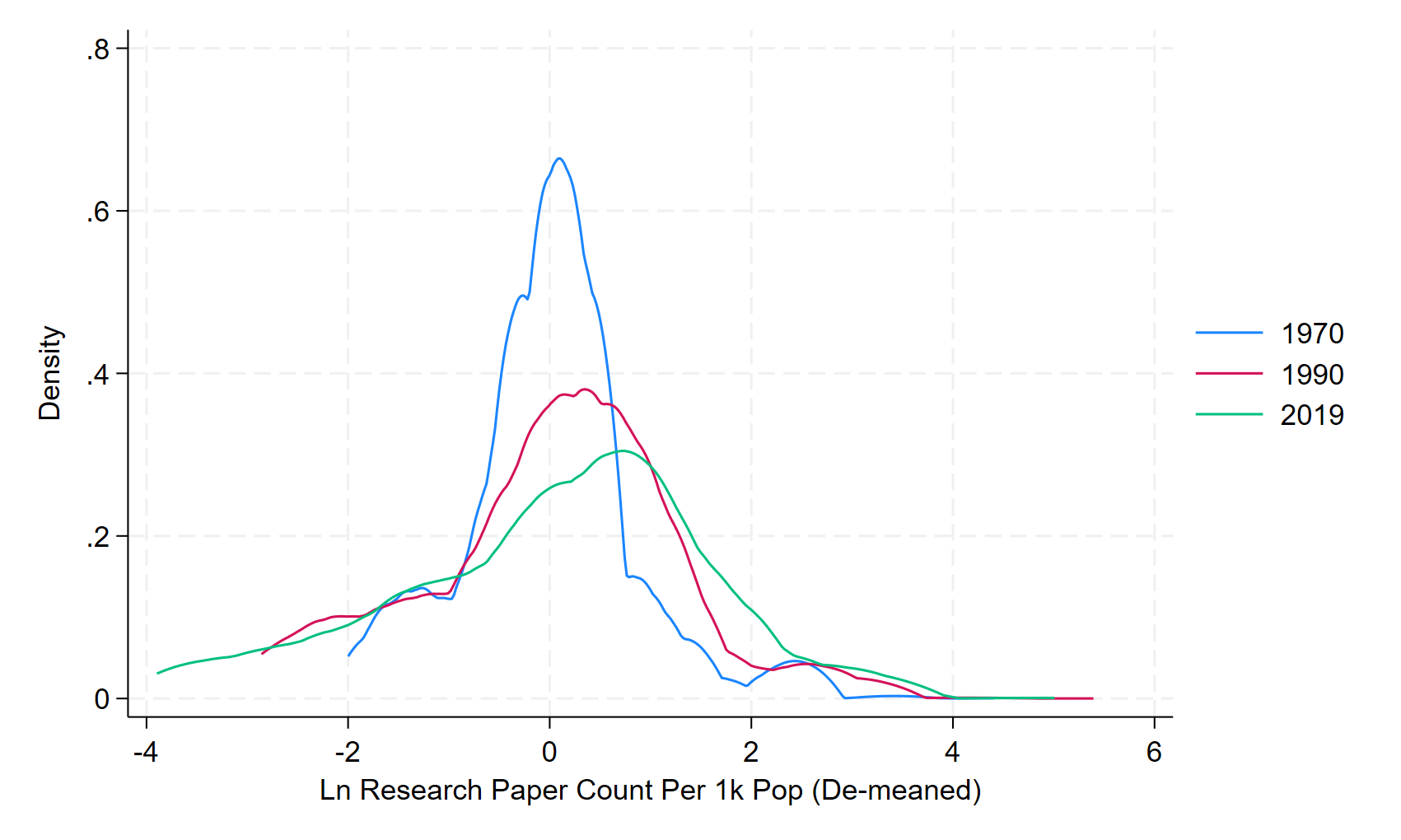}
      \caption{\centering Non-University Research per 1,000 Residents}
             \label{fig:dist_nonuniv}
     \end{subfigure}
\hfill
        \begin{subfigure}[b]{0.48\linewidth}
         \centering
         \includegraphics[width=\linewidth]{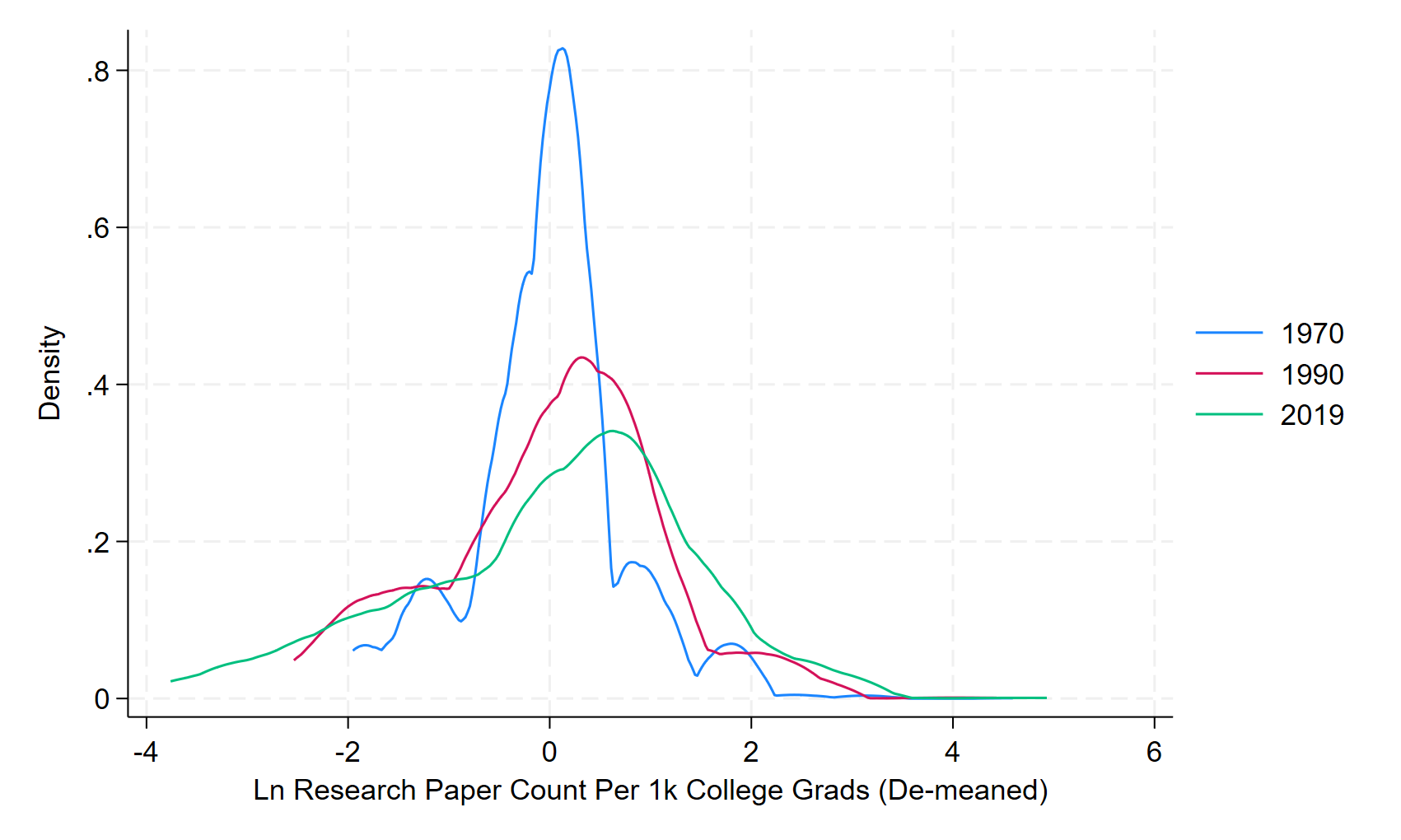}
     \caption{\centering Non-University Research per 1,000 College Graduates}
        \label{fig:dist_nonuniv_college}
     \end{subfigure}
\begin{minipage}{\textwidth}
\footnotesize{{\it Notes:} The figure plots distributions of access to research papers across MSAs in 1970, 1990, and 2019. Panels (a) and (b) include papers produced by university-affiliated authors, while Panels (c) and (d) include papers produced by authors affiliated with non-university institutions. Panels (a) and (c) show demeaned log papers per 1,000 residents, weighted by MSA population. Panels (b) and (d) show demeaned log papers per 1,000 college graduates, weighted by the college-educated population. Publication counts are aggregated over five-year periods.}
\end{minipage}
        \label{fig:dist_univ_nonuniv}
\end{figure}

\clearpage

\begin{figure}[!h]
     \centering
        \caption{Research Papers, Researchers, and Faculty per 1,000 Residents} 
     \begin{subfigure}[b]{0.47\linewidth}
         \centering
         \includegraphics[width=\linewidth]{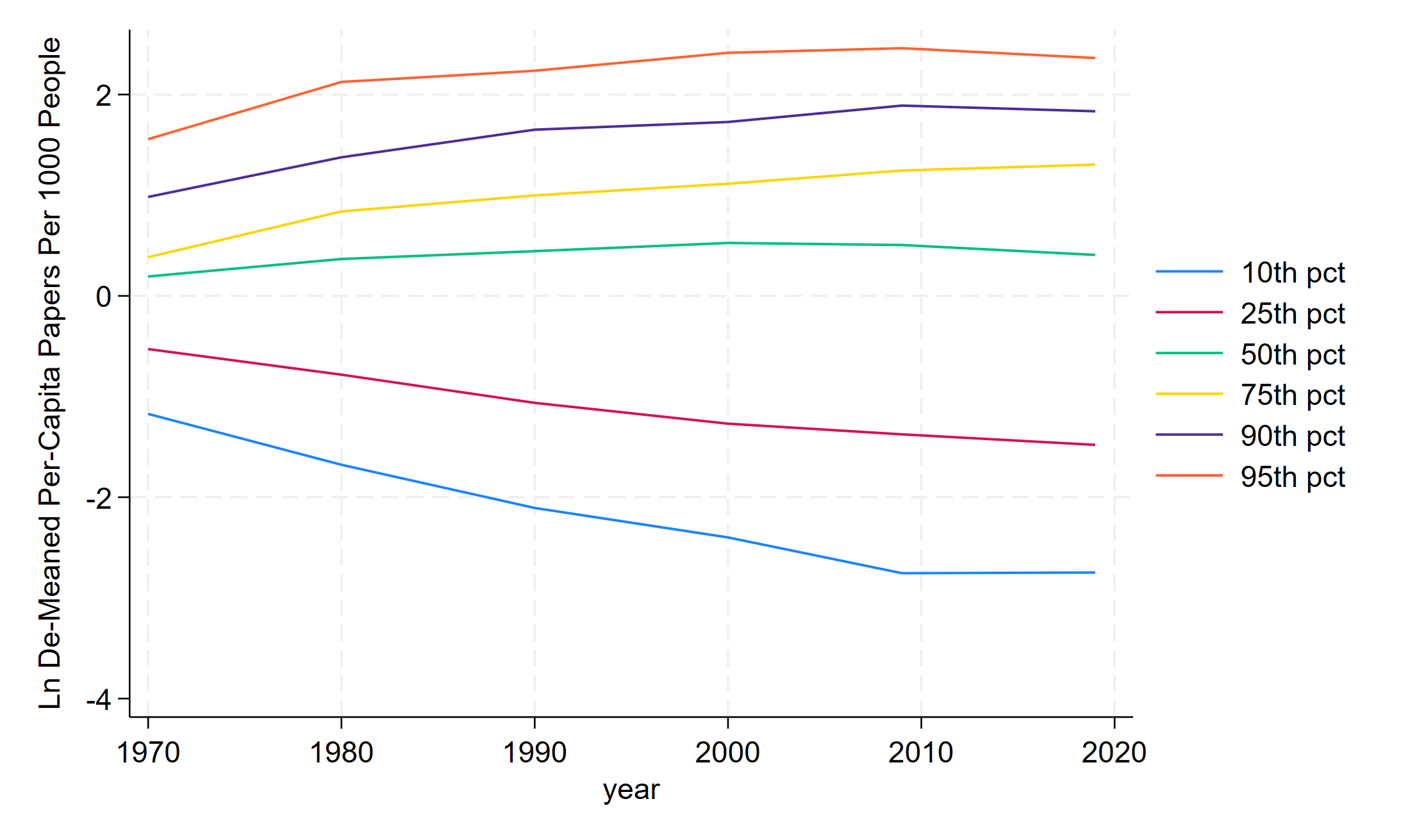}
      \caption{\centering Research Papers}
             \label{fig:per_capita_research}
     \end{subfigure}
     \hfil
    \begin{subfigure}[b]{0.47\linewidth}
         \centering
         \includegraphics[width=\linewidth]{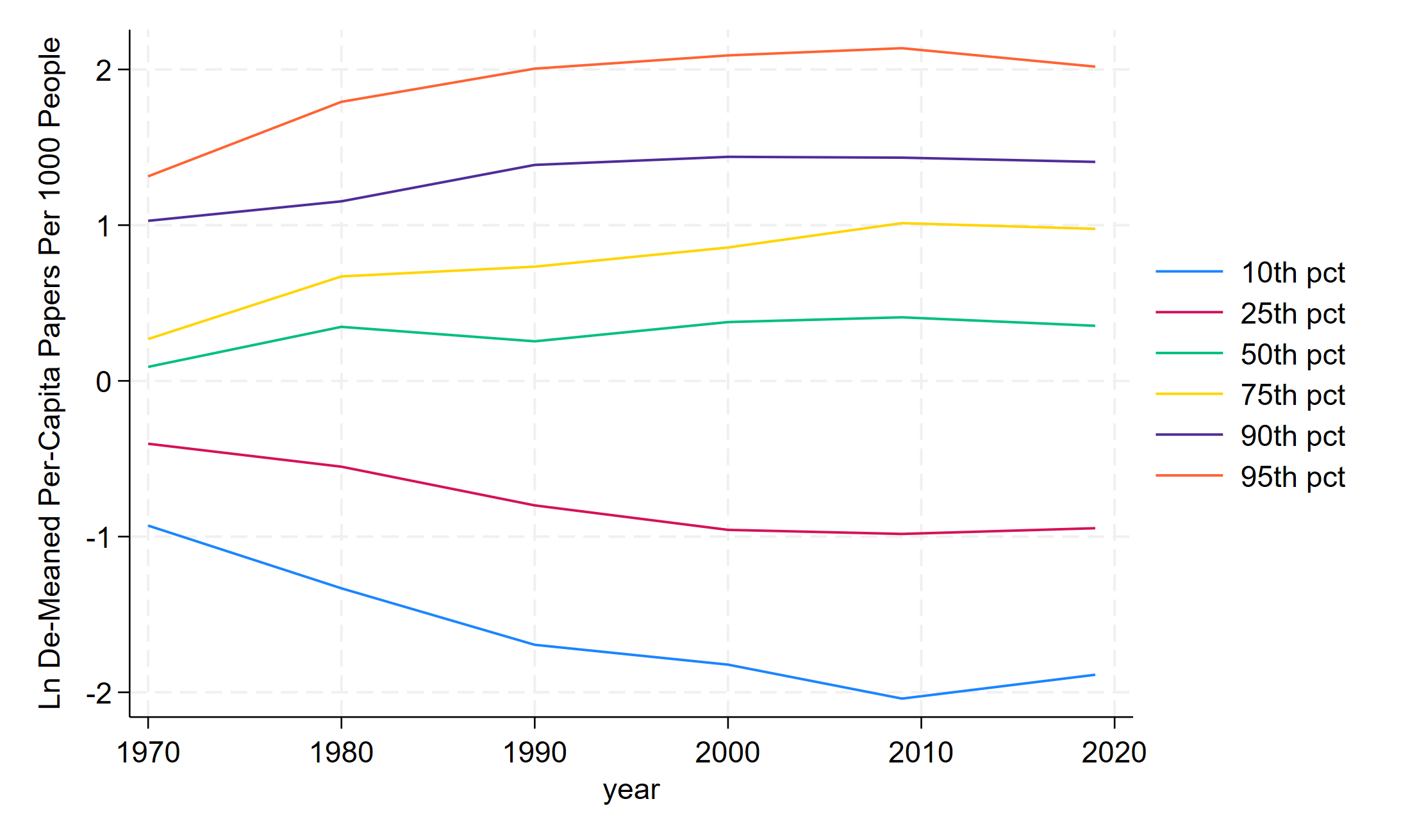}
      \caption{\centering Researchers}
\label{fig:per_capita_researcher}
     \end{subfigure}
                 \\ 
        \begin{subfigure}[b]{0.47\linewidth}
         \centering
         \includegraphics[width=\linewidth]{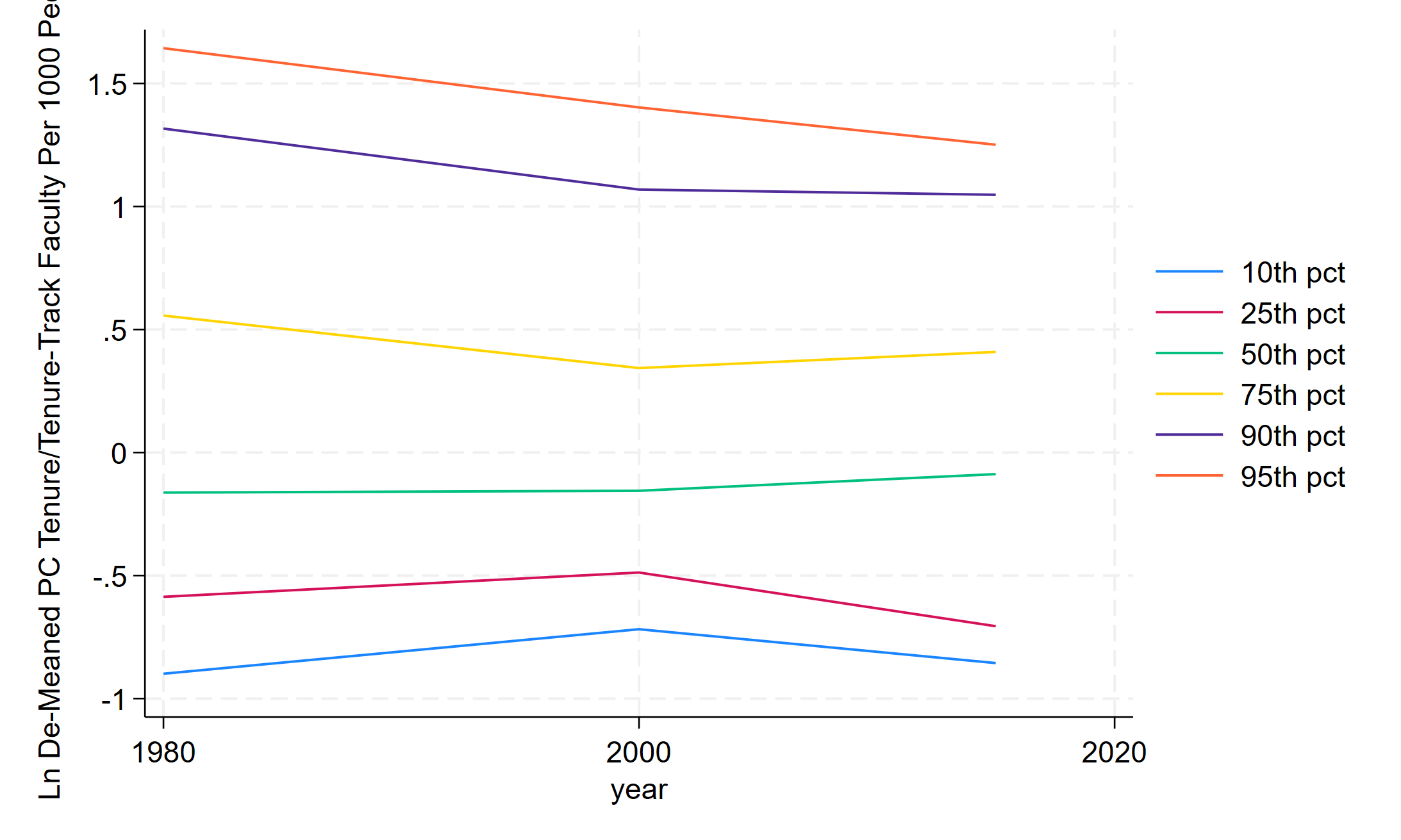}
         \label{fig:per_capita_faculty}
     \caption{\centering Tenured and Tenure-Track Faculty}
     \end{subfigure}
\begin{minipage}{\textwidth}
\footnotesize{{\it Notes:} The figure plots selected percentiles of the MSA-level distributions of demeaned log research papers, active researchers, and tenured and tenure-track faculty per 1,000 residents over time. Research papers and researcher counts come from the Microsoft Academic Graph, while faculty counts come from IPEDS.} 
\end{minipage}
        \label{fig:research_vs_faculty}
\end{figure}

\clearpage

\begin{figure}[!h]
     \centering
    \caption{Event Study: Moves to Higher vs. Lower-Output Institutions} 
     \begin{subfigure}[b]{0.4\linewidth}
         \centering
         \includegraphics[width=\linewidth]{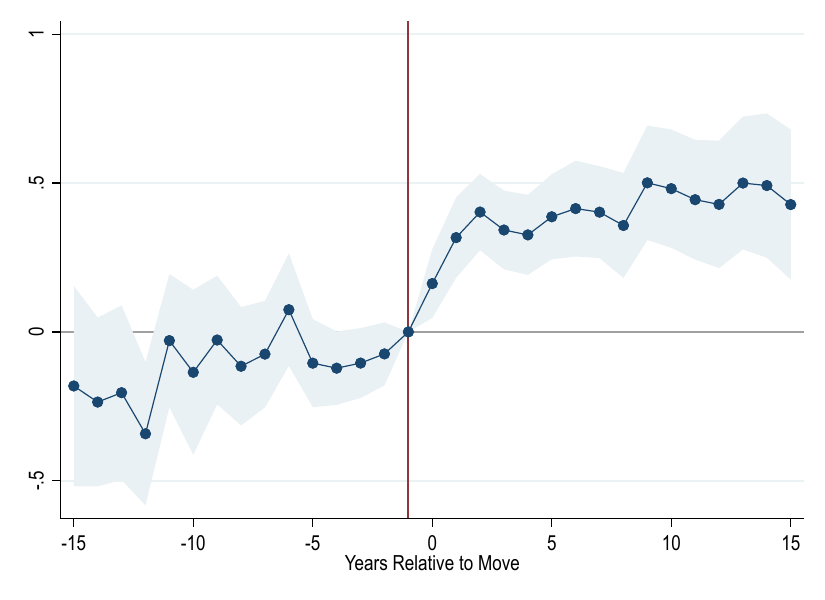}
         \caption{\centering Publications (Moves Up)}
         \label{fig:event_location_publication_up}
     \end{subfigure}
     \qquad
    \begin{subfigure}[b]{0.4\linewidth}
         \centering
         \includegraphics[width=\linewidth]{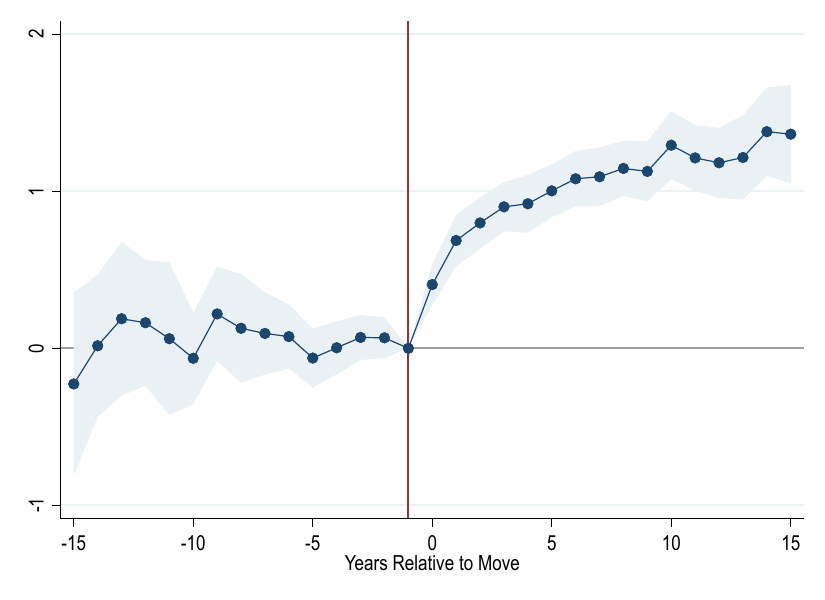}
         \caption{\centering Publications (Moves Down)}
         \label{fig:event_location_publication_down}
     \end{subfigure}
     \\
     \vspace{0.25cm}    
     \begin{subfigure}[b]{0.4\linewidth}
         \centering
         \includegraphics[width=\linewidth]{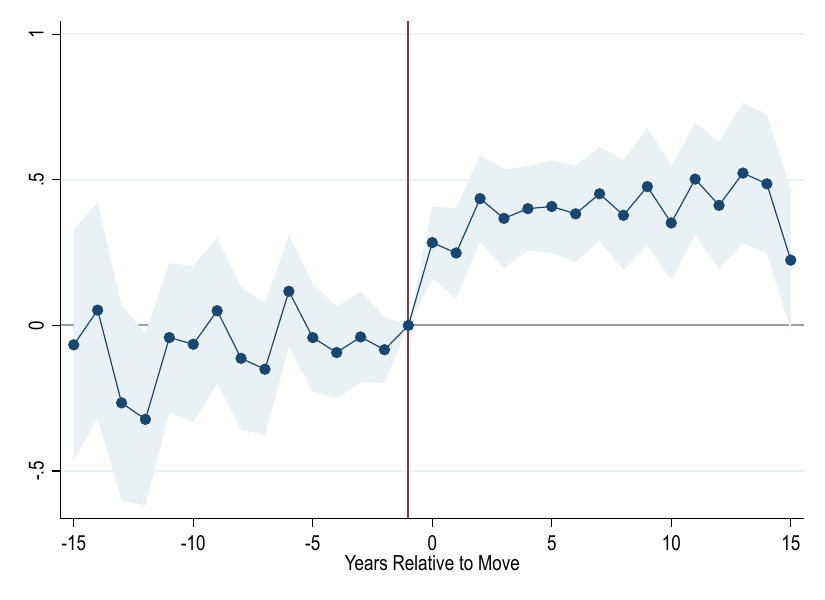}
         \caption{\centering Citations (Moves Up) }
         \label{fig:event_location_citation_up}
     \end{subfigure}
     \qquad
     \begin{subfigure}[b]{0.4\linewidth}
         \centering
         \includegraphics[width=\linewidth]{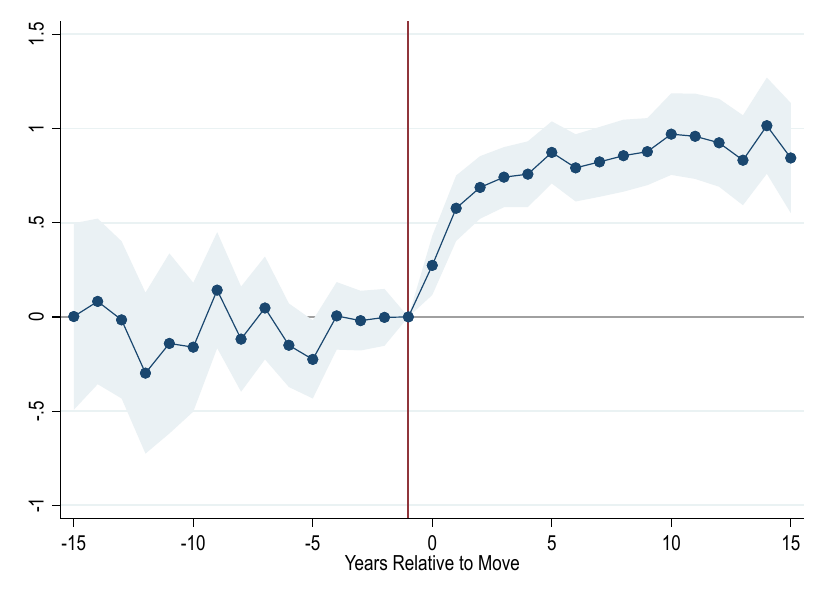}
         \caption{\centering Citations (Moves Down)}
         \label{fig:event_location_citation_down}
     \end{subfigure} 
          \\
     \vspace{0.25cm}    
     \begin{subfigure}[b]{0.4\linewidth}
         \centering
         \includegraphics[width=\linewidth]{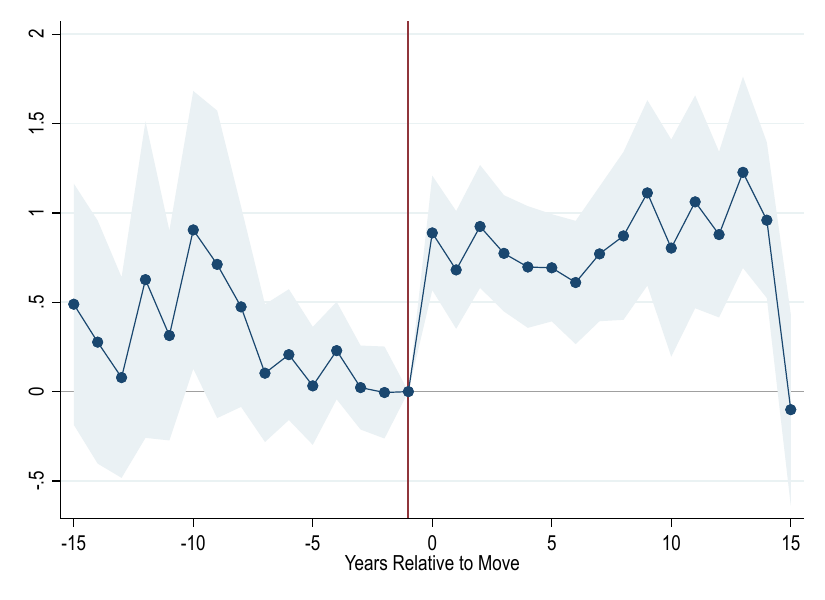}
         \caption{\centering Patent Citations (Moves Up) }
         \label{fig:event_location_patent_up}
     \end{subfigure}
     \qquad
     \begin{subfigure}[b]{0.4\linewidth}
         \centering
         \includegraphics[width=\linewidth]{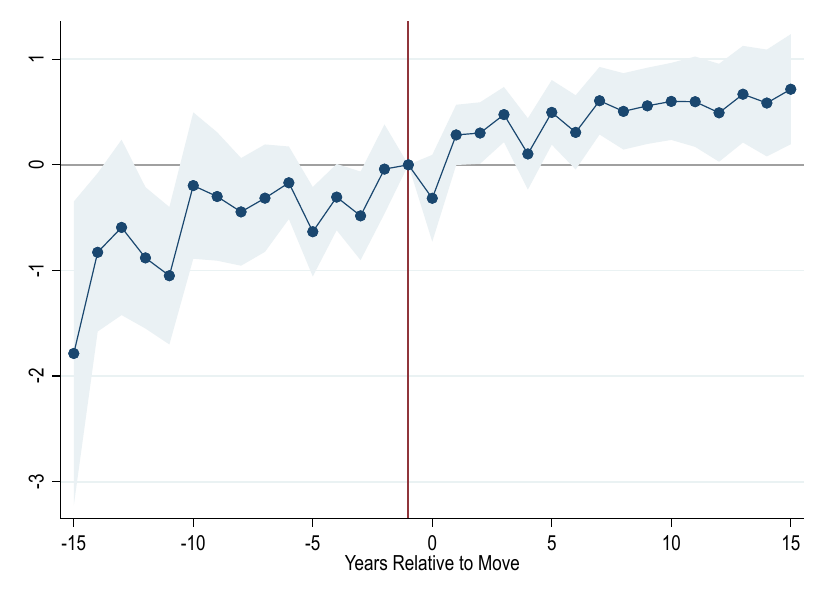}
         \caption{\centering Patent Citations (Moves Down)}
         \label{fig:event_location_patent_down}
     \end{subfigure}
    \vspace{0.25cm}
    \label{fig:event_location_direction}
    \vspace{0.25cm}
\begin{minipage}{\textwidth}
\footnotesize{{\it Notes:} Each panel plots the estimated coefficients $\delta_s$ for $-15\leq s \leq 15$ ($s \neq -1$) from Equation \ref{eq:event}, using the sample of researchers who changed institutions once (analogous to Figure \ref{fig:event_location}). The dependent variable is the inverse hyperbolic sine of publications (Panels (a) and (b)), citations (Panels (c) and (d)), and patent citations (Panels (e) and (f)). Panels (a), (c), and (e) use movers who move up (i.e., $\Delta_i>0$). Panels (b), (d), and (f) use movers who move down (i.e., $\Delta_i<0$).}
\end{minipage}
\end{figure}

\clearpage

\begin{figure}[!h]
    \captionsetup{justification=centering}
    \caption{Additive Decomposition of Research Outcomes By Field: \\ Top vs. Bottom Productivity Deciles} 
    \centering
    \begin{subfigure}[b]{0.47\linewidth}
         \centering
         \includegraphics[width=\linewidth]{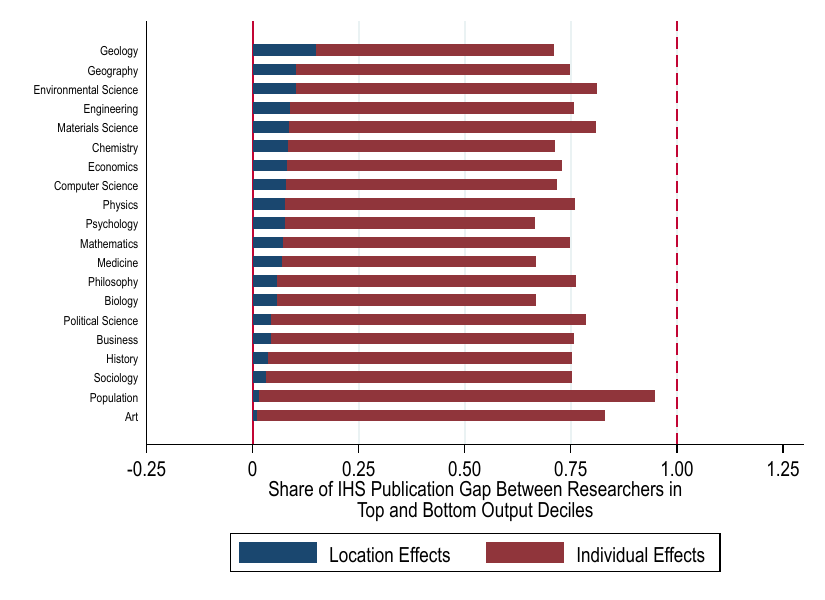}
         \caption{\centering Publications}
         \label{fig:decomp_field_output_publications}
     \end{subfigure}
     \qquad 
     \begin{subfigure}[b]{0.47\linewidth}
         \centering
         \includegraphics[width=\linewidth]{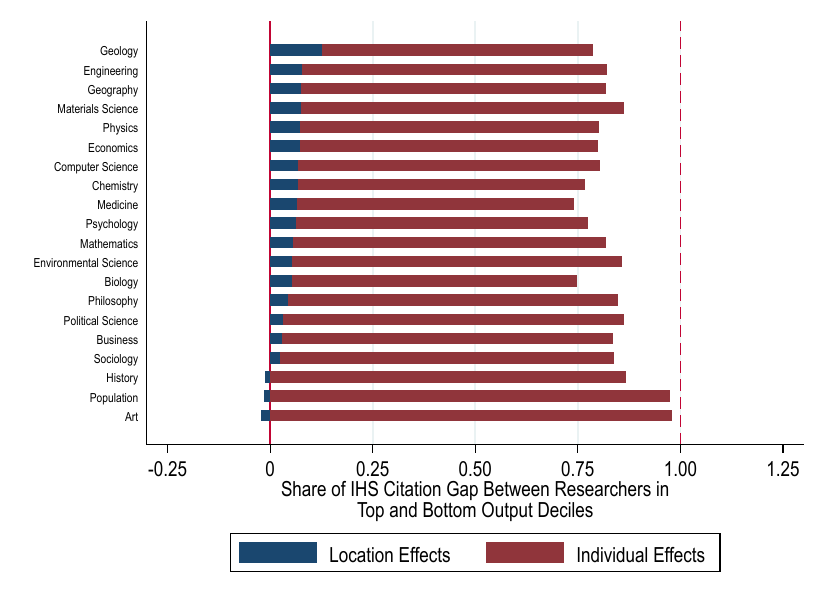}
         \caption{\centering Citations }
         \label{fig:decomp_field_output_citations}
     \end{subfigure}
     \\
     \vspace{0.25cm}    
     \begin{subfigure}[b]{0.47\linewidth}
         \centering
         \includegraphics[width=\linewidth]{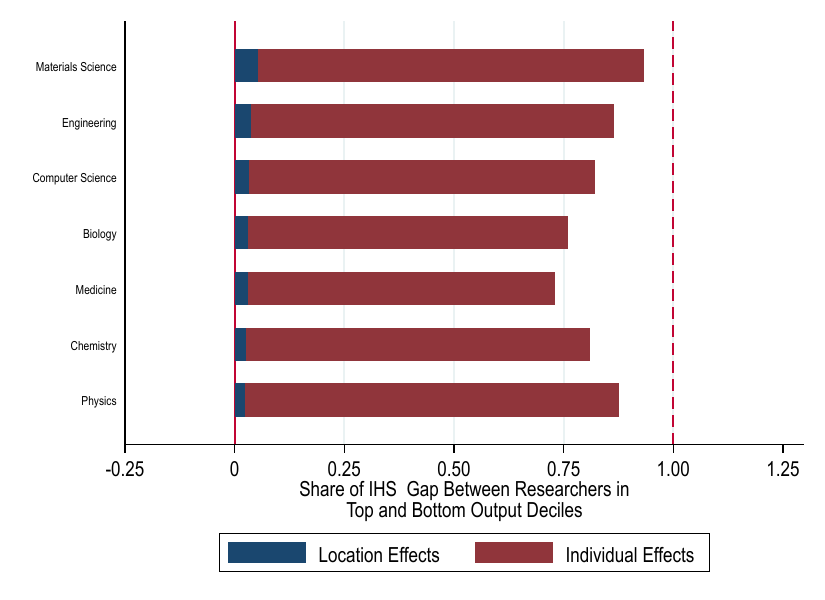}
         \caption{\centering Patent Citations}
         \label{fig:decomp_field_output_patents}
     \end{subfigure}
    \vspace{0.25cm}
    \label{fig:decomp_field_output}
    \vspace{0.25cm}
\begin{minipage}{\textwidth}
\footnotesize{{\it Notes:} This figure reports field-specific additive decompositions of research outcomes across researcher productivity groups during 2010--2015. Panels (a), (b), and (c) report results for IHS publications, citations, and patent citations, respectively. Within each field, researchers are ranked by their average value of the corresponding outcome, and the figure compares researchers in the top and bottom deciles. The researcher component includes author fixed effects and field-specific academic-age effects, while the location component is predicted using lagged observable characteristics of the institution--field and external MSA--field research environments. Component shares are calculated as the difference in each component divided by the overall difference in the corresponding outcome. Subfigure (c) includes only patent-relevant fields, i.e., fields in which more than 30\% of papers receive at least one patent citation.}
\end{minipage}
\end{figure}

\clearpage

\begin{figure}[!h]
    
    \captionsetup{justification=centering}
    \caption{Additive Decomposition of Research Outcomes By Field: \\ Top vs. Bottom Institution Size Deciles} 
    \centering
    \begin{subfigure}[b]{0.47\linewidth}
         \centering
         \includegraphics[width=\linewidth]{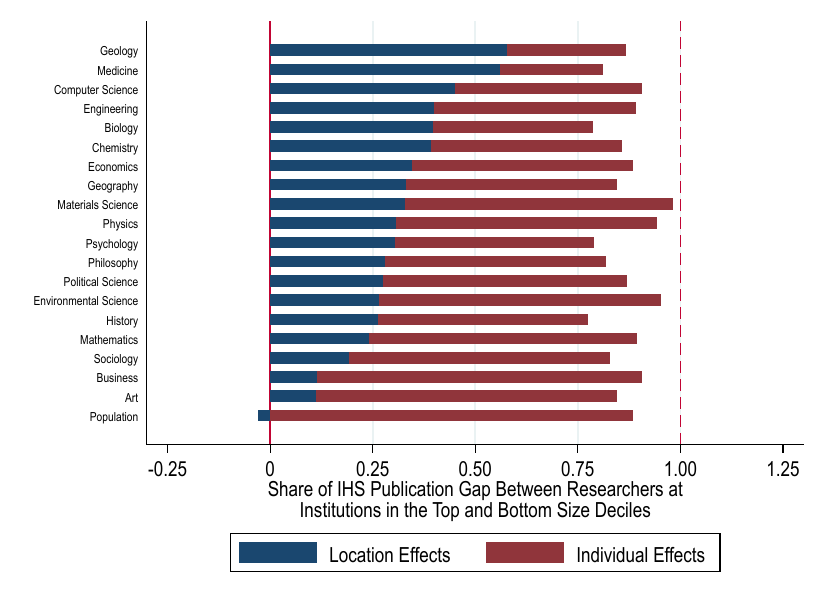}
         \caption{\centering Publications}
         \label{fig:decomp_field_instsize_publications}
     \end{subfigure}
     \qquad 
     \begin{subfigure}[b]{0.47\linewidth}
         \centering
         \includegraphics[width=\linewidth]{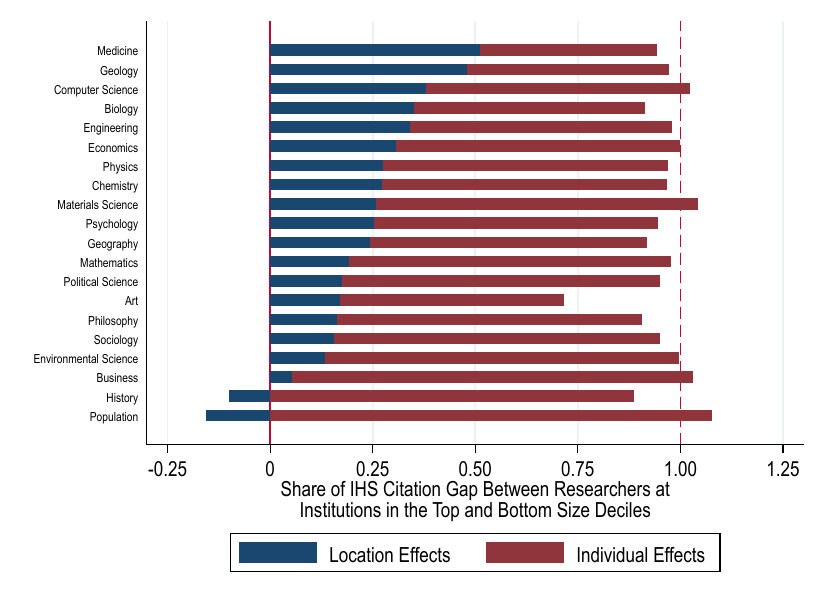}
         \caption{\centering Citations }
         \label{fig:decomp_field_instsize_citations}
     \end{subfigure}
     \\
     \vspace{0.25cm}    
     \begin{subfigure}[b]{0.47\linewidth}
         \centering
         \includegraphics[width=\linewidth]{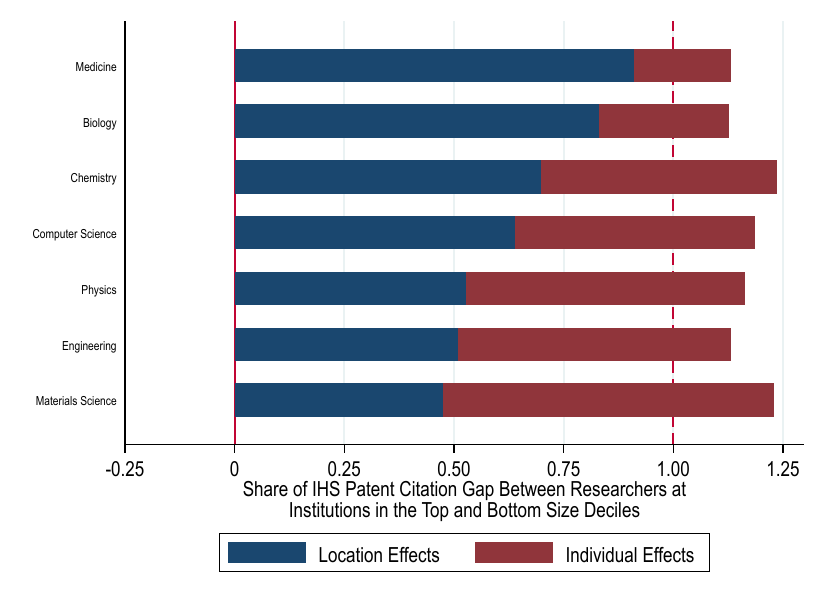}
         \caption{\centering Patent Citations}
         \label{fig:decomp_field_instsize_patents}
     \end{subfigure}
    \vspace{0.25cm}
    \label{fig:decomp_field_instsize}
    \vspace{0.25cm}
\begin{minipage}{\textwidth}
\footnotesize{{\it Notes:} This figure reports field-specific additive decompositions of research outcomes across researchers by institution size during 2010--2015. Panels (a), (b), and (c) report results for IHS publications, citations, and patent citations, respectively. Within each field, institutions are ranked by the average number of active researchers, and the figure compares researchers in the top and bottom deciles. The researcher component includes author fixed effects and field-specific academic-age effects, while the location component is predicted using lagged observable characteristics of the institution--field and external MSA--field research environments. Component shares are calculated as the difference in each component divided by the overall difference in the corresponding outcome. Panel (c) includes only patent-relevant fields, i.e., fields in which more than 30\% of papers receive at least one patent citation.}
\end{minipage}
\end{figure}

\clearpage

\begin{figure}[!h]
    
    \captionsetup{justification=centering}
    \caption{Additive Decomposition of Research Outcomes By Field: \\ Top vs. Bottom MSA Size Deciles} 
    \centering
    \begin{subfigure}[b]{0.47\linewidth}
         \centering
         \includegraphics[width=\linewidth]{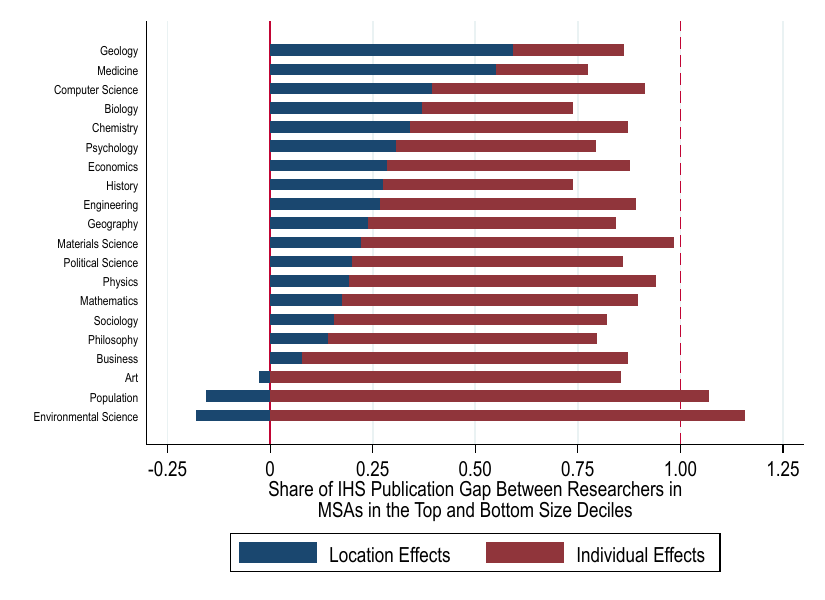}
         \caption{\centering Publications}
         \label{fig:decomp_field_msasize_publications}
     \end{subfigure}
     \qquad 
     \begin{subfigure}[b]{0.47\linewidth}
         \centering
         \includegraphics[width=\linewidth]{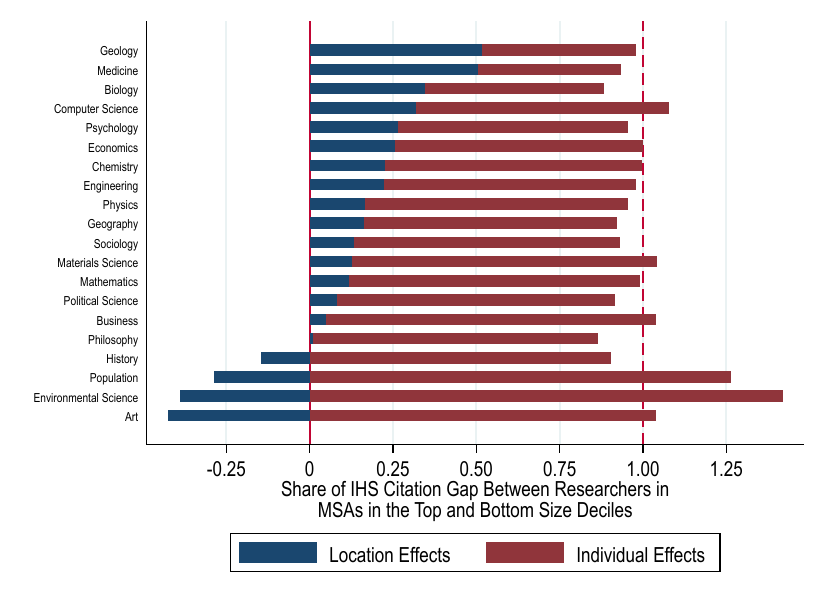}
         \caption{\centering Citations }
         \label{fig:decomp_field_msasize_citations}
     \end{subfigure}
     \\
     \vspace{0.25cm}    
     \begin{subfigure}[b]{0.47\linewidth}
         \centering
         \includegraphics[width=\linewidth]{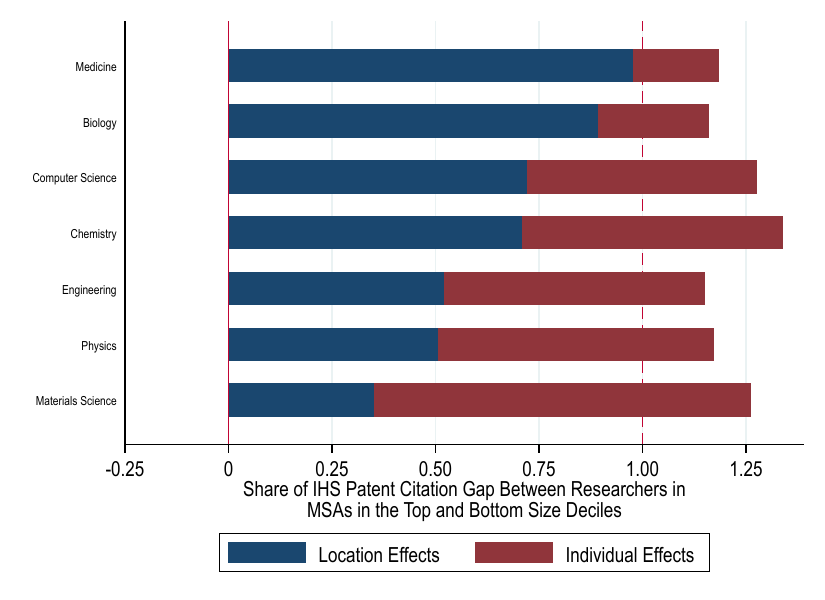}
         \caption{\centering Patent Citations}
         \label{fig:decomp_field_msasize_patents}
     \end{subfigure}
    \vspace{0.25cm}
    \label{fig:decomp_field_msasize}
    \vspace{0.25cm}
\begin{minipage}{\textwidth}
\footnotesize{{\it Notes:} This figure reports field-specific additive decompositions of research outcomes across researchers by MSA size during 2010--2015. Panels (a), (b), and (c) report results for IHS publications, citations, and patent citations, respectively. Within each field, MSAs are ranked by the average number of active researchers, and the figure compares researchers in the top and bottom deciles. The researcher component includes author fixed effects and field-specific academic-age effects, while the location component is predicted using lagged observable characteristics of the institution--field and external MSA--field research environments. Component shares are calculated as the difference in each component divided by the overall difference in the corresponding outcome. Subfigure (c) includes only patent-relevant fields, i.e., fields in which more than 30\% of papers receive at least one patent citation.}
\end{minipage}
\end{figure}

\clearpage

\begin{figure}[!h]
     \centering
    \caption{Leave-One-Out Test for IV Estimates: Top Contributing Subfields} 
     \begin{subfigure}[b]{0.45\linewidth}
         \centering
         \includegraphics[width=\linewidth]{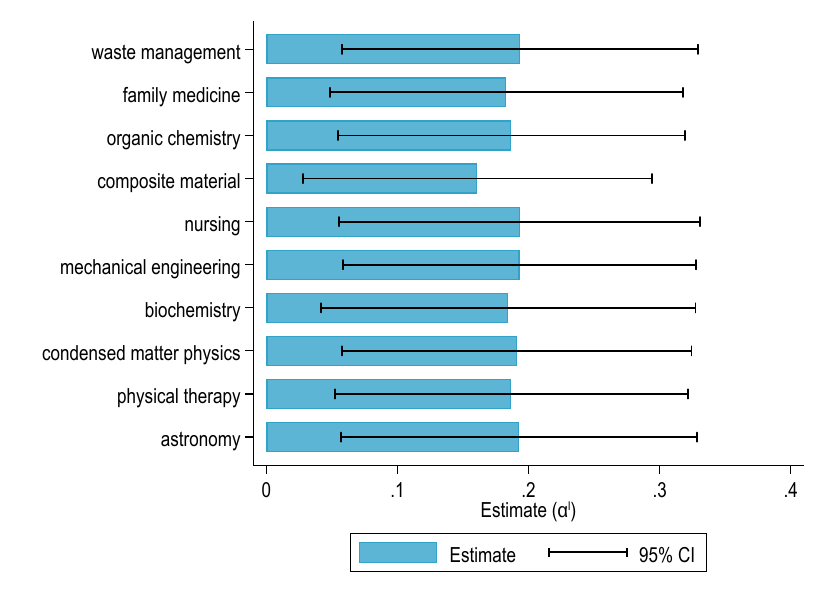}
         \caption{\centering Publications ($\alpha^I$)}
         \label{fig:leave_one_out_metrofips_pub_alphaI}
     \end{subfigure}
     \qquad
    \begin{subfigure}[b]{0.45\linewidth}
         \centering
         \includegraphics[width=\linewidth]{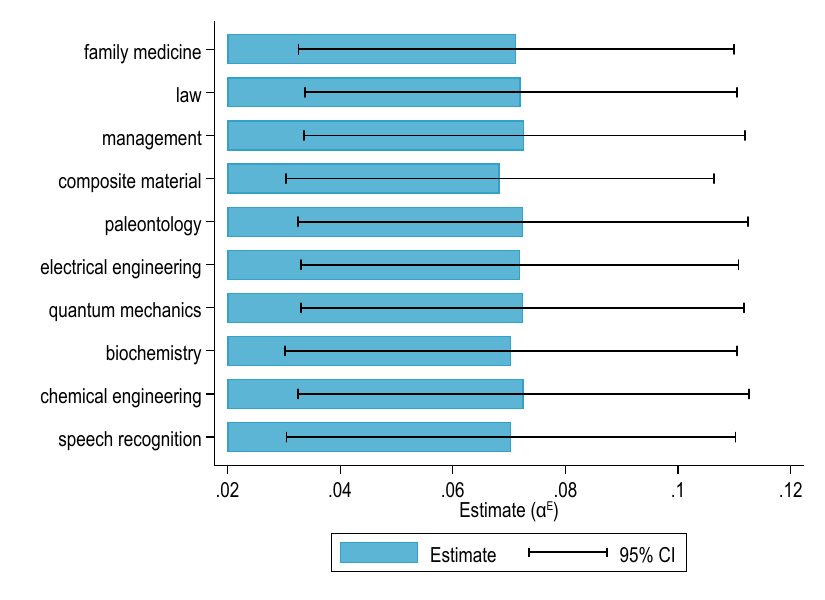}
         \caption{\centering Publications ($\alpha^E$)}
         \label{fig:leave_one_out_metrofips_pub_alphaE}
     \end{subfigure}
     \\
     \vspace{0.25cm}    
     \begin{subfigure}[b]{0.45\linewidth}
         \centering
         \includegraphics[width=\linewidth]{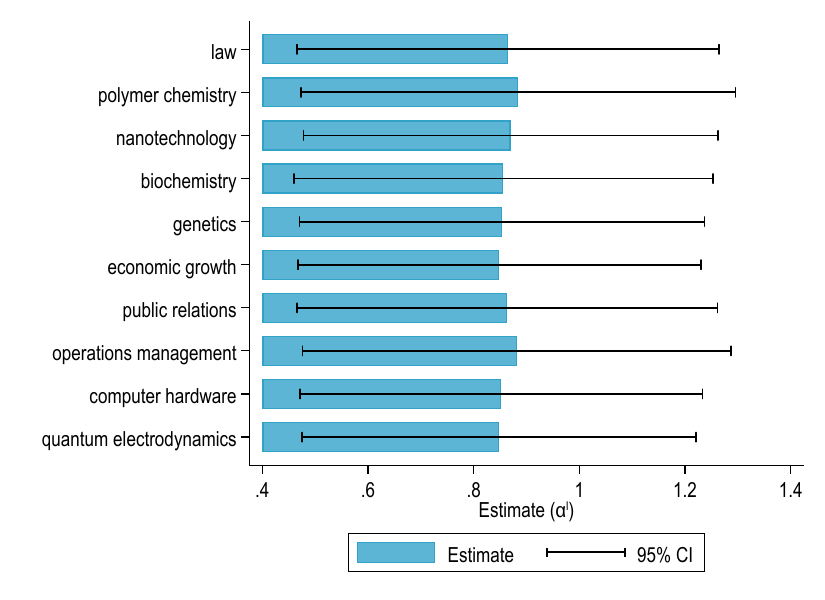}
         \caption{\centering Citations ($\alpha^I$)}
         \label{fig:leave_one_out_metrofips_cit_alphaI}
     \end{subfigure}
     \qquad
     \begin{subfigure}[b]{0.45\linewidth}
         \centering
         \includegraphics[width=\linewidth]{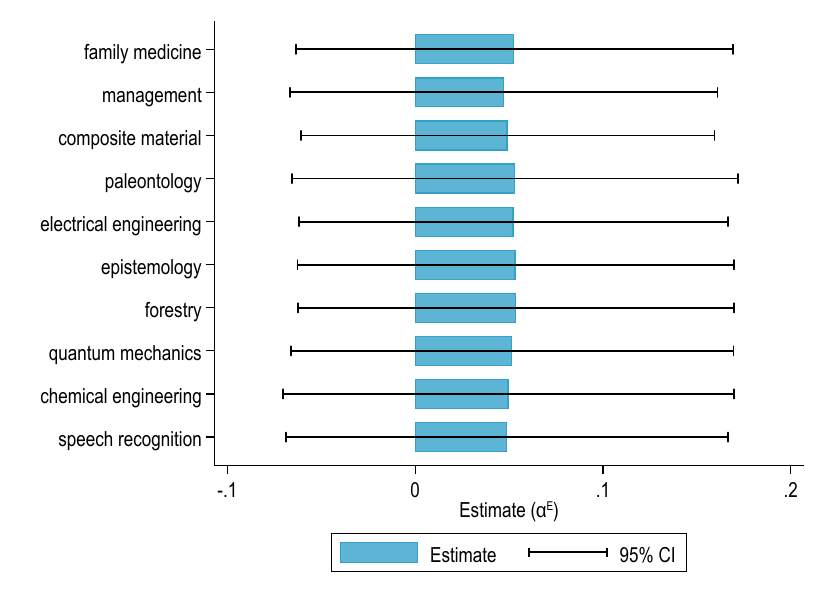}
         \caption{\centering Citations ($\alpha^E$)}
         \label{fig:leave_one_out_metrofips_cit_alphaE}
     \end{subfigure} 
      \\
     \vspace{0.25cm}    
     \begin{subfigure}[b]{0.45\linewidth}
         \centering
         \includegraphics[width=\linewidth]{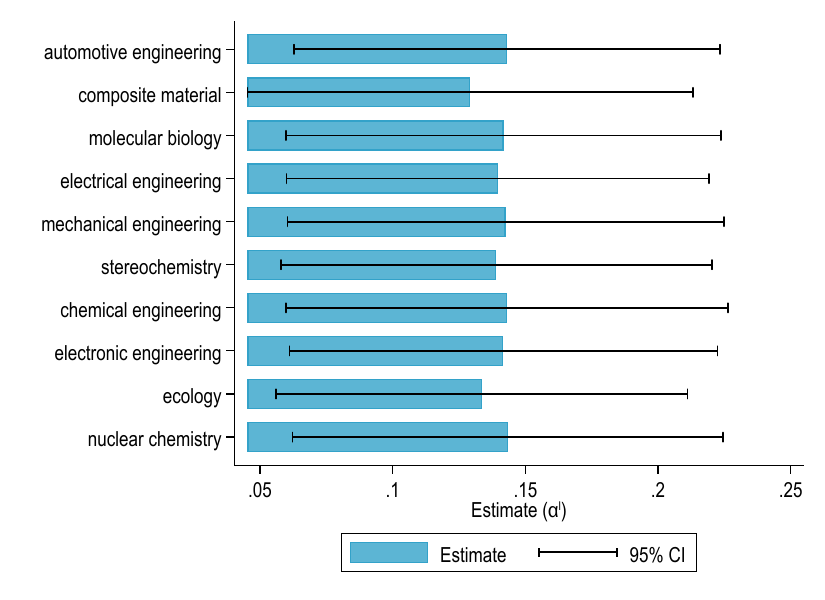}
         \caption{\centering Patent Citations ($\alpha^I$)}
         \label{fig:leave_one_out_metrofips_pcit_alphaI}
     \end{subfigure}
     \qquad
     \begin{subfigure}[b]{0.45\linewidth}
         \centering
         \includegraphics[width=\linewidth]{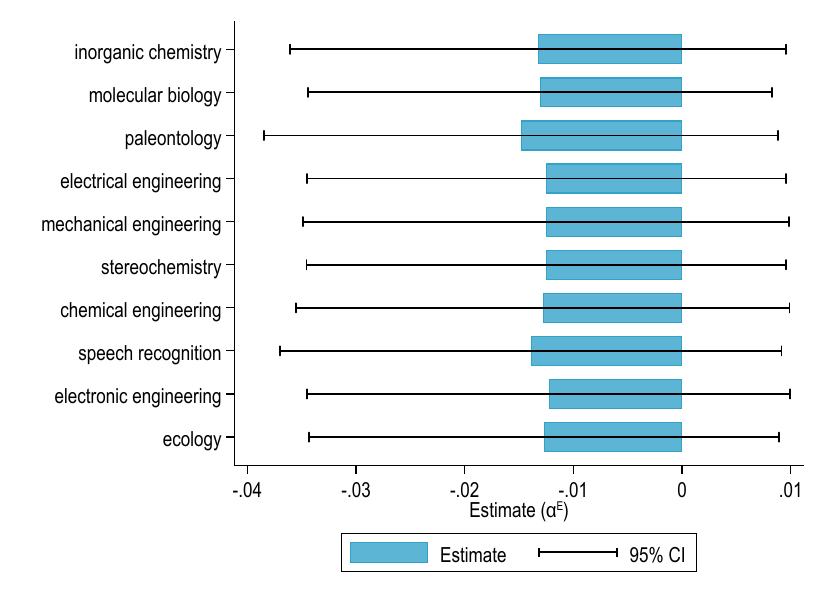}
         \caption{\centering Patent Citations ($\alpha^E$)}
         \label{fig:leave_one_out_metrofips_pcit_alphaE}
     \end{subfigure} 
    \vspace{0.25cm}
\begin{minipage}{\textwidth}
\footnotesize{{\it Notes:} This figure reports leave-one-subfield-out estimates from the preferred IV specification. For each outcome, we reconstruct the internal and external Bartik instruments 364 times, each time omitting one detailed subfield, and display the ten estimates with the largest absolute deviations from the corresponding baseline estimate. Panels (a) and (b) report results for IHS publications, Panels (c) and (d) for IHS citations, and Panels (e) and (f) for IHS patent citations. Panels (a), (c), and (e) report the own-institution elasticity, $\alpha^I$, while Panels (b), (d), and (f) report the external-cluster elasticity, $\alpha^E$. 95\% confidence intervals are displayed based on standard errors clustered at the institution level.}
\label{fig:leave_one_out}
\end{minipage}
\end{figure}

\clearpage

\begin{figure}[p]
    \centering
    \caption{Anderson--Rubin Confidence Sets for IV Estimates}
    \label{fig:ar}

    \begin{subfigure}[b]{0.8\linewidth}
        \centering
        \includegraphics[
    width=\linewidth,
    height=0.5\textheight,
    keepaspectratio
]{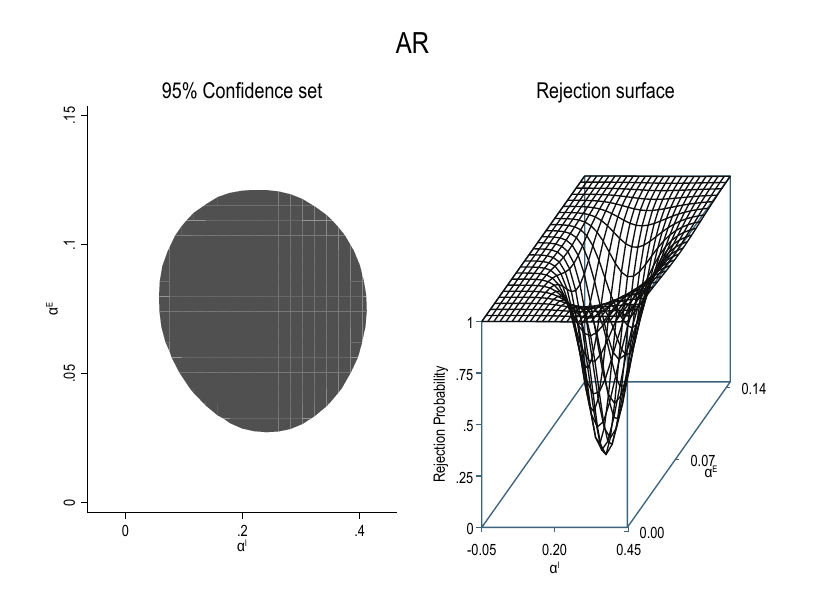}
        \caption{Publications}
        \label{fig:ar_pub}
    \end{subfigure}

    \vspace{1em}

    \begin{subfigure}[b]{0.8\linewidth}
        \centering
        \includegraphics[
    width=\linewidth,
    height=0.5\textheight,
    keepaspectratio
]{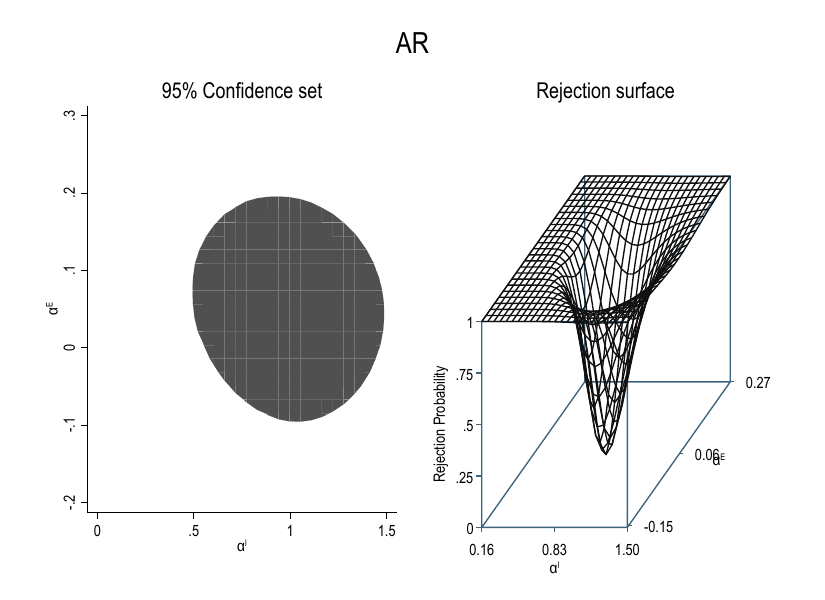}
        \caption{Citations}
        \label{fig:ar_cit}
    \end{subfigure}
    \end{figure}

    \clearpage

    \begin{figure}[p]
    \ContinuedFloat
    \centering
    \caption[]{Anderson--Rubin Confidence Sets for IV Estimates
    (continued)}

    \setcounter{subfigure}{2}

    \begin{subfigure}[b]{0.8\linewidth}
        \centering
        \includegraphics[
    width=\linewidth,
    height=0.5\textheight,
    keepaspectratio
]{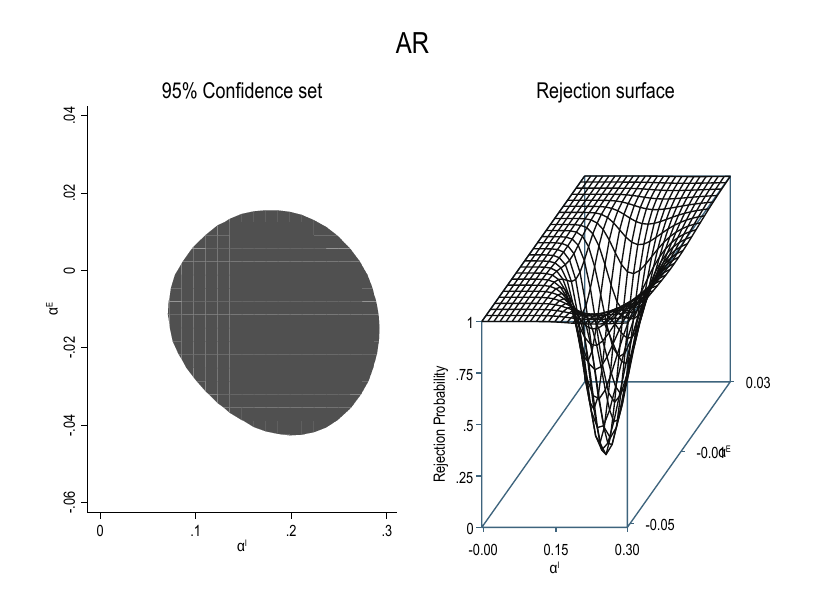}
        \caption{Patent Citations}
        \label{fig:ar_pcit}
    \end{subfigure}

    \vspace{1em}

    \begin{minipage}{\textwidth}
        \footnotesize
        \textit{Notes:}
        Each panel reports the 95\% Anderson–Rubin confidence set for the own-institution and external-cluster elasticities, $(\alpha^{I},\alpha^{E})$, from our preferred IV specification for IHS publications in Panel (a), IHS citations in Panel (b), and IHS patent citations in Panel (c). The shaded region contains parameter pairs that cannot be rejected at the 5\% level, and the accompanying surface plots the rejection probability, $1-p$.
        Anderson–Rubin inference remains valid under weak identification, with standard errors clustered at the institution level. For all three outcomes, the confidence sets lie entirely in the region with $\alpha^I>0$. 
    \end{minipage}
\end{figure}

\clearpage

\begin{table}[h]
\centering
\captionsetup{justification=centering}
\caption{Sample Sizes}
\begin{tabular}{lr}
\toprule
Sample                                                                 & Observations \\ \midrule 
{\it Panel A: Full Sample}                                             &           \\
Researcher-years   (with imputation)                                   & 3,092,028 \\
Researcher-years (without imputation)                                  & 1,191,273 \\
Researchers                                                            & 187,865   \\
Non-movers                                                             & 143,856   \\
Movers                                                                 & 44,009    \\
Movers (moved once)                                                    & 17,482    \\ \addlinespace
{\it Panel B: IV Sample}                                               &           \\
Researcher-years   (with imputation)                                   & 798,548   \\
Researcher-years (without imputation)                                  & 271,602   \\
Researchers                                                            & 35,263    \\
Non-movers                                                             & 29,302    \\
Movers                                                                 & 5,961     \\
Movers (moved once)                                                    & 2,652     \\
\bottomrule  \noalign{\vskip 0.1in}
\multicolumn{2}{l}{%
\begin{minipage}{13.3cm}%
\footnotesize{{\it Note:} This table reports sample sizes under different sample restrictions. Each observation in the data is a researcher-year. ``With imputation'' includes years in which no publication is observed; for these years, we impute zero publications and carry forward the most recent affiliation. ``Without imputation'' restricts the sample to researcher–years with at least one publication. The IV sample includes three periods: 1995, 2005, and 2015; each period pools data over the focal year and the two preceding years.  }
\end{minipage}}%
\end{tabular}
\label{table:sample_size}
\end{table}

\begin{table}[h]
\centering
\captionsetup{justification=centering}
\caption{Top 10 Metropolitan Research Clusters}
\begin{tabular}{lcc}
\toprule
Top 10 Clusters in 2019          & Number of Authors & Number of Publications \\ \midrule
Boston-Cambridge-Newton          & 3812              & 22805                  \\
New York-Newark-Jersey City      & 3536              & 15137                  \\
Washington-Arlington-Alexandria  & 3102              & 13520                  \\
Los Angeles-Long Beach-Anaheim   & 1706              & 8155                   \\
San Francisco-Oakland-Hayward    & 1316              & 7155                   \\
Philadelphia-Camden-Wilmington   & 1451              & 6426                   \\
Chicago-Naperville-Elgin         & 1396              & 6351                   \\
Houston-The Woodlands-Sugar Land & 850               & 5127                   \\
San Jose-Sunnyvale-Santa Clara   & 1268              & 4945                   \\
Seattle-Tacoma-Bellevue          & 862               & 4527                   \\
\bottomrule  \noalign{\vskip 0.1in}
\multicolumn{3}{l}{%
\begin{minipage}{13.2cm}%
\footnotesize{{\it Note:} This table reports the ten largest MSAs in academic research in 2019, ranked by the number of publications authored by researchers affiliated with institutions in each MSA. It also reports the corresponding number of active authors affiliated with institutions in each MSA. }
\end{minipage}}%
\end{tabular}
\label{table:cluster}
\end{table}

\begin{table}[h]
\centering
\captionsetup{justification=centering}
\caption{OLS Estimates of Agglomeration Elasticities: Full Sample}
\begin{tabular}{lcccc}
\toprule
                               & (1)       & (2)       & (3)       & (4)       \\ \midrule
\multicolumn{5}{l}{\textit{Panel A: Publications}}                             \\
Ln(Institution Size)           & 0.134***  & 0.131***  & 0.258***  & 0.196***  \\
                               & (0.006)   & (0.006)   & (0.007)   & (0.005)   \\
Ln(External Size)              & -0.039*** & 0.001     & -0.007*** & 0.016***  \\
                               & (0.004)   & (0.002)   & (0.002)   & (0.001)   \\ \addlinespace
\multicolumn{5}{l}{\textit{Panel B: Citations}}                                \\
Ln(Institution Size)           & 0.351***  & 0.283***  & 0.547***  & 0.412***  \\
                               & (0.015)   & (0.011)   & (0.015)   & (0.011)   \\
Ln(External Size)              & -0.087*** & 0.011**   & -0.017*** & 0.036***  \\
                               & (0.009)   & (0.005)   & (0.004)   & (0.003)   \\ \addlinespace
\multicolumn{5}{l}{\textit{Panel C: Patent Citations}}                         \\
Ln(Institution Size)           & 0.044***  & 0.027***  & 0.043***  & 0.031***  \\
                               & (0.003)   & (0.001)   & (0.002)   & (0.003)   \\
Ln(External Size)              & -0.004*** & 0.002***  & -0.007*** & -0.002*** \\
                               & (0.002)   & (0.001)   & (0.001)   & (0.001)   \\  \midrule
Observations                   & 3,092,028 & 3,088,219 & 3,086,886 & 2,755,515 \\
Field $\times$ Year FE         & Yes       & Yes       & Yes       & Yes       \\
Author FE                      & No        & Yes       & Yes       & Yes       \\
Academic Age $\times$ Field FE & No        & Yes       & Yes       & Yes       \\
Affiliation $\times$ Field FE  & No        & No        & Yes       & Yes       \\
MSA $\times$ Year FE           & No        & No        & No        & Yes       \\
\bottomrule  \noalign{\vskip 0.1in}
\multicolumn{5}{l}{%
\begin{minipage}{12cm}%
\footnotesize{\textit{Note:} This table reports OLS estimates of Equation \ref{eq:agg_reg} using the full annual author-year panel from 1970 to 2015. The dependent variable is the inverse hyperbolic sine (IHS) of publications in Panel A, IHS citations in Panel B, and IHS patent citations in Panel C. Institution size is the number of active researchers in the author's field and institution in a given year. External size is the number of active researchers in the same field and MSA, excluding the home institution. Each column corresponds to a separate regression with the fixed effects listed at the bottom of the table. Academic age is measured as years since the author's first publication. Standard errors are in parentheses and clustered at the institution level: *** $p< 0.01$, ** $p< 0.05$, * $p< 0.1$.
}
\end{minipage}}%
\end{tabular}
\label{table:result_OLS_full}
\end{table}

\begin{table}[h]
\centering
\captionsetup{justification=centering}
\caption{Pre-Trend Tests for Bartik-Predicted Cluster Growth}
\begin{tabular}{lcccccc}
\toprule
                         & \multicolumn{6}{c}{1990--1995 Change in IHS Output}                                 \\ \cmidrule(lr){2-7}
                         & Pubs         & Citations  & Patent Cites & Pubs         & Citations  & Patent Cites \\
                         & (1)          & (2)        & (3)          & (4)          & (5)        & (6)          \\ \midrule
Pred. Institution Growth & -0.044       & 0.044      & -0.173**     & -0.017       & -0.006     & -0.054**     \\
                         & (0.044)      & (0.072)    & (0.075)      & (0.012)      & (0.020)    & (0.023)      \\
Pred. External Growth    & 0.005        & -0.035     & 0.059***     & 0.003        & -0.014     & 0.023***     \\
                         & (0.012)      & (0.025)    & (0.019)      & (0.004)      & (0.009)    & (0.008)      \\
Observations             & 9,781        & 9,781      & 9,781        & 9,781        & 9,781      & 9,781        \\
R-squared                & 0.052        & 0.033      & 0.119        & 0.052        & 0.033      & 0.119        \\
Joint p-value            & 0.586        & 0.329      & 0.002        & 0.324        & 0.262      & 0.003        \\
Predicted Growth Window  & 95--05       & 95--05     & 95--05       & 95--15       & 95--15     & 95--15       \\ \midrule
                         & \multicolumn{6}{c}{1985--1995 Change in IHS Output}                                 \\ \cmidrule(lr){2-7}
                         & Pubs         & Citations  & Patent Cites & Pubs         & Citations  & Patent Cites \\
                         & (1)          & (2)        & (3)          & (4)          & (5)        & (6)          \\ \midrule
Pred. Institution Growth & -0.085       & -0.018     & -0.366***    & -0.041**     & -0.032     & -0.159***    \\
                         & (0.064)      & (0.107)    & (0.103)      & (0.018)      & (0.029)    & (0.034)      \\
Pred. External Growth    & 0.040**      & 0.024      & 0.170***     & 0.015***     & 0.010      & 0.066***     \\
                         & (0.016)      & (0.030)    & (0.032)      & (0.006)      & (0.011)    & (0.012)      \\
Observations             & 8,657        & 8,657      & 8,657        & 8,657        & 8,657      & 8,657        \\
R-squared                & 0.161        & 0.097      & 0.346        & 0.162        & 0.097      & 0.348        \\
Joint p-value            & 0.024        & 0.733      & 0.000        & 0.004        & 0.381      & 0.000        \\
Predicted Growth Window  & 95--05       & 95--05     & 95--05       & 95--15       & 95--15     & 95--15       \\
\bottomrule  \noalign{\vskip 0.1in}
\multicolumn{7}{l}{%
\begin{minipage}{15.7cm}%
\footnotesize{Note: This table reports estimates of $\pi^I$ and $\pi^E$ from Equation \ref{eq:pretrend}, which relate pre-base-period growth in IHS research output to Bartik-predicted growth in home-institution and external-cluster size. The dependent variable is the change in IHS research output between 1990 and 1995 in the upper panel and between 1985 and 1995 in the lower panel. In Columns 1–3, the key regressors are Bartik-predicted growth in home-institution and external-cluster size between 1995 and 2005. In Columns 4–6, the key regressors are the corresponding predicted growth measures between 1995 and 2015. All regressions include field fixed effects. Standard errors are in parentheses and clustered at the institution level: *** $p< 0.01$, ** $p< 0.05$, * $p< 0.1$. }
\end{minipage}}%
\end{tabular}
\label{table:result_pretrend}
\end{table}

\begin{table}[h]
\centering
\captionsetup{justification=centering}
\caption{IV Estimates of Agglomeration Elasticities by Cluster-Size Tertile}
\begin{tabular}{lccccc}
\toprule
                          & Publications & Citations & Patent Cits \\
                          & (1)          & (2)       & (3)         \\ \midrule
Institution Size: Small   & 0.205**      & 0.960***  & 0.193***    \\
                          & (0.091)      & (0.226)   & (0.047)     \\
Institution Size: Medium  & 0.157*       & 0.801***  & 0.156***    \\
                          & (0.083)      & (0.213)   & (0.043)     \\
Institution Size: Large   & 0.225***     & 0.793***  & 0.109***    \\
                          & (0.067)      & (0.169)   & (0.036)     \\
External Size: Small      & 0.071***     & 0.035     & -0.011      \\
                          & (0.027)      & (0.058)   & (0.010)     \\
External Size: Medium     & 0.056***     & 0.051     & -0.009      \\
                          & (0.020)      & (0.054)   & (0.010)     \\
External Size: Large      & 0.090***     & 0.081     & -0.004      \\
                          & (0.025)      & (0.071)   & (0.013)     \\ \midrule
Observations              & 778,733      & 778,733   & 778,733     \\
Baseline FEs              & Yes          & Yes       & Yes         \\
Author/Coauthor Controls  & Yes          & Yes       & Yes         \\
Kleibergen–Paap Wald F    & 5.97         & 5.97      & 5.97        \\
Equal Institution p-value & 0.057        & 0.053     & 0.000       \\
Equal External p-value    & 0.156        & 0.703     & 0.809       \\
\bottomrule  \noalign{\vskip 0.1in}
\multicolumn{4}{l}{%
\begin{minipage}{10.5cm}%
\footnotesize{
\textit{Note:} This table reports IV estimates of Equation \ref{eq:heter} using the IV sample. The sample consists of author-year observations in three periods: 1995, 2005, and 2015, where each period pools the focal year and the two preceding years. The dependent variable is IHS publications in Column 1, IHS citations in Column 2, and IHS patent citations in Column 3. The key regressors are log own-institution size and log external-cluster size interacted with indicators for their respective 1995 size tertiles. These interactions are instrumented with the corresponding Bartik instruments interacted with the same tertile indicators. All columns include the baseline fixed effects from Equation \ref{eq:agg_reg} and the author and coauthor direct-exposure controls. The final two rows report $p$-values for equality of the three own-institution coefficients and the three external-cluster coefficients, respectively. Standard errors are in parentheses and clustered at the institution level: *** $p< 0.01$, ** $p< 0.05$, * $p< 0.1$.
}
\end{minipage}}%
\end{tabular}
\label{table:result_IV_heter}
\end{table}

\begin{table}[h]
\centering
\captionsetup{justification=centering}
\caption{IV Estimates of Agglomeration Elasticities \\ Using Alternative Cluster Measures}
\begin{tabular}{lcccc}
\toprule
Cluster Definition      & \multicolumn{2}{c}{MSA}    & \multicolumn{2}{c}{25-km Radius} \\ \cmidrule(lr){2-3} \cmidrule(lr){4-5}
Size Measurement        & Author Count & Paper Count & Author Count    & Paper Count    \\
                        & (1)          & (2)         & (3)             & (4)            \\ \midrule
\multicolumn{5}{l}{{\it Panel A: Publications}}                                         \\
Ln(Institution Size)    & 0.193***     & 0.101       & 0.221***        & 0.133          \\
                        & (0.069)      & (0.097)     & (0.082)         & (0.110)        \\
Ln(External Size)       & 0.081***     & 0.049***    & -0.139          & -0.062         \\
                        & (0.019)      & (0.014)     & (0.119)         & (0.053)        \\ \addlinespace
\multicolumn{5}{l}{{\it Panel B: Citations}}                                            \\
Ln(Institution Size)    & 0.795***     & 0.898***    & 0.806***        & 0.870***       \\
                        & (0.183)      & (0.281)     & (0.214)         & (0.320)        \\
Ln(External Size)       & 0.084        & 0.048       & -0.199          & -0.162         \\
                        & (0.053)      & (0.053)     & (0.279)         & (0.167)        \\ \addlinespace
\multicolumn{5}{l}{{\it Panel C: Patent Citations}}                                     \\
Ln(Institution Size)    & 0.140***     & 0.336***    & 0.138***        & 0.335***       \\
                        & (0.041)      & (0.114)     & (0.052)         & (0.127)        \\
Ln(External Size)       & -0.008       & 0.007       & 0.076           & 0.000          \\
                        & (0.011)      & (0.020)     & (0.082)         & (0.057)        \\
Observations            & 794,658      & 794,658     & 692,763         & 692,763        \\
Kleibergen-Paap Wald F  & 15.33        & 4.718       & 3.282           & 3.926          \\
\bottomrule  \noalign{\vskip 0.1in}
\multicolumn{5}{l}{%
\begin{minipage}{13.5cm}%
\footnotesize{Note: This table reports IV estimates of Equation \ref{eq:agg_reg} using the IV sample. The dependent variable is IHS publications in Panel A, IHS citations in Panel B, and IHS patent citations in Panel C. The sample consists of author-year observations in three periods: 1995, 2005, and 2015, where each period pools the focal year and the two preceding years. In Columns 1–2, the external cluster is defined as institutions in the same field and MSA as the home institution; in Columns 3–4, it is defined as institutions in the same field within 25 km of the home institution. Cluster size is measured by the number of active researchers in Columns 1 and 3 and by the number of papers published in Columns 2 and 4. The size measures are instrumented using the corresponding Bartik instruments. All specifications include author fixed effects, academic-age-by-field fixed effects, field-by-year fixed effects, affiliation-by-field fixed effects, and the author and coauthor direct-exposure controls. The Kleibergen–Paap rk Wald F statistic is reported for each specification. Standard errors are in parentheses and clustered at the institution level: *** $p< 0.01$, ** $p< 0.05$, * $p< 0.1$. }
\end{minipage}}%
\end{tabular}
\label{table:result_IV_cluster}
\end{table}

\begin{table}[h]
\centering
\captionsetup{justification=centering}
\caption{IV Estimates of Agglomeration Elasticities \\ Under Explicit Extensive-Margin Calibrations}
\begin{tabular}{
    >{\raggedright\arraybackslash}p{4cm}
    *{3}{>{\centering\arraybackslash}p{1.8cm}}
}
\toprule
                       & (1)      & (2)      & (3)      \\ \midrule
\multicolumn{4}{l}{{\it Panel A: Publications}}         \\
Ln(Institution Size)   & 0.199*** & 0.112**  & 0.236*** \\
                       & (0.070)  & (0.054)  & (0.079)  \\
Ln(External Size)      & 0.074*** & 0.060*** & 0.082*** \\
                       & (0.020)  & (0.015)  & (0.023)  \\ \addlinespace
\multicolumn{4}{l}{{\it Panel B: Citations}}            \\
Ln(Institution Size)   & 0.828*** & 0.737*** & 0.879*** \\
                       & (0.190)  & (0.171)  & (0.201)  \\
Ln(External Size)      & 0.056    & 0.043    & 0.062    \\
                       & (0.060)  & (0.054)  & (0.063)  \\ \addlinespace
\multicolumn{4}{l}{{\it Panel C: Patent Citations}}     \\
Ln(Institution Size)   & 0.148*** & 0.117*** & 0.163*** \\
                       & (0.043)  & (0.034)  & (0.047)  \\
Ln(External Size)      & -0.011   & -0.007   & -0.013   \\
                       & (0.012)  & (0.009)  & (0.013)  \\
                       &          &          &          \\
Extensive Margin Value & Baseline & 0.1      & 1        \\
\bottomrule  \noalign{\vskip 0.1in}
\multicolumn{4}{l}{%
\begin{minipage}{10.8cm}%
\footnotesize{Note: This table reports IV estimates of Equation \ref{eq:agg_reg} using the IV sample. The sample consists of author-year observations in three periods: 1995, 2005, and 2015, where each period pools the focal year and the two preceding years. Column 1 uses the baseline IHS transformation. Columns 2 and 3 use $m(F)=\log(F)$ for $F>0$ and $m(0)=-x$, with $x=0.1$ and $x=1$, respectively. All specifications include author fixed effects, academic-age-by-field fixed effects, field-by-year fixed effects, affiliation-by-field fixed effects, and the author and coauthor direct-exposure controls. Standard errors are in parentheses and clustered at the institution level: *** $p< 0.01$, ** $p< 0.05$, * $p< 0.1$. }
\end{minipage}}%
\end{tabular}
\label{table:result_IV_extensive}
\end{table}

\begin{table}[h]
\centering
\captionsetup{justification=centering}
\caption{IV Estimates of Agglomeration Elasticities \\ Using Alternative Log Transformations of Cluster Size}
\begin{tabular}{
    >{\raggedright\arraybackslash}p{4cm}
    *{3}{>{\centering\arraybackslash}p{1.8cm}}
}
\toprule 
                                   & (1)      & (2)      & (3)      \\ \midrule
\textit{Panel A: Publications}     &          &          &          \\
Ln(Institution Size)               & 0.199*** & 0.537*** & 1.031**  \\
                                   & (0.070)  & (0.188)  & (0.467)  \\
Ln(External Size)                  & 0.074*** & 0.043    & -0.023   \\
                                   & (0.020)  & (0.030)  & (0.066)  \\ \addlinespace
\textit{Panel B: Citations}        &          &          &          \\
Ln(Institution Size)               & 0.828*** & 1.932*** & 3.527**  \\
                                   & (0.190)  & (0.599)  & (1.583)  \\
Ln(External Size)                  & 0.056    & -0.071   & -0.313   \\
                                   & (0.060)  & (0.108)  & (0.235)  \\ \addlinespace
\textit{Panel C: Patent Citations} &          &          &          \\
Ln(Institution Size)               & 0.148*** & 0.320*** & 0.566*   \\
                                   & (0.043)  & (0.121)  & (0.290)  \\
Ln(External Size)                  & -0.011   & -0.034*  & -0.075*  \\
                                   & (0.012)  & (0.020)  & (0.042)  \\
                                   &          &          &          \\
Transformations                    & Ln(N+1)  & Ln(N+5)  & Ln(N+10) \\
Kleibergen–Paap Wald F             & 14.9     & 5.34     & 2.34     \\
\bottomrule  \noalign{\vskip 0.1in}
\multicolumn{4}{l}{%
\begin{minipage}{10.8cm}%
\footnotesize{Note: This table reports IV estimates of Equation \ref{eq:agg_reg} using different transformations of own-institution and external-cluster size. Columns 1–3 use ln(N+1), ln(N+5), and ln(N+10), respectively.
The dependent variable is IHS publications in Panel A, IHS citations in Panel B, and IHS patent citations in Panel C. The sample consists of author-year observations in three periods: 1995, 2005, and 2015, where each period pools the focal year and the two preceding years. All specifications include author fixed effects, academic-age-by-field fixed effects, field-by-year fixed effects, affiliation-by-field fixed effects, and the author and coauthor direct-exposure controls. The Kleibergen–Paap rk Wald F statistic is reported for each specification. Standard errors are in parentheses and clustered at the institution level: *** $p< 0.01$, ** $p< 0.05$, * $p< 0.1$. }
\end{minipage}}%
\end{tabular}
\label{table:result_IV_zero}
\end{table}

\end{document}